\documentclass[11pt,reqno,oneside]{amsart}
\usepackage{geometry}
\usepackage[parfill]{parskip}
\usepackage{graphicx}
\usepackage{amssymb}
\usepackage{epstopdf}
\usepackage{hyperref}
\usepackage{bm}
\usepackage{amsmath,amsfonts,amsthm,mathrsfs,amssymb,cite}
\usepackage[usenames]{color}
\usepackage{ulem} \usepackage{booktabs} 
\usepackage{array} 
\usepackage{paralist} 
\usepackage{verbatim} 
\usepackage{subfig} 

\usepackage{tabularx}
\usepackage{fancyhdr}
\newtheorem{thm}{Theorem}[section]

\newtheorem{lem}{Lemma}[section]

\theoremstyle{definition}

\theoremstyle{remark}

\newtheorem{rem}{Remark}[section]
\numberwithin{equation}{section}
\DeclareMathAlphabet{\itbf}{OML}{cmm}{b}{it}

\numberwithin{equation}{section}
\numberwithin{equation}{section}
\newcounter{saveeqn}

\begin{document}
\title[Surface Polariton Resonance in Elastic Nanorods]{On surface polariton resonance and its curvature concentration effects from 3D elastic nanorods}

\author{Youjun Deng}
\address{School of Mathematics and Statistics, HNP-LAMA, Central South University, Changsha 410083, China}
\email{youjundeng@csu.edu.cn, dengyijun\_001@163.com}

\author{Hongyu Liu}
\address{Department of Mathematics, City University of Hong Kong, Hong Kong SAR, China}
\email{hongyliu@cityu.edu.hk}	

\author{Wanjing Tang \textsuperscript{*}}
\address{School of Mathematics and Statistics, Central South University, Changsha 410083, China}
\email{wanjingtang@csu.edu.cn}
\thanks{\textsuperscript{*}Corresponding author}

\author{Guang-Hui Zheng}
\address{School of Mathematics, Hunan University, Changsha 410082, China}
\email{zhenggh2012@hnu.edu.cn; zhgh1980@163.com}
\maketitle

\begin{abstract}
This paper investigates surface polariton resonance (SPR) in three-dimensional elastic metamaterials with nanorod geometry. The primary motivation is to surpass the physical limitations imposed by the quasi-static approximation for SPRs through anisotropic geometric design. The analysis boils down to analyzing the spectral properties of the matrix-valued elastic Neumann-Poincar\'e (NP) operator defined on the nanorod boundary. We develop novel analytical techniques and conduct a rigorous asymptotic analysis of elastic layer potential operators, specifically adapted for highly anisotropic structures. Within this framework, we derive precise asymptotic formulas for the scattered field in the quasi-static regime. A thorough examination of these expressions yields explicit resonance conditions that intricately link three fundamental parameters: elastic material parameters, wave frequency, and nanorod geometry. Furthermore, we characterize the intrinsic relationship between these parameters and the associated energy blow-up rate of the resonant field. This analysis explicitly establishes a sharp curvature concentration effect at the nanorod extremities, where field enhancement is locally maximized. Our work provides a rigorous theoretical foundation for harnessing elastic SPRs through anisotropic geometric engineering, with implications for sensing, wave focusing, and metamaterial applications.

\vspace{0.3cm}
\noindent{$\mathbf{Keywords:}$}~ Surface polariton resonance, quasi-static approximation, nanorod, anisotropic geometry, curvature concentration, elastic Neumann-Poincar\'e operator, spectral analysis, asymptotic analysis


\end{abstract}

\section{Introduction}

The unique ability of surface waves to confine energy and generate intense local field enhancements has positioned plasmonics at the forefront of modern photonics. Surface Plasmon Resonance (SPR), the resonant oscillation of conduction electrons at a metal-dielectric interface, is the foundational phenomenon enabling this capability. By confining light to subwavelength scales, SPR underpins revolutionary applications across diverse fields, including ultrasensitive biomedical imaging and therapy, nanoscale lasing, and advanced optical cloaking \cite{bi1, HAYDPM, bi4,bi5, bi7, bi8, bi10,DT,bi17,FD23, bi20, bi22, bi23, bi24, bi25, bi29, bi33, bi34, bi35}. Consequently, the design and analysis of plasmonic devices composed of metallic nanostructures remain a central research focus in physics and materials science. By analogy, the study of Surface Polariton Resonance (SPR) in elastic systems---where mechanical waves interact with nanostructures---has garnered significant interest for its potential to manipulate solid waves at the nanoscale, offering parallel promise for novel acoustic and mechanical metamaterials \cite{bi13, bi25, ANDOJK, bi12, bi14, ANDOKH}.

A fundamental question arises regarding the intrinsic nature of SPRs: Are they exclusively subwavelength phenomena, or is their origin more deeply tied to geometry? Resolving this question is critical, as it could reveal pathways to overcome size limitations and unlock new application regimes. Numerical evidence suggests that SPRs can persist beyond the conventional quasi-static regime, localized specifically at regions of high boundary curvature \cite{bi26}. Rigorous theoretical justification for this observed curvature concentration effect, however, remains a significant challenge. Recent analytical work for scalar wave systems (e.g., governed by the Laplace or Helmholtz equations) has made pivotal progress \cite{bi2, bi001, DYLHZG, DYLHZGH}. These studies demonstrate that SPR intensity concentrates at boundary points of high curvature, with its magnitude proportional to the curvature ratio between distinct boundary regions. As the operating wavelength decreases---equivalent to increasing the structural size relative to the wavelength---the resonance generally weakens. Crucially, however, it vanishes first from low-curvature areas, while persisting longer at high-curvature points. This provides a rigorous basis for prior numerical observations and establishes a profound geometric perspective: the essential origin of strong SPR may lie not strictly in the subwavelength scale, but in the presence of high local curvature. This insight carries substantial practical implications, suggesting that technologies leveraging SPRs may transcend traditional size constraints by strategically employing high-curvature geometric features.

A critical gap persists in the existing literature. The analyses in \cite{bi2, bi001, DYLHZG, DYLHZGH, bi26} are confined to {scalar} wave models, where the mathematical problem reduces to studying the spectral properties of the scalar Neumann-Poincar\'e (NP) operator. Extending this geometric understanding to {vectorial} elastic waves is both theoretically demanding and practically vital, as it requires confronting the complexities of matrix-valued elastic NP operators. This article aims to bridge this gap by providing a rigorous analysis of surface polariton resonance and its curvature concentration effects in three-dimensional elastic nanorods.

We develop novel analytical techniques and perform a detailed asymptotic analysis of elastic layer potential operators tailored for highly anisotropic geometries. Within this framework, we derive precise asymptotic formulas for the scattered field in the quasi-static regime. A thorough analysis of these expressions yields explicit resonance conditions that intricately link three fundamental parameters: elastic material constants, wave frequency, and nanorod geometry. Furthermore, we characterize the intrinsic relationship between these parameters and the associated energy blow-up rate of the resonant field. This analysis definitively establishes a sharp curvature concentration effect at the nanorod extremities, where field enhancement is locally maximized.

We focus our primary analysis on straight nanorods, we refer to \cite{bi15,bi16} for its application to warm hole cloaking; the extension to curved geometries, while technically more involved, follow from similar arguments. As our subsequent discussion will reveal, the analysis for the straight rod---already highly technical due to the intricacies of the matrix-valued elastic NP operator combined with anisotropic geometry---sufficiently captures the essential physical insight: that high curvature is a primary driver of intense elastic field concentration.

The rest of the paper is organized as follows. In Section 2, we present the mathematical setup of our study and state the main results. The asymptotic analysis for layer potential operators with respect to nanorod shape geometry is given in Section 3. Finally, we show the proof to our main results in Section 4.

\section{Mathematical setup and statement of main results}

\subsection{Geometric setup and elastic system}

We consider $D$ to be a straight nanorod. Let $\Gamma_{0}$ be a straight line on the x-axis in $\mathbb{R}^{3}$. $P=(-L/2, 0, 0)^T$ and $Q=(L/2, 0, 0)^T$ are the two endpoints of $\Gamma_0$. We start with a thin cylindrical structure $D^f$ defined by
\begin{equation}\label{Df}
D^f:=\left\{\mathbf{x} \mid x_2^2+x_3^2<\delta^2,-L / 2\leq x_1\leq L / 2,-\delta<x_3<\delta\right\},
\end{equation}
the lower and upper half boundaries of $D^f$ are denoted by $\Gamma_{1}$ and $\Gamma_2$ respectively, namely
\begin{equation}\label{T12}
\begin{gathered}
\Gamma_1=\left\{\mathbf{x} \mid x_2^2+x_3^2=\delta^2,-L / 2\leq x_1\leq L / 2,-\delta\leq x_3\leq 0 \right\}, \\
\Gamma_2=\left\{\mathbf{x} \mid x_2^2+x_3^2=\delta^2,-L / 2\leq x_1\leq L / 2, 0\leq x_3\leq \delta \right\} .
\end{gathered}
\end{equation}
 Let $D^a$ and $D^b$ be two half balls with radius $\delta$ and centered at P and Q, namely
 \begin{equation}\label{Dab}
\begin{aligned}
D^a & :=\left\{\mathbf{x} \mid\left(x_1+L / 2\right)^2+x_2^2+x_3^2<\delta^2, x_1<-L / 2\right\}, \\
D^b & :=\left\{\mathbf{x} \mid\left(x_1-L / 2\right)^2+x_2^2+x_3^2<\delta^2, x_1>L / 2\right\} .
\end{aligned}
\end{equation}
 Here and in what follows, $\mathbf{x}=(x_{1}, x_2, x_3)^T$ and $\mathbf{y}=(y_{1}, y_2, y_3)^T$, we also define $\mathbf{e}_j$, $j=1,2,3$ the orthnormal basis in $\mathbb{R}^3$ with $j$-th entry being one. Furthermore, we define $$S^a=\left\{\mathbf{x}|\left(x_1+L/2\right)^2+x_2^2+x_3^2=\delta^2, x_1<-L/2\right\},$$ $$S^b=\left\{\mathbf{x}|\left(x_1-L/2\right)^2+x_2^2+x_3^2=\delta^2, x_1>L/2\right\},$$
 $S^c=S^a\cup S^b$ and $S^f=\Gamma_1\cup\Gamma_2$. Then the rod shape inclusion $D$ is denoted by $D=D^a\cup D^f\cup D^b$ and the boundary $\partial D$ is given by $\partial D=S^a\cup S^f\cup S^b$. In the definition of $D$, we take $\delta=1$, and can immediately obtain the definition of $D_1$, we also use $S_1^f$ to denote the lateral part of $\partial D_1$. It is not difficult to see that $\partial D$ is of class $C^{1, \alpha}$ for some $\alpha\in \mathbb{R}_+$. For any $\mathbf{x}\in S^f$ or $\mathbf{y}\in S^f$, we use $\mathbf{z}_{\mathbf{x}}$ or $\mathbf{z}_{\mathbf{y}}$ to denote the orthogonal projection of $\mathbf{x}$ or $\mathbf{y}$ on $\Gamma_0$, respectively. In addition, by introducing the cylindrical coordinate, for any $\mathbf{x}$, $\mathbf{y}$ $\in S^f$, they can be represented by $\mathbf{x}=\mathbf{z}_{\mathbf{x}}+\delta\boldsymbol{\nu}_{\mathbf{x}}$ and $\mathbf{y}=\mathbf{z}_{\mathbf{y}}+\delta\boldsymbol{\nu}_{\mathbf{y}}$, respectively, where $\mathbf{z}_{\mathbf{x}}=(x_1, 0, 0)^T$, $\mathbf{z}_{\mathbf{y}}=(y_1, 0, 0)^T$, $\boldsymbol{\nu}_{\mathbf{x}}=(0, \cos\theta_1, \sin\theta_1)^T$ and $\boldsymbol{\nu}_{\mathbf{y}}=(0, \cos\theta_2, \sin\theta_2)^T$, noting that when $\mathbf{x}$, $\mathbf{y}$ belong to $\Gamma_1$, there holds $\pi\leq\theta_i<2\pi$, and when $\mathbf{x}$, $\mathbf{y}$ belong to $\Gamma_2$, $0\leq\theta_i<\pi, i=1, 2$. The elastic material properties of the nanorod $D$ are characterized by the Lam\'e parameters $c\lambda$ and $c\mu$. While the elastic homogeneous medium in $\mathbb{R}^{3}\setminus D$ is characterized by the Lam\'e parameters $\lambda$ and $\mu$. The constants $\lambda$ and $\mu$ satisfy the strong convexity condition
\begin{equation}\label{3lam}
\mu>0, \quad 3\lambda+ 2\mu>0.
\end{equation}
Now, set
$
(\tilde{\lambda}, \tilde{\mu})=A(\mathbf{x})(\lambda, \mu),
$
where
\begin{equation}\label{Ax}
A(\mathbf{x})= \begin{cases}c, & \mathbf{x} \in D, \\ 1, & \mathbf{x} \in \mathbb{R}^{3} \backslash \overline{D},\end{cases}
\end{equation}
with $c \in \mathbb{C}$.
The non-quasielastic wave scattering due to the presence of nanorod $D$ is described by the following Lam\'e system,
 \begin{equation}\label{3equa}
\begin{cases}
\nabla \cdot \mathcal{C}(\mathbf{x}) \nabla^s \mathbf{u}(\mathbf{x})+\omega^2 \mathbf{u}=0, \quad \mathbf{x}\in \mathbb{R}^3,\\[2pt]
 \mathbf{u}^s=\mathbf{u}(\mathbf{x})-\mathbf{H}(\mathbf{x})\quad \text{satisfies the radiation condition,}
 \end{cases}
\end{equation}
where $\mathbf{u} \in H_{l o c}^1\left(\mathbb{R}^3\right)^3, \mathcal{C}(\mathbf{x})=\tilde{\lambda} \delta_{i j} \delta_{k l}+\tilde{\mu}\left(\delta_{i k} \delta_{j l}+\delta_{j k} \delta_{i l}\right)$, the operator $\nabla^s \mathbf{u}$ is defined by
$$
\nabla^s \mathbf{u}=\left(\nabla \mathbf{u}+\nabla \mathbf{u}^T\right) / 2,
$$
with $T$ denoting the transpose. Here, $\mathcal{L}_{\lambda, \mu} \mathbf{H}+\omega^2\mathbf{H}=\mu \Delta \mathbf{H}+(\lambda+\mu) \nabla \nabla \cdot \mathbf{H}+\omega^2\mathbf{H}=0$ in $\mathbb{R}^3$, which signifies a background field.
In \eqref{3equa}, the radiation condition satisfies the following formulas as $|\mathbf{x}| \rightarrow+\infty$,
$$
(\nabla \times \nabla \times \left(\mathbf{u}-\mathbf{H}\right))(\mathbf{x}) \times \frac{\mathbf{x}}{|\mathbf{x}|}-\mathrm{i} k_{s} \nabla \times \left(\mathbf{u}-\mathbf{H}\right)(\mathbf{x})=\mathcal{O}\left(|\mathbf{x}|^{-2}\right),
$$
$$
\frac{\mathbf{x}}{|\mathbf{x}|} \cdot[\nabla(\nabla \cdot \left(\mathbf{u}-\mathbf{H}\right))](\mathbf{x})-\mathrm{i} k_{p} \nabla \left(\mathbf{u}-\mathbf{H}\right)(\mathbf{x})=\mathcal{O}\left(|\mathbf{x}|^{-2}\right),
$$
where $\mathrm{i}:=\sqrt{-1}$ and
\begin{equation}\label{eqqq:1522}
k_{s}=\omega /c_s,\quad k_{p}=\omega /c_p,
\end{equation}
with
\begin{equation}
c_{s}=\sqrt{\mu},\quad c_{p}=\sqrt{\lambda+2 \mu}.
\end{equation}

\subsection{Layer potentials}
Let $\boldsymbol{\Gamma}^\omega=\left(\Gamma_{i, j}^\omega\right)_{i, j=1}^3$ be the Kupradze matrix of fundamental solution to the operator $\mathcal{L}_{\lambda, \mu}+\omega^{2}$ in $\mathbb{R}^3$, i.e.,
\begin{equation}\label{eq:06}
\left(\Gamma_{i, j}^\omega\right)_{i, j=1}^3(\mathbf{x})=-\frac{\boldsymbol{\delta}_{i j}}{4 \pi \mu|\mathbf{x}|} e^{\mathrm{i} k_s|\mathbf{x}|}+\frac{1}{4 \pi \omega^2} \partial_i \partial_j \frac{e^{\mathrm{i} k_p|\mathbf{x}|}-e^{\mathrm{i} k_s|\mathbf{x}|}}{|\mathbf{x}|},
\end{equation}
and
\begin{equation}\label{gamma}
\begin{aligned}
\mathbf\Gamma^\omega(\boldsymbol{x}-\boldsymbol{y})= & -\frac{1}{4 \pi} \sum_{n=0}^{+\infty} \frac{\mathrm{i}^n}{(n+2) n!}\left(\frac{n+1}{c_s^{n+2}}+\frac{1}{c_p^{n+2}}\right) \omega^n |\boldsymbol{x}-\boldsymbol{y}|^{n-1}\mathbf{I} \\
& +\frac{1}{4 \pi} \sum_{n=0}^{+\infty} \frac{\mathrm{i}^n(n-1)}{(n+2) n!}\left(\frac{1}{c_s^{n+2}}-\frac{1}{c_p^{n+2}}\right) \omega^n|\boldsymbol{x}-\boldsymbol{y}|^{n-3}(\boldsymbol{x}-\boldsymbol{y})(\boldsymbol{x}-\boldsymbol{y})^T.
\end{aligned}
\end{equation}
If $\omega=0$, then $\boldsymbol{\Gamma}:=\boldsymbol{\Gamma}^{0}$ is the Kelvin matrix of the fundamental solution to the operator $\mathcal{L}_{\lambda, \mu}$, is given by
$$
\boldsymbol{\Gamma}(\mathbf{x}-\mathbf{y})=\frac{\alpha_1}{|\boldsymbol{x}-\boldsymbol{y}|}\mathbf{I}+\alpha_2 \frac{(\boldsymbol{x}-\boldsymbol{y})
(\boldsymbol{x}-\boldsymbol{y})^T}{|\boldsymbol{x}-\boldsymbol{y}|^3},
$$
where
$$
\alpha_1=-\frac{1}{8\pi}\left(\frac{1}{\mu}+\frac{1}{\lambda+2 \mu}\right), \quad \alpha_2=-\frac{1}{8\pi}\left(\frac{1}{\mu}-\frac{1}{\lambda+2 \mu}\right).
$$

For $\boldsymbol{\varphi} \in H^{-1 / 2}(\partial D)^{3}$, we define the single and double layer potentials
$$
\begin{gathered}
\mathbf{S}_{D}^{\omega}[\boldsymbol{\varphi}](\mathbf{x})=\int_{\partial D} \boldsymbol{\Gamma}^{\omega}(\mathbf{x}-\mathbf{y}) \boldsymbol{\varphi}(\mathbf{y}) d s(\mathbf{y}), \quad \mathbf{x} \in \mathbb{R}^{3}, \\
\mathbf{D}_{D}^{\omega}[\boldsymbol{\varphi}](\mathbf{x})=\int_{\partial D} \frac{\partial}{\partial \boldsymbol{\nu}_{\mathbf{y}}} \boldsymbol{\Gamma}^{\omega}(\mathbf{x}-\mathbf{y}) \boldsymbol{\varphi}(\mathbf{y}) d s(\mathbf{y}), \quad \mathbf{x} \in \mathbb{R}^{3} \backslash \partial D,
\end{gathered}
$$
where the conormal derivative $\partial / \partial \boldsymbol{\nu}$ of $\mathbf{u}$ on $\partial D$ is defined by
\begin{equation}\label{eq:01}
\frac{\partial \mathbf{u}}{\partial \boldsymbol{\nu}}=\lambda(\nabla \cdot \mathbf{u}) \boldsymbol{\nu}+\mu\left(\nabla \mathbf{u}+\nabla \mathbf{u}^{T}\right) \boldsymbol{\nu}.
\end{equation}
Here, $\boldsymbol{\nu}$ is the outward unit normal to the boundary to the boundary $\partial D$. On the boundary $\partial D$, the conormal derivative of the single layer potential satisfies the following jump relation,
\begin{equation}\label{eq:04}
\frac{\partial}{\partial \boldsymbol{\nu}} \mathbf{S}_{D}^{\omega}[\boldsymbol{\varphi}]|_{\pm}(\mathbf{x})=\left(\pm \frac{1}{2} I+\left(\mathbf{K}_{D}^{\omega}\right)^{*}\right)[\boldsymbol{\varphi}](\mathbf{x}), \quad \text { a.e. } \mathbf{x} \in \partial D,
\end{equation}
where the Neumann-Poincar\'e (NP) operator $\left(\mathbf{K}_{D}^{\omega}\right)^{*}$ is defined by
$$
\left(\mathbf{K}_{D}^{\omega}\right)^{*}[\boldsymbol{\varphi}](\mathbf{x})=\text { p.v. } \int_{\partial D} \frac{\partial}{\partial \boldsymbol{\nu}_{\mathbf{x}}} \boldsymbol{\Gamma}^{\omega}(\mathbf{x}-\mathbf{y}) \boldsymbol{\varphi}(\mathbf{y}) d s(\mathbf{y}),
$$
and $\left(\mathbf{K}_{D}^{\omega}\right)$ is the $L^{2}$-adjoint of $\left(\mathbf{K}_{D}^{\omega}\right)^{*}$, is given by
$$
\mathbf{K}_{D}^{\omega}[\boldsymbol{\varphi}](\mathbf{x})=\text { p.v. } \int_{\partial D} \frac{\partial}{\partial \boldsymbol{\nu}_{\mathbf{y}}} \boldsymbol{\Gamma}^{\omega}(\mathbf{x}-\mathbf{y}) \boldsymbol{\varphi}(\mathbf{y}) d s(\mathbf{y}),
$$
with p.v. denoting the Cauchy principal value. For simplicity, when $\omega=0$, we use $\mathbf{S}_{D}, \mathbf{D}_{D}, \mathbf{K}_{D}$, and $\mathbf{K}_{D}^{*}$ to signify $\mathbf{S}_{D}^{0}, \mathbf{D}_{D}^{0}, \mathbf{K}_{D}^{0}$, and $\left(\mathbf{K}_{D}^{0}\right)^{*}$.
The system \eqref{3equa} can be rewritten as the following transmission problem
\begin{equation}\label{eq:02}
\begin{cases}\mathcal{L}_{\lambda, \mu} \mathbf{u}(\mathbf{x})+\frac{\omega^{2}}{c} \mathbf{u}(\mathbf{x})=\mathbf{0}, & \mathbf{x} \in D \\ \mathcal{L}_{\lambda, \mu} \mathbf{u}(\mathbf{x})+\omega^{2} \mathbf{u}(\mathbf{x})=\mathbf{0}, & \mathbf{x} \in \mathbb{R}^{3} \backslash \overline{D} \\ \left.\mathbf{u}(\mathbf{x})\right|_{-}=\mathbf{u}(\mathbf{x})|_{+}, & \mathbf{x} \in \partial D \\ c\frac{\partial\mathbf{u}(\mathbf{x})}{\partial\boldsymbol{\nu}}|_{-}=\frac{\partial\mathbf{u}(\mathbf{x})}{\partial\boldsymbol{\nu}}|_{+}, & \mathbf{x} \in \partial D\\
\mathbf{u}(\mathbf{x})-\mathbf{H}(\mathbf{x})\quad \text{satisfies the radiation condition.}
\end{cases}
\end{equation}
Using the layer potential method, the solution to \eqref{eq:02} has the following form
\begin{equation}\label{eq:03}
\mathbf{u}(\mathbf{x})= \begin{cases}\mathbf{S}_{D}^{\omega_{1}}[\boldsymbol{\varphi}](\mathbf{x}), & \mathbf{x} \in D, \\[2pt] \mathbf{S}_{D}^{\omega_{2}}[\boldsymbol{\psi}](\mathbf{x})+\mathbf{H}(\mathbf{x}), & \mathbf{x} \in \mathbb{R}^{3} \backslash \overline{D},\end{cases}
\end{equation}
where $\boldsymbol{\varphi}, \boldsymbol{\psi} \in H^{-1 / 2}(\partial D)^{3}$ and $\omega_{1}=\frac{\omega}{\sqrt{c}}, \omega_{2}=\omega$. In order to satisfy the third and forth conditions in \eqref{eq:02} across $\partial D$, and using the jump relation \eqref{eq:04}, we can obtain
\begin{equation}\label{eq:16140000a}
\begin{cases}\mathbf{S}_{D}^{\omega_{1}}[\boldsymbol{\varphi}]-\mathbf{S}_{D}^{\omega_{2}}[\boldsymbol{\psi}]=\mathbf{H} & \text { on } \partial D \\[5pt] c\left(-\frac{I}{2}+\left(\mathbf{K}_{D}^{\omega_1}\right)^*\right)[\boldsymbol{\varphi}]-\left(\frac{I}{2} +\left(\mathbf{K}_{D}^{\omega_2}\right)^*\right)[\boldsymbol{\psi}]=\frac{\partial \mathbf{H}}{\partial \boldsymbol{\nu}} & \text { on } \partial D\end{cases}
\end{equation}

It is known that for $\omega$ small enough, $\mathbf{S}_{D}^{\omega}$ is invertible. Therefore, by using the first equation in \eqref{eq:16140000a}, one can directly obtain that
$$
\boldsymbol{\varphi}=\left(\mathbf{S}_{D}^{\omega_1}\right)^{-1}\left(\mathbf{S}_{D}^{\omega_2}[\boldsymbol{\psi}]+\mathbf{H}\right).
$$
Then, from the second equation in \eqref{eq:16140000a}, we have that
\begin{equation}\label{eq:main1}
\mathcal{A}_{D}(\omega)[\boldsymbol{\psi}]=\mathbf{F},
\end{equation}
where
\begin{equation}\label{F}
\begin{aligned}
\mathcal{A}_{D}(\omega) & =c\left(-\frac{I}{2}+\left(\mathbf{K}_{D}^{\omega_1}\right)^*\right)\left(\mathbf{S}_{D}^{\omega_1}\right)^{-1} \mathbf{S}_{D}^{\omega_2}-\left(\frac{I}{2}+\left(\mathbf{K}_{D}^{\omega_2}\right)^*\right),\\
\mathbf{F} & =\frac{\partial \mathbf{H}}{\partial \boldsymbol{\nu}}-c\left(-\frac{I}{2}+\left(\mathbf{K}_{D}^{\omega_1}\right)^*\right)\left(\mathbf{S}_{D}^{\omega_1}\right)^{-1}\left[\mathbf{H}\right].
\end{aligned}
\end{equation}

Clearly,

\begin{equation}\label{eq:AD0}
\begin{aligned}
\mathcal{A}_{D}(0)=\mathcal{A}_{D, 0} & =c\left(-\frac{I}{2}+\mathbf{K}_{D}^*\right)-\left(\frac{I}{2}+\mathbf{K}_{D}^*\right) \\
& =-\frac{c+1}{2}I+(c-1)\mathbf{K}_{D}^*.
\end{aligned}
\end{equation}
For subsequent use, we define $S_1^a$ and $S_1^b$ by
$$
\begin{aligned}
&S_1^a=\{\mathbf{x}|(x_1+L/2)^2+x_2^2+x_3^2=1, x_1<-L/2\},\\
&S_1^b=\{\mathbf{x}|(x_1-L/2)^2+x_2^2+x_3^2=1, x_1>L/2\},
\end{aligned}
$$
and also introduce the following regions:
$$
\iota_\delta(P):=\left\{\mathbf{x} ;\left|P-\mathbf{z}_{\mathbf{x}}\right|=\mathcal{O}(\delta), \mathbf{x} \in S^f\right\},
$$
$$
\iota_\delta(Q):=\left\{\mathbf{x} ;\left|Q-\mathbf{z}_{\mathbf{x}}\right|=\mathcal{O}(\delta), \mathbf{x} \in S^f\right\}.
$$
Let
\begin{small}
\begin{equation}\label{Psi}
\boldsymbol{\Psi}:=\left\{\left.\mathbf{v}\right|_{\partial D} \in H^{-1 / 2}(\partial D)^{3} \mid \mathbf{v}=\left(v_1, v_2, v_3\right)^T \text { satisfies } \partial_j v_i+\partial_i v_j=0 \text { for all } 1 \leq i, j \leq 3\right\} .
\end{equation}
\end{small}
Define $\tilde{\boldsymbol{\psi}}(\tilde{\mathbf{x}}):=\boldsymbol{\psi}(\mathbf{x})$, where $\mathbf{x} \in S^a, S^b$ and $\tilde{\mathbf{x}} \in S_1^a, S_1^b$.

\subsection{Statement of main results }
With the above preparation, we present our main results on quantitative analysis of the perturbed non-quasielastic field and the proofs will be given in the last section.

Before proceeding, we shall define some boundary integral for later usage. Let $\mathbf{y}\in S_1^f$ be parameterized by $(y,\theta_y)$, where $y\in[-L/2, L/2]$ and $\theta\in [0, 2\pi)$. For $\mathbf{x}\neq \mathbf{z}_{\mathbf{y}}$, $\boldsymbol{\phi}\in L^2[-L/2, L/2]^3$ and $\mathbf{N}\in L^2[-L/2, L/2]^{3\times 3}$, define
\begin{equation}\label{AA1}
\mathcal{A}_1[\boldsymbol{\phi}](\mathbf{x})=\int_{-L/2}^{L/2}\mathbf\Gamma\left(\mathbf{x}-\mathbf{z}_{\mathbf{y}}\right)\boldsymbol{\phi}(y) dy,
\end{equation}
and
\begin{equation}\label{AA2}
\mathcal{A}_2[\mathbf{N}](\mathbf{x})=\int_{-L/2}^{L/2}\sum_{i=2}^3\partial_{x_i}\mathbf\Gamma\left(\mathbf{x}-\mathbf{z}_{\mathbf{y}}\right)\mathbf{N}(y)\mathbf{e}_i dy.
\end{equation}
\begin{thm}\label{3main}
Let $\mathbf{u}$ be the solution of \eqref{3equa} related to the rod-shaped inclusion $D$, where the parameters are defined in \eqref{3lam}-\eqref{Ax}. Then, for $\mathbf{x} \in \mathbb{R}^3 \backslash \overline{D}$, it holds that
\begin{equation}\label{u}
\begin{aligned}
&\mathbf{u}(\mathbf{x})-\mathbf{H}(\mathbf{x})
=\delta^2\mathcal{A}_1[\mathbf{G}_{\mathbf{H}_0}+\omega\mathbf{G}_{\mathbf{H}_1}](\mathbf{x})
-\delta^2\mathcal{A}_2[\mathbf{R}_{\mathbf{H}_0}+\omega\mathbf{R}_{\mathbf{H}_1}](\mathbf{x})\\
&-\pi\delta^2(\lambda_1-\frac{1}{2})^{-1}\mathbf\Gamma\left(\mathbf{x}-P\right)\left[\lambda\Big(\nabla\cdot \big(\mathbf{H}_0(P)+\omega\mathbf{H}_1(P)\big)\Big)\mathbf{e}_1+2\mu\nabla^s\big(\mathbf{H}_0(P)+\omega\mathbf{H}_1(P)\big)\mathbf{e}_1\right]
\\&+\pi\delta^2(\lambda_1-\frac{1}{2})^{-1}\mathbf\Gamma\left(\mathbf{x}-Q\right)\left[\lambda\Big(\nabla\cdot \big(\mathbf{H}_0(Q)+\omega\mathbf{H}_1(Q)\big)\Big)\mathbf{e}_1+2\mu\nabla^s\big(\mathbf{H}_0(Q)+\omega\mathbf{H}_1(Q)\big)\mathbf{e}_1\right]\\
&+\alpha_3\omega\delta^2\int_{-L/2}^{L/2}\mathbf{G}_{\mathbf{H}_0}(y) dy-\alpha_3\pi\omega\delta^2(\lambda_1-\frac{1}{2})^{-1}\left[\lambda(\nabla\cdot \mathbf{H}_0(P))\mathbf{e}_1+2\mu\nabla^s \mathbf{H}_0(P)\mathbf{e}_1\right]\\
&+\alpha_3\pi\omega\delta^2(\lambda_1-\frac{1}{2})^{-1}\left[\lambda(\nabla\cdot \mathbf{H}_0(Q))\mathbf{e}_1+2\mu\nabla^s \mathbf{H}_0(Q)\mathbf{e}_1\right]\\
&+o(\delta^2)+\omega\cdot o(\delta^2)+\delta\cdot\mathcal{O}(\omega^2),
   \end{aligned}
\end{equation}
where $\lambda_1=\frac{c+1}{2(c-1)}$ and $\alpha_3$ will be given in \eqref{al3}. $\mathbf{H}_0$ and $\mathbf{H}_1$ are the zeroth-order and first-order terms of $\mathbf{H}$ with respect to $\omega$, respectively. $\mathbf{G}_{\mathbf{H}_i}, i=0, 1$ is defined by
\begin{equation}\label{eq:compG}
\begin{aligned}
\mathbf{G}_{\mathbf{H}_i}=&2\pi\Big((\lambda+2\mu)+ \frac{\lambda+\mu}{2\lambda_1-1}\Big)l_1 B\Big(\lambda\nabla\nabla\cdot\mathbf{H}_i(\mathbf{z}_{\mathbf{y}})+2\mu^2l_2\nabla\nabla\cdot \big(B \mathbf{H}_i(\mathbf{z}_{\mathbf{y}})\big)\Big)\\
& + \frac{\mu\pi}{\lambda_1}\Big(I_3-\mu l_2B\Big)\Big(\big(B\nabla\big)^T\big(\nabla\mathbf{H}_i(\mathbf{z}_{\mathbf{y}})+\nabla\mathbf{H}_i(\mathbf{z}_{\mathbf{y}})^T\big)\Big)^T\\
& +\frac{\mu\pi}{\lambda_1(2\lambda_1-1)}\Big(I_3-\mu l_2B\Big)\Big(\Delta \mathbf{H}_i(\mathbf{z}_{\mathbf{y}})-\partial_{y_1}^2 \mathbf{H}_i(\mathbf{z}_{\mathbf{y}})+ \nabla\nabla\cdot \big(B \mathbf{H}_i(\mathbf{z}_{\mathbf{y}})\big)\Big)\\
& -\frac{2\pi\mu^2}{2\lambda_1-1}l_2C\nabla\Big((\nabla\times \mathbf{H}_i(\mathbf{z}_{\mathbf{y}}))\cdot \mathbf{e}_1\Big),
\end{aligned}
\end{equation}\
with $l_1$ and $l_2$ are
\begin{equation}\label{l1}
l_1=\frac{1}{2\lambda_1(\lambda+2\mu)-\mu}, \quad \quad
l_2=\frac{1}{2\lambda_1(\lambda+2\mu)+\mu}.
\end{equation}
$I_3$ represents the three-dimensional identity matrix. The matrix $B$ and $C$ are defined by
\begin{equation}\label{eq:BC}
B=\left(\begin{array}{ccc} 0& 0 &0 \\
0 & 1& 0 \\
0& 0 & 1\end{array} \right), \quad
C=\left(\begin{array}{ccc} 0& 0 &0 \\
0 & 0& 1 \\
0& -1 & 0\end{array} \right).
\end{equation}
$\mathbf{R}_{\mathbf{H}_i}, i=0, 1$ is defined by
\begin{equation}\label{3F}
\begin{aligned}
\mathbf{R}_{\mathbf{H}_i}=&\pi\Big(\big(l_{01}\nabla\cdot \mathbf{H}_i(\mathbf{z}_{\mathbf{y}})+l_{02}\nabla\cdot(B\mathbf{H}_i(\mathbf{z}_{\mathbf{y}}))\big)I_3 \\
&+\frac{\mu}{\lambda_1}\big(I_3-\mu l_2 B\big)\big(\nabla \mathbf{H}_i(\mathbf{z}_{\mathbf{y}})+\nabla \mathbf{H}_i(\mathbf{z}_{\mathbf{y}})^T\big)\Big),
\end{aligned}
\end{equation}
with $l_{0i}, i=1, 2$ are defined by
\begin{equation}\label{l0}
l_{01}=
\frac{2\lambda(\lambda+2\mu)}{2(\lambda+2\mu)\lambda_1+\lambda}, \quad \quad
l_{02}=-\frac{2\mu(\lambda-\mu)(\lambda+2\mu)}{[2(\lambda+2\mu)\lambda_1+\mu][2(\lambda+2\mu)\lambda_1+\lambda]}.
\end{equation}
\end{thm}

\begin{rem}
It can be seen from Theorem \ref{3main} that the scattered field concentrate near the two end points (a.k.a, high curvature parts) of the nanorod, which demonstrate that high curvature is a primary driver of intense elastic field concentration.
\end{rem}

With the asymptotic result for the non-quasielastic field, we are now in the position of presenting our main results regarding polariton resonance.
Define the elastic energy functional
\begin{equation}\label{Eu}
E(\mathbf{u})=\int_{\mathbb{R}^{3} \backslash \overline{D}}\nabla^s \overline{\mathbf{u}(\mathbf{x})}: \mathcal{C} \nabla^s \mathbf{u}(\mathbf{x}) d \mathbf{x}.
\end{equation}
Here $A: B=\sum_{i, j} a_{i j} b_{i j} \text { for two matrices } A=\left(a_{i j}\right) \text { and } B=\left(b_{i j}\right).$ In fact, the macroscopic effect exhibited by polarization resonance is the maximization of elastic energy, or even its blow up. Therefore, we will explore in detail the conditions under which the elastic parameters, frequency, and geometry lead to maximized or blown up elastic energy, and further characterize the intrinsic relationship between the blow-up rate and the aforementioned parameters. As for the rigorous definition of resonance, we can refer to \cite{DYLHZGH}. Next we obtain the results for polariton resonance as follows.
\begin{thm}\label{Thm11}
{\rm Let $c=c_0+i\varrho$, where $\varrho>0$ is the loss parameter. Let $\mathbf{u}^s$ be the scattering solution of \eqref{3equa}.\\
(1) If $c_0=-\frac{\lambda+\mu}{\lambda+3\mu}$ and the source densities $\mathbf{G}^{11}_{\mathbf{H}_0}(y) \not\equiv \mathbf{0}$ or $\mathbf{R}^{11}_{\mathbf{H}_0}(y)\mathbf{e}_i \not\equiv \mathbf{0}$ for $i=2, 3$ on $[-L/2, L/2]$, let
$$
\mathbf{B}_0^1(\mathbf{x})
=\mathcal{A}_1[\mathbf{G}^{11}_{\mathbf{H}_0}](\mathbf{x})-\mathcal{A}_2[\mathbf{R}^{11}_{\mathbf{H}_0}](\mathbf{x}),
$$
where the operators $\mathcal{A}_1$, $\mathcal{A}_2$ are defined in \eqref{AA1}-\eqref{AA2}, and $\mathbf{G}^{11}_{\mathbf{H}_0}$,  $\mathbf{R}^{11}_{\mathbf{H}_0}$ are given in \eqref{eqG11} and \eqref{eqR11}.
Then it holds
$$
E(\mathbf{u}^s)\sim\frac{4(\lambda+2\mu)^2\delta^4}{(\lambda+3\mu)^4\varrho^2}\int_{\mathbb{R}^{3} \backslash \overline{D}}\nabla^s \overline{\mathbf{B}_0^1(\mathbf{x})}: \mathcal{C} \nabla^s \mathbf{B}_0^1(\mathbf{x}) d \mathbf{x}=\mathcal{O}(\varrho^{-2}\delta^4).
$$
Furthermore, if $\varrho=o\left(\delta^2\right)$ (as $\delta \rightarrow 0, \varrho \rightarrow 0$ ), then it follows that
$$
\sqrt{E(\mathbf{u}^s)} \rightarrow \infty.
$$
(2) If $c_0=-\frac{\mu}{\lambda+\mu}$ and the source densities $\mathbf{R}^{21}_{\mathbf{H}_0}(y)\mathbf{e}_i \not\equiv \mathbf{0}$ for $i=2, 3$ on $[-L/2, L/2]$, where
\begin{equation}
\begin{aligned}
\mathbf{R}^{21}_{\mathbf{H}_0}=&\pi\Big(\nabla\cdot \mathbf{H}_0(\mathbf{z}_{\mathbf{y}})-\frac{\mu(\lambda-\mu)}{\lambda\big(2(\lambda+2\mu)\lambda_1+\mu\big)}\nabla\cdot(B\mathbf{H}_0(\mathbf{z}_{\mathbf{y}}))\Big)I_3,
\end{aligned}
\end{equation}
then it holds
$$
E(\mathbf{u}^s)\sim\frac{\lambda^2(\lambda+2\mu)^4\delta^4}{(\lambda+\mu)^4\varrho^2}\int_{\mathbb{R}^{3} \backslash \overline{D}}\nabla^s \overline{\mathcal{A}_2[\mathbf{R}^{21}_{\mathbf{H}_0}](\mathbf{x})}: \mathcal{C} \nabla^s \mathcal{A}_2[\mathbf{R}^{21}_{\mathbf{H}_0}](\mathbf{x}) d \mathbf{x}=\mathcal{O}(\varrho^{-2}\delta^4).
$$
Furthermore, if $\varrho=o\left(\delta^2\right)$ (as $\delta \rightarrow 0, \varrho \rightarrow 0$ ), then it follows that
$$
\sqrt{E(\mathbf{u}^s)} \rightarrow \infty.
$$
(3) If $c_0\rightarrow -1$ and the source densities $\mathbf{G}^{31}_{\mathbf{H}_0}(y) \not\equiv \mathbf{0}$ or $\mathbf{R}^{31}_{\mathbf{H}_0}(y)\mathbf{e}_i \not\equiv \mathbf{0}$ for $i=2, 3$ on $[-L/2, L/2]$, let
$$
\mathbf{B}_0^3(\mathbf{x})=\mathcal{A}_1[\mathbf{G}^{31}_{\mathbf{H}_0}](\mathbf{x})-\mathcal{A}_2[\mathbf{R}^{31}_{\mathbf{H}_0}](\mathbf{x}),
$$
where
$\mathbf{G}^{31}_{\mathbf{H}_0}$ and $\mathbf{R}^{31}_{\mathbf{H}_0}$ are defined in \eqref{GH4} and \eqref{RH4}, respectively. Then it holds
$$
E(\mathbf{u}^s)\sim\frac{16\delta^4}{(c_0+1)^2+\varrho^2}\int_{\mathbb{R}^{3} \backslash \overline{D}}\nabla^s \overline{\mathbf{B}_0^3(\mathbf{x})}: \mathcal{C} \nabla^s \mathbf{B}_0^3(\mathbf{x}) d \mathbf{x}=\mathcal{O}\left(\frac{\delta^4}{(c_0+1)^2+\varrho^2}\right).
$$
Furthermore, if $(c_0+1)^2+\varrho^2=o\left(\delta^4\right)$ (as $\delta \rightarrow 0, \varrho \rightarrow 0$ ), then it follows that
$$
\sqrt{E(\mathbf{u}^s)} \rightarrow \infty.
$$
(4) If $c_0\rightarrow +\infty$ and $\mathbf{G}^{41}_{\mathbf{H}_0}(y) \not\equiv \mathbf{0}$ on $[-L/2, L/2]$ or $\mathbf{C}_P \neq \mathbf{0}$ or $\mathbf{C}_Q \neq \mathbf{0}$, let
\begin{equation}
\begin{aligned}
\mathbf{B}_0^4(\mathbf{x})
=&\mathcal{A}_1[\mathbf{G}^{41}_{\mathbf{H}_0}](\mathbf{x})-\pi\mathbf\Gamma\left(\mathbf{x}-P\right)\mathbf{C}_P
+\pi\mathbf\Gamma\left(\mathbf{x}-Q\right)\mathbf{C}_Q,
\end{aligned}
\end{equation}
where $\mathbf{G}^{41}_{\mathbf{H}_0}$ is defined in \eqref{GH5}, $\mathbf{C}_P = \lambda\big(\nabla\cdot \mathbf{H}_0(P)\big)\mathbf{e}_1+2\mu\nabla^s\mathbf{H}_0(P)\mathbf{e}_1$ and $\mathbf{C}_Q = \lambda\big(\nabla\cdot \mathbf{H}_0(Q)\big)\mathbf{e}_1+2\mu\nabla^s\mathbf{H}_0(Q)\mathbf{e}_1$. Then it holds
$$
E(\mathbf{u}^s)\sim|c-1|^2\delta^4\int_{\mathbb{R}^{3} \backslash \overline{D}}\nabla^s \overline{\mathbf{B}_0^4(\mathbf{x})}: \mathcal{C} \nabla^s \mathbf{B}_0^4(\mathbf{x}) d \mathbf{x}=\mathcal{O}(\delta^4c_0^2).
$$
Furthermore, if $c_0^{-2}=o\left(\delta^4\right)$ (as $\delta \rightarrow 0$ ), then it follows that
$$
\sqrt{E(\mathbf{u}^s)} \rightarrow \infty.
$$
}
\end{thm}
\begin{rem}
We highlight a unique feature of the resonance in Case (4) of Theorem \ref{Thm11} that distinguishes it from the surface polariton resonances in Cases (1)-(3). In Cases (1)-(3), the resonant fields are distributed along the entire lateral surface of the nanorod. In sharp contrast, the energy concentration in Case (4) occurs exclusively near the two endpoints $P$ and $Q$.\\
To visualize this phenomenon explicitly, consider the specific background field $\mathbf{H}_0(\mathbf{x})$, which, as we will verify later, satisfies the static Lam\'e system:
$$
\mathbf{H}_0(\mathbf{x}) = (x_2, x_1, 0)^T.
$$
In this scenario, direct calculations show that the distributed source density vanishes identically, i.e., $\mathbf{G}^{41}_{\mathbf{H}_0}(y) \equiv \mathbf{0}$ on $[-L/2, L/2]$. Consequently, the layer potential term $\mathcal{A}_1$ makes no contribution, and the corrector field $\mathbf{B}_0^4$ is determined solely by:
$$
\mathbf{B}_0^4(\mathbf{x}) = 2\pi\mu \Big[ \mathbf\Gamma(\mathbf{x}-Q) - \mathbf\Gamma(\mathbf{x}-P) \Big] \mathbf{e}_2.
$$
Substituting this explicit expression into the energy estimate in Theorem \ref{Thm11} (4), we obtain:
\begin{small}
$$
E(\mathbf{u}^s) \sim 4\pi^2\mu^2|c-1|^2\delta^4\int_{\mathbb{R}^{3} \backslash \overline{D}} \nabla^s \left( \Big[ \mathbf\Gamma(\mathbf{x}-Q) - \mathbf\Gamma(\mathbf{x}-P) \Big] \mathbf{e}_2 \right) : \mathcal{C} \nabla^s \left( \Big[ \mathbf\Gamma(\mathbf{x}-Q) - \mathbf\Gamma(\mathbf{x}-P) \Big] \mathbf{e}_2 \right) d \mathbf{x}.
$$
\end{small}
This integral representation clearly demonstrates that the scattering energy is generated entirely by the interaction between the fundamental solutions centered at the two endpoints $P$ and $Q$, which shows the enhancement of resonance near the high curvature points.
\end{rem}
\section{Asymptotic expansions of layer potentials}
In this section, we shall consider the asymptotic behavior of the perturbed elastic field, while the key is to derive some key asymptotic behavior for layer potentials related to lam\'e system.

Define a new inner product on $H^{-1 / 2}(\partial D)^3$ by

\begin{equation}\label{eq111}
(\boldsymbol{\varphi}, \boldsymbol{\psi})_{\mathcal{H}^*}:=-\left\langle\boldsymbol{\varphi}, \mathbf{S}_D[\boldsymbol{\psi}]\right\rangle, \quad \boldsymbol{\varphi}, \boldsymbol{\psi} \in H^{-1 / 2}(\partial D)^3,
\end{equation}

where $\langle\cdot, \cdot\rangle$ is the duality pairing of $H^{-1 / 2}(\partial D)^3$ and $H^{1 / 2}(\partial D)^3$. Since the operator $\mathbf{S}_D$ is invertible from $H^{-1 / 2}(\partial D)^3$ to $H^{1 / 2}(\partial D)^3$, hence the right side of \eqref{eq111} is well-defined. It can be verified that $(\cdot, \cdot)_{\mathcal{H}^*}$ is the inner product on $H^{-1 / 2}(\partial D)^3$ and we denote by $\|\cdot\|_{\mathcal{H}^*}$ the norm induced by this inner product. We denote the space $H^{-1 / 2}(\partial D)^3$ with the new inner product by $\mathcal{H}^*$. Also define a new inner product on $H^{1 / 2}(\partial D)^3$  by
\begin{equation}
(\boldsymbol{f}, \boldsymbol{g})_{\mathcal{H}}:=-\left\langle\mathbf{S}_D^{-1}[\boldsymbol{f}], \boldsymbol{g}\right\rangle, \quad \boldsymbol{f}, \boldsymbol{g} \in H^{1 / 2}(\partial D)^3,
\end{equation}
$(\cdot, \cdot)_{\mathcal{H}}$ is the inner product on $H^{1 / 2}(\partial D)^3$ and we denote by $\|\cdot\|_{\mathcal{H}}$ the norm induced by this inner product. We denote the space $H^{1 / 2}(\partial D)^3$ with the new inner product by $\mathcal{H}$. Moreover, $\mathbf{S}_D$ is an isometry between $\mathcal{H}^*$ and $\mathcal{H}$. $A \lesssim B$ means $A \leq C B$ for some generic positive constant $C$. In this paper, the notation $A \sim B$ (as $\delta \rightarrow 0$ ) is defined to mean that $A$ is asymptotically equivalent to $B$.

\begin{lem}\label{le:14151111}
For static single layer potential operator $\mathbf{S}_D$, we have
\begin{equation}\label{eq:13520927}
\mathbf{S}_{D}[\boldsymbol{\varphi}](\mathbf{x})=-4 \pi\delta \ln \delta\left(\begin{array}{c}

\left(\alpha_1+\alpha_2\right) \kappa_0^1 \\

\alpha_1 \kappa_0^2 \\

\alpha_1 \kappa_0^3

\end{array}\right)+o(\delta \ln \delta),
\end{equation}
and for any $\boldsymbol{\varphi}\in \mathcal{H}^*$, $\mathbf{u}\in \mathcal{H}$, it holds
\begin{equation}\label{eq:15510927}
\begin{aligned}
\|\boldsymbol{\varphi}\|_{\mathcal{H}^*}&\sim \sqrt{8} \pi \left(\delta |\ln \delta|^{1/2}\right) \sqrt{-\bigg(\int_{-L / 2}^{L / 2}  \Big(\left(\alpha_1+\alpha_2\right)(\kappa_0^1)^2+\alpha_1(\kappa_0^2)^2+\alpha_1(\kappa_0^3)^2\Big) d x_1\bigg)},\\
\|\mathbf{u}\|_{\mathcal{H}}&=\mathcal{O}(\delta |\ln \delta|^{1/2}),
\end{aligned}
\end{equation}
where
$$
\kappa^i_0=\frac{1}{2 \pi} \int_0^{2 \pi} \varphi_i\left(x_1, \theta_1\right)d \theta_1,\quad i=1, 2, 3,
$$
and
$$
\alpha_1=-\frac{1}{8\pi}\left(\frac{1}{\mu}+\frac{1}{\lambda+2 \mu}\right)<0, \quad
\alpha_2=-\frac{1}{8\pi}\left(\frac{1}{\mu}-\frac{1}{\lambda+2 \mu}\right)<0.
$$
\end{lem}
\begin{proof}
From the definition of $\mathbf{S}_{D}[\boldsymbol{\varphi}]$ and the geometry of the domain $D$, we know that
\begin{equation}\label{13588}
\begin{aligned}
\mathbf{S}_{D}[\boldsymbol{\varphi}](\mathbf{x})\sim &\int_{ S^f}\left(\frac{\alpha_1}{|\mathbf{x}-\mathbf{y}|}\mathbf{I}+\alpha_2 \frac{(\mathbf{x}-\mathbf{y})
(\mathbf{x}-\mathbf{y})^T}{|\mathbf{x}-\mathbf{y}|^3}\right)\boldsymbol{\varphi}(\mathbf{y}) d s(\mathbf{y})\\=&\alpha_1\int_{S^f} \frac{1}{|\mathbf{x}-\mathbf{y}|}\boldsymbol{\varphi}(\mathbf{y}) d s(\mathbf{y})+\alpha_2 \int_{S^f}\frac{(\mathbf{x}-\mathbf{y})
(\mathbf{x}-\mathbf{y})^T}{|\mathbf{x}-\mathbf{y}|^3}\boldsymbol{\varphi}(\mathbf{y}) d s(\mathbf{y})\\=&\alpha_1\mathbf{T}[\boldsymbol{\varphi}]+\alpha_2\mathbf{V}[\boldsymbol{\varphi}]=\alpha_1\left(\begin{array}{c}T_1\\
T_2\\
T_3
\end{array} \right)
+\alpha_2\left(\begin{array}{c}V_1\\
V_2\\
V_3
\end{array} \right),
\end{aligned}
\end{equation}
where
$$
T_i(\mathbf{x})=\int_{S^f} \frac{1}{|\mathbf{x}-\mathbf{y}|}\varphi_i(\mathbf{y}) d s(\mathbf{y}),
$$
$$
\begin{aligned}
V_1=&\int_{S^f}\frac{(x_1-y_1)^2}{|\mathbf{x}-\mathbf{y}|^3}\varphi_1(\mathbf{y}) d s(\mathbf{y})+\int_{S^f}\frac{(x_1-y_1)(x_2-y_2)}{|\mathbf{x}-\mathbf{y}|^3}\varphi_2(\mathbf{y}) d s(\mathbf{y})\\&+\int_{S^f}\frac{(x_1-y_1)(x_3-y_3)}{|\mathbf{x}-\mathbf{y}|^3}\varphi_3(\mathbf{y}) d s(\mathbf{y})=V_1^1+V_1^2+V_1^3,
\end{aligned}
$$
$$
\begin{aligned}
V_2=&\int_{S^f}\frac{(x_1-y_1)(x_2-y_2)}{|\mathbf{x}-\mathbf{y}|^3}\varphi_1(\mathbf{y}) d s(\mathbf{y})+\int_{S^f}\frac{(x_2-y_2)^2}{|\mathbf{x}-\mathbf{y}|^3}\varphi_2(\mathbf{y}) d s(\mathbf{y})\\&+\int_{S^f}\frac{(x_2-y_2)(x_3-y_3)}{|\mathbf{x}-\mathbf{y}|^3}\varphi_3(\mathbf{y}) d s(\mathbf{y})=V_2^1+V_2^2+V_2^3,
\end{aligned}
$$
and
$$
\begin{aligned}
V_3=&\int_{S^f}\frac{(x_1-y_1)(x_3-y_3)}{|\mathbf{x}-\mathbf{y}|^3}\varphi_1(\mathbf{y}) d s(\mathbf{y})+\int_{S^f}\frac{(x_2-y_2)(x_3-y_3)}{|\mathbf{x}-\mathbf{y}|^3}\varphi_2(\mathbf{y}) d s(\mathbf{y})\\&+\int_{S^f}\frac{(x_3-y_3)^2}{|\mathbf{x}-\mathbf{y}|^3}\varphi_3(\mathbf{y}) d s(\mathbf{y})=V_3^1+V_3^2+V_3^3.
\end{aligned}
$$
Firstly, we analyze $T_i, i=1, 2, 3$. Through cylindrical coordinate transformation, we have
$$
\begin{aligned}
T_i(\mathbf{x})=&\int_{S^f} \frac{1}{|\mathbf{x}-\mathbf{y}|}\varphi_i(\mathbf{y}) d s(\mathbf{y})\\=&\int_{0}^{2\pi}\int_{-L/2}^{L/2}\frac{\delta}{\big((x_1-y_1)^2+2\delta^2(1-\cos(\theta_1-\theta_2))\big)^{1/2}}
\varphi_i(y_1, \theta_2)dy_1d\theta_2,
\end{aligned}
$$
by the change of variables $\frac{y_1-x_1}{\delta}=t$ and $2(1-\cos(\theta_1-\theta_2))=u$, the following holds
$$
\begin{aligned}
T_i(\mathbf{x})=&\int_{0}^{2\pi}\int_{\frac{-L/2-x_1}{\delta}}^{\frac{L/2-x_1}{\delta}}\frac{\delta}{\big(t^2+u\big)^{1/2}}
\varphi_i(x_1+\delta t, \theta_2)dtd\theta_2\\ \sim &\int_{0}^{2\pi}\varphi_i(x_1, \theta_2)\int_{\frac{-L/2-x_1}{\delta}}^{\frac{L/2-x_1}{\delta}}\frac{\delta}{\big(t^2+u\big)^{1/2}}dtd\theta_2,
\end{aligned}
$$
then
$$
\begin{aligned}
\delta J^{\prime}=&\int_{\frac{-L/2-x_1}{\delta}}^{\frac{L/2-x_1}{\delta}}\frac{\delta}{\big(t^2+u\big)^{1/2}}dt
=\delta\left.\operatorname{arsinh}\left(\frac{t}{\sqrt{u}}\right)\right|_{\frac{-L/2-x_1}{\delta}}^{\frac{L/2-x_1}{\delta}}\\=&\delta\left[\operatorname{arsinh}\left(\frac{L / 2-x_1}{\delta \sqrt{u}}\right)-\operatorname{arsinh}\left(\frac{-L / 2-x_1}{\delta\sqrt{u}}\right)\right].
\end{aligned}
$$
For $x \gg 1$, $\operatorname{arsinh}(x) \sim \ln (2 x)$, then
$$
\delta J^{\prime}\sim \delta\ln \left(\frac{4\left(L^2 / 4-x_1^2\right)}{\delta^2 u}\right)=\delta \ln [4(L^2 / 4-x_1^2)]-2\delta\ln\delta-\delta\ln\big[2\big(1-\cos(\theta_1-\theta_2)\big)\big].
$$
Thus
$$
\begin{aligned}
T_i(\mathbf{x})\sim & \int_{0}^{2\pi}\varphi_i(x_1, \theta_2)\left(\delta \ln [4(L^2 / 4-x_1^2)]-2\delta\ln\delta-\delta\ln\big[2\big(1-\cos(\theta_1-\theta_2)\big)\big]\right)d\theta_2\\=&\left(\delta \ln [4(L^2 / 4-x_1^2)]-2\delta\ln\delta\right)\int_{0}^{2\pi}\varphi_i(x_1, \theta_2)d\theta_2\\&-\delta\int_{0}^{2\pi}\varphi_i(x_1, \theta_2)\ln\big[2\big(1-\cos(\theta_1-\theta_2)\big)\big]d\theta_2.
\end{aligned}
$$
Using Fourier series expansion, we have
$$
\varphi_i\left(x_1, \theta^{\prime}\right)=\sum_{n=-\infty}^{\infty} \kappa^i_n e^{i n \theta^{\prime}},
$$
where
$$
\kappa^i_n=\frac{1}{2 \pi} \int_0^{2 \pi} \varphi_i\left(x_1, \theta^{\prime}\right) e^{-i n \theta^{\prime}} d \theta^{\prime}.
$$
Let $\theta_2-\theta_1=\theta$, by direct calculation and using the fact that $\ln |\sin (x)|=-\ln (2)-\sum_{k=1}^{\infty} \frac{\cos (2 k x)}{k}$, we can obtain
$$
\ln2[1-\cos\theta]=-\sum_{n \neq 0}^{\infty} \frac{e^{i n \theta}}{|n|},
$$
applying the convolution theorem, we can deduce that
\begin{equation}\label{eq:070926}
-\delta\int_{0}^{2\pi}\varphi_i(x_1, \theta_2)\ln2[1-\cos(\theta_2-\theta_1)]d\theta_2=2 \pi\delta \sum_{n \neq 0} \frac{\kappa^i_n}{|n|} e^{i n \theta_1}.
\end{equation}
Thus
$$
\begin{aligned}
T_i(\mathbf{x})\sim 2\pi\left(\delta \ln [4(L^2 / 4-x_1^2)]-2\delta\ln\delta\right)\kappa^i_0+2 \pi\delta \sum_{n \neq 0} \frac{\kappa^i_n}{|n|} e^{i n \theta_1}.
\end{aligned}
$$
We can see that the order of $V_1^1$ with respect to $\delta$ is $\mathcal{O}(\delta\ln\delta)$.
Next we we analyze $V_i, i=1, 2, 3$. Using the cylindrical coordinate transformation again, we can obtain
$$
\begin{aligned}
V_1^1=&\int_{S^f}\frac{(x_1-y_1)^2}{|\mathbf{x}-\mathbf{y}|^3}\varphi_1(\mathbf{y}) d s(\mathbf{y})\\ = &\int_{0}^{2\pi}\int_{-L/2}^{L/2}\frac{\delta(x_1-y_1)^2}{\big((x_1-y_1)^2+2\delta^2(1-\cos(\theta_1-\theta_2))\big)^{3/2}}
\varphi_1(y_1, \theta_2)dy_1d\theta_2.
\end{aligned}
$$
By the change of variables $\frac{y_1-x_1}{\delta}=t$, it holds
$$
\begin{aligned}
V_1^1=&\delta\int_{0}^{2\pi}\int_{\frac{-L/2-x_1}{\delta}}^{\frac{L/2-x_1}{\delta}}
\frac{t^2}{\big(t^2+2(1-\cos(\theta_1-\theta_2))\big)^{3/2}}\varphi_1(x_1+\delta t, \theta_2)dtd\theta_2\\
 \sim & \int_{0}^{2\pi}\varphi_1(x_1, \theta_1+\theta)\int_{\frac{-L/2-x_1}{\delta}}^{\frac{L/2-x_1}{\delta}}\frac{\delta t^2}{\big(t^2+2(1-\cos(\theta_1-\theta_2))\big)^{3/2}}dtd\theta_2.
\end{aligned}
$$
Set $2(1-\cos(\theta_1-\theta_2))=u$, then
$$
\begin{aligned}
\delta J_0=&\int_{\frac{-L/2-x_1}{\delta}}^{\frac{L/2-x_1}{\delta}}\frac{ \delta t^2}{\big(t^2+2(1-\cos(\theta_1-\theta_2))\big)^{3/2}}dt=\int_{\frac{-L/2-x_1}{\delta}}^{\frac{L/2-x_1}{\delta}}\frac{ \delta t^2}{\big(t^2+u\big)^{3/2}}dt\\=&\delta\left(\left.\operatorname{arsinh}
\left(\frac{t}{\sqrt{u}}\right)-{\frac{t}{\sqrt{t^2+u}}}\right)\right|_{\frac{-L/2-x_1}{\delta}}^{\frac{L/2-x_1}{\delta}},
\end{aligned}
$$
let $\frac{-L/2-x_1}{\delta}=B_1$ and $\frac{L/2-x_1}{\delta}=B_2$, one has
$$
J_0=\operatorname{arsinh}\left(\frac{B_2}{\sqrt{u}}\right)-{\frac{B_2}{\sqrt{B_2^2+u}}}
-\left(\operatorname{arsinh}\left(\frac{B_1}{\sqrt{u}}\right)-{\frac{B_1}{\sqrt{B_1^2+u}}} \right).
$$
For $x \gg 1$, $\operatorname{arsinh}(x) \sim \ln (2 x)$,
and noting that $\frac{B_2}{\sqrt{B_2^2+u}}\sim 1$, $\frac{B_1}{\sqrt{B_1^2+u}}\sim -1$, then
$$
\begin{aligned}
J_0\sim & \ln\left(\frac{2B_2}{\sqrt{u}}\right)-1-\left( -\ln \left(\frac{2\left|B_1\right|}{\sqrt{u}}\right)-(-1)\right)\\=&
\ln\left(\frac{2(L/2-x_1)}{\delta\sqrt{u}}\right)-1-\left(-\ln\left(\frac{2(L/2+x_1)}{\delta\sqrt{u}}\right)-(-1) \right)\\=&
\ln\left[ \frac{4(L^2/4 - x_1^2)}{u \delta^2} \right]-2=\ln\left[ \frac{4(L^2/4 - x_1^2)}{u} \right]-2-2\ln\delta,
\end{aligned}
$$
namely
$$
\delta J_0\sim \delta\ln\left[ \frac{4(L^2/4 - x_1^2)}{u} \right]-2\delta-2\delta\ln\delta.
$$
From \eqref{eq:070926}, we have
$$
\begin{aligned}
V_1^1\sim &\int_{0}^{2\pi}\varphi_1(x_1, \theta_2)\delta J_0 d\theta_2\\ \sim &\int_{0}^{2\pi}\varphi_1(x_1, \theta_2)\Big( \delta\ln\left[4(L^2/4 - x_1^2) \right]-\delta \ln u-2\delta-2\delta\ln\delta\Big)d\theta_2\\
=&2\pi \Big( \delta\ln\left[4(L^2/4 - x_1^2) \right]-2\delta-2\delta\ln\delta\Big)\kappa^1_0+2 \pi\delta\sum_{n \neq 0} \frac{\kappa^1_n}{|n|} e^{i n \theta_1}.
\end{aligned}
$$
We can see that the order of $V_1^1$ with respect to $\delta$ is $\mathcal{O}(\delta\ln\delta)$.
$$
\begin{aligned}
V_1^2=&\int_{S^f}\frac{(x_1-y_1)(x_2-y_2)}{|\mathbf{x}-\mathbf{y}|^3}\varphi_2(\mathbf{y}) d s(\mathbf{y})\\ = &\int_{0}^{2\pi}\int_{-L/2}^{L/2}\frac{\delta^2(x_1-y_1)(\cos\theta_1-\cos\theta_2)}{\big((x_1-y_1)^2+2\delta^2(1-\cos(\theta_1-\theta_2))\big)^{3/2}}
\varphi_2(y_1, \theta_2)dy_1d\theta_2.
\end{aligned}
$$
By variable substitution $\frac{y_1-x_1}{\delta}=t$, it holds
$$
\begin{aligned}
V_1^2=&\delta\int_{0}^{2\pi}\int_{\frac{-L/2-x_1}{\delta}}^{\frac{L/2-x_1}{\delta}}
\frac{ t (\cos\theta_2-\cos\theta_1)}{\big(t^2+2(1-\cos(\theta_1-\theta_2))\big)^{3/2}}\varphi_2(x_1+\delta t, \theta_2)dtd\theta_2\\
 \sim & \int_{0}^{2\pi}\varphi_2(x_1, \theta_2)\int_{\frac{-L/2-x_1}{\delta}}^{\frac{L/2-x_1}{\delta}}\frac{\delta t(\cos\theta_2-\cos\theta_1)}{\big(t^2+2(1-\cos(\theta_1-\theta_2))\big)^{3/2}}dtd\theta_2.
\end{aligned}
$$
Set $2(1-\cos(\theta_1-\theta_2))=u$, then
$$
\begin{aligned}
\delta J_1=&\int_{\frac{-L/2-x_1}{\delta}}^{\frac{L/2-x_1}{\delta}}\frac{ \delta t}{\big(t^2+2(1-\cos(\theta_1-\theta_2))\big)^{3/2}}dt=\int_{\frac{-L/2-x_1}{\delta}}^{\frac{L/2-x_1}{\delta}}\frac{ \delta t}{\big(t^2+u\big)^{3/2}}dt\\=&
\left.\frac{-\delta}{\sqrt{t^2+u}}\right|_{\frac{-L/2-x_1}{\delta}}^{\frac{L/2-x_1}{\delta}}=\delta^2\left(\frac{1}{\sqrt{(L / 2+x_1)^2+u \delta^2}}-\frac{1}{\sqrt{(L / 2-x_1)^2+u \delta^2}}\right),
\end{aligned}
$$
noting that
$$
\begin{aligned}
&\frac{1}{\sqrt{(L / 2-x_1)^2+u \delta^2}}=\frac{1}{(L / 2-x_1) \sqrt{1+u \delta^2 / (L / 2-x_1)^2}}\\&=\frac{1}{(L / 2-x_1)}\left(1+u \delta^2 / (L / 2-x_1)^2\right)^{-1 / 2} \sim \frac{1}{(L / 2-x_1)}\left(1-\frac{u \delta^2}{2 (L / 2-x_1)^2}\right)\\&\sim \frac{1}{(L / 2-x_1)},
\end{aligned}
$$
then
$$
\delta J_1\sim\delta^2\left(\frac{1}{(L / 2+x_1)}-\frac{1}{(L / 2-x_1)}\right)=-2\delta^2\frac{x_1}{L^2 / 4-x_1^2}.
$$
By using Fourier series expansion and the convolution theorem again, we can derive
$$
\begin{aligned}
V_1^2 \sim & -2\delta^2\frac{x_1}{L^2 / 4-x_1^2}\int_{0}^{2\pi}(\cos\theta_2-\cos\theta_1)\varphi_2(x_1, \theta_2)d\theta_2
\\ = &-2 \delta^2\pi\frac{x_1}{L^2 / 4-x_1^2}\Big(\left(\kappa^2_1+\kappa^2_{-1}\right)-2 \kappa^2_0 \cos \theta_1\Big),
\end{aligned}
$$
similarly, we can obtain
$$
\begin{aligned}
V_2^1=&\int_{S^f}\frac{(x_1-y_1)(x_2-y_2)}{|\mathbf{x}-\mathbf{y}|^3}\varphi_1(\mathbf{y}) d s(\mathbf{y})\\ \sim &-2\delta^2\frac{x_1}{L^2 / 4-x_1^2}\int_{0}^{2\pi}(\cos\theta_2-\cos\theta_1)\varphi_1(x_1, \theta_2)d\theta_2\\=&-2 \delta^2\pi\frac{x_1}{L^2 / 4-x_1^2}\Big(\left(\kappa^1_1+\kappa^1_{-1}\right)-2 \kappa^1_0 \cos \theta_1\Big),
\end{aligned}
$$
$$
\begin{aligned}
V_1^3=&\int_{S^f}\frac{(x_1-y_1)(x_3-y_3)}{|\mathbf{x}-\mathbf{y}|^3}\varphi_3(\mathbf{y}) d s(\mathbf{y})\\ \sim &-2\delta^2\frac{x_1}{L^2 / 4-x_1^2}\int_{0}^{2\pi}(\sin\theta_2-\sin\theta_1)\varphi_3(x_1, \theta_2)d\theta_2\\=&-2\pi\delta^2\frac{x_1}{L^2 / 4-x_1^2}\Big(i \left(\kappa^3_1-\kappa^3_{-1}\right)-2 \kappa^3_0 \sin \theta_1\Big),
\end{aligned}
$$
and
$$
\begin{aligned}
V_3^1=&\int_{S^f}\frac{(x_1-y_1)(x_3-y_3)}{|\mathbf{x}-\mathbf{y}|^3}\varphi_1(\mathbf{y}) d s(\mathbf{y})\\ \sim &-2\delta^2\frac{x_1}{L^2 / 4-x_1^2}\int_{0}^{2\pi}(\sin\theta_2-\sin\theta_1)\varphi_1(x_1, \theta_2)d\theta_2\\=&-2\pi\delta^2\frac{x_1}{L^2 / 4-x_1^2}\Big(i \left(\kappa^1_1-\kappa^1_{-1}\right)-2 \kappa^1_0 \sin \theta_1\Big),
\end{aligned}
$$
where $\kappa^j_{-1}$, $\kappa^j_0$ and $\kappa^j_1$, $j=1, 2, 3$ are the Fourier coefficients of $\varphi_3(x_1, \theta_2)$.
We can see that the order of $V_1^j$ and $V_j^1$, $j=2, 3$ with respect to $\delta$ is $\mathcal{O}(\delta^2)$.
$$
\begin{aligned}
V_2^2=&\int_{S^f}\frac{(x_2-y_2)^2}{|\mathbf{x}-\mathbf{y}|^3}\varphi_2(\mathbf{y}) d s(\mathbf{y})\\ = &\int_{0}^{2\pi}\int_{-L/2}^{L/2}\frac{\delta^3(\cos\theta_1-\cos\theta_2)^2}{\big((x_1-y_1)^2+2\delta^2(1-\cos(\theta_1-\theta_2))\big)^{3/2}}
\varphi_2(y_1, \theta_2)dy_1d\theta_2.
\end{aligned}
$$
By variable substitution $\frac{y_1-x_1}{\delta}=t$, it holds
$$
\begin{aligned}
V_2^2=&\delta\int_{0}^{2\pi}\int_{\frac{-L/2-x_1}{\delta}}^{\frac{L/2-x_1}{\delta}}
\frac{  (\cos\theta_1-\cos\theta_2)^2}{\big(t^2+2(1-\cos(\theta_1-\theta_2))\big)^{3/2}}\varphi_2(x_1+\delta t, \theta_2)dtd\theta_2\\
 \sim & \int_{0}^{2\pi}\varphi_2(x_1, \theta_2)\int_{\frac{-L/2-x_1}{\delta}}^{\frac{L/2-x_1}{\delta}}\frac{\delta (\cos\theta_1-\cos\theta_2)^2}{\big(t^2+2(1-\cos(\theta_1-\theta_2))\big)^{3/2}}dtd\theta_2.
\end{aligned}
$$
Set $2(1-\cos(\theta_1-\theta_2))=u$, then
$$
\begin{aligned}
\delta J_2=&\int_{\frac{-L/2-x_1}{\delta}}^{\frac{L/2-x_1}{\delta}}\frac{ \delta }{\big(t^2+2(1-\cos(\theta_1-\theta_2))\big)^{3/2}}dt=\int_{\frac{-L/2-x_1}{\delta}}^{\frac{L/2-x_1}{\delta}}\frac{ \delta }{\big(t^2+u\big)^{3/2}}dt\\=&
\left.\frac{\delta t}{u \sqrt{t^2+u}}\right|_{\frac{-L/2-x_1}{\delta}}^{\frac{L/2-x_1}{\delta}}\sim\frac{2\delta}{u}.
\end{aligned}
$$
We expand $\varphi_2\left(x_1, \theta_2\right)$ into a real valued Fourier series
$$
\varphi_2\left(x_1, \theta_2\right)=\frac{\varpi^2_0}{2}+\sum_{n=1}^{\infty}\left(\varpi^2_n \cos \left(n \theta_2\right)+\vartheta^2_n \sin \left(n \theta_2\right)\right),
$$
where
$$
\begin{aligned}
& \varpi^2_n=\frac{1}{\pi} \int_0^{2 \pi} \varphi_2\left(x_1, \theta^{\prime}\right) \cos \left(n \theta^{\prime}\right) d \theta^{\prime}, \\
& \vartheta^2_n=\frac{1}{\pi} \int_0^{2 \pi} \varphi_2\left(x_1, \theta^{\prime}\right) \sin \left(n \theta^{\prime}\right) d \theta^{\prime}.
\end{aligned}
$$
By direct calculation, we can deduce that
$$
\frac{(\cos\theta_1-\cos\theta_2)^2}{1-\cos(\theta_2-\theta_1)}=1-\cos \left(2 \theta_1\right) \cos \theta+\sin \left(2 \theta_1\right) \sin \theta,
$$
where $\theta=\theta_2-\theta_1$.
Then
$$
\int_{0}^{2\pi}\varphi_2(x_1, \theta_2)\frac{(\cos\theta_1-\cos\theta_2)^2}{1-\cos(\theta_2-\theta_1)}d\theta_2= \pi\left(\varpi^2_0-\varpi^2_1 \cos \theta_1+\vartheta^2_1 \sin \theta_1\right).
$$
Thus
$$
\begin{aligned}
V_2^2 & \sim\delta\int_{0}^{2\pi}\varphi_2(x_1, \theta_2)\frac{(\cos\theta_1-\cos\theta_2)^2}{1-\cos(\theta_2-\theta_1)}d\theta_2
\\&=\pi\delta\left(\varpi^2_0-\varpi^2_1 \cos \theta_1+\vartheta^2_1 \sin \theta_1\right).
\end{aligned}
$$
Similarly, we have
$$
\begin{aligned}
V_2^3&=\int_{S^f}\frac{(x_2-y_2)(x_3-y_3)}{|\mathbf{x}-\mathbf{y}|^3}\varphi_3(\mathbf{y}) d s(\mathbf{y})
\\&\sim\delta\int_{0}^{2\pi}\varphi_3(x_1, \theta_2)\frac{(\cos\theta_1-\cos\theta_2)(\sin\theta_1-\sin\theta_2)}{1-\cos(\theta_2-\theta_1)}d\theta_2,
\end{aligned}
$$
$$
\begin{aligned}
V_3^2&=\int_{S^f}\frac{(x_2-y_2)(x_3-y_3)}{|\mathbf{x}-\mathbf{y}|^3}\varphi_2(\mathbf{y}) d s(\mathbf{y})
\\&\sim\delta\int_{0}^{2\pi}\varphi_2(x_1, \theta_2)\frac{(\cos\theta_1-\cos\theta_2)(\sin\theta_1-\sin\theta_2)}{1-\cos(\theta_2-\theta_1)}d\theta_2,
\end{aligned}
$$
and
$$
\begin{aligned}
V_3^3&=\int_{S^f}\frac{(x_3-y_3)^2}{|\mathbf{x}-\mathbf{y}|^3}\varphi_3(\mathbf{y}) d s(\mathbf{y})
\\&\sim \delta\int_{0}^{2\pi}\varphi_3(x_1, \theta_2)\frac{(\sin\theta_1-\sin\theta_2)^2}{1-\cos(\theta_2-\theta_1)}d\theta_2.
\end{aligned}
$$
By direct calculation, we can obtain
$$
\frac{(\cos\theta_1-\cos\theta_2)(\sin\theta_1-\sin\theta_2)}{1-\cos(\theta_2-\theta_1)}=-\sin \left(\theta_2+\theta_1\right),
$$
and
$$
\frac{(\sin\theta_1-\sin\theta_2)^2}{1-\cos(\theta_2-\theta_1)}=\left(1+\cos \left(\theta_2+\theta_1\right)\right).
$$
Thus
$$
V_2^3\sim-\pi\delta\left(\varpi^3_1 \sin \theta_1+\vartheta^3_1 \cos \theta_1\right),
$$
$$
V_3^2\sim-\pi\delta\left(\varpi^2_1 \sin \theta_1+\vartheta^2_1 \cos \theta_1\right),
$$
and
$$
V_3^3\sim\pi\delta\left(\varpi^3_0+\varpi^3_1 \cos \theta_1-\vartheta^3_1 \sin \theta_1\right),
$$
where $\varpi^2_1$, $\vartheta^2_1$ and $\varpi^3_i$, $i=0, 1$ and $\vartheta^3_1$ are the  real valued Fourier coefficients of $\varphi_2(x_1, \theta_2)$ and $\varphi_3(x_1, \theta_2)$, respectively. We can see that the orders of $V_2^j$ and $V_3^j$, $j=2, 3$ with respect to $\delta$ are $\mathcal{O}(\delta)$. From the above derivation, we can see that only $T_i, i=1, 2, 3$ and $V_1^1$ contain the order of $\delta\ln\delta$, while the other terms are higher-order. Combining \eqref{13588}, we can obtain \eqref{eq:13520927}. For any $\boldsymbol{\varphi}\in\mathcal{H}^*$, one has
$$
\begin{aligned}
&\|\boldsymbol{\varphi}\|_{\mathcal{H}^*}^2 =- \langle\boldsymbol{\varphi}, \mathbf{S}_D[\boldsymbol{\varphi}]\rangle = -\int_{\partial D}\boldsymbol{\varphi} \cdot \mathbf{S}_{D}[\boldsymbol{\varphi}] ds\sim-\int_{\partial D}\boldsymbol{\varphi}\cdot \Big(\alpha_1\mathbf{T}[\boldsymbol{\varphi}]+\alpha_2\mathbf{V}[\boldsymbol{\varphi}]\Big)ds\\&=-\int_{\partial D}\Big(\alpha_1\left(\varphi_1 T_1+\varphi_2 T_2+\varphi_3 T_3\right)+\alpha_2\left(\varphi_1 V_1+\varphi_2 V_2+\varphi_3 V_3\right)\Big)d s\\&\sim\int_{S^f}4 \pi(\delta \ln \delta)\Big(\left(\alpha_1+\alpha_2\right) \varphi_1 \kappa_0^1+\alpha_1 \varphi_2 \kappa_0^2+\alpha_1 \varphi_3 \kappa_0^3\Big)ds\\
&=8 \pi^2\left(\delta^2 \ln \delta\right) \int_{-L / 2}^{L / 2}  \Big(\left(\alpha_1+\alpha_2\right)(\kappa_0^1)^2+\alpha_1(\kappa_0^2)^2+\alpha_1(\kappa_0^3)^2\Big) d x_1,
\end{aligned}
$$
then the first equation in \eqref{eq:15510927} holds. Due to the invertibility of $\mathbf{S}_D: \mathcal{H}^*\rightarrow \mathcal{H},$ we know that for any $\mathbf{u}\in \mathcal{H}$, there exists $\boldsymbol{\varphi^{\prime}}\in \mathcal{H}^*$, so that $\mathbf{S}_D[\boldsymbol{\varphi^{\prime}}]=\mathbf{u},$
 then
$$
\|\mathbf{u}\|_{\mathcal{H}}=\|\mathbf{S}_D[\boldsymbol{\varphi^{\prime}}]\|_{\mathcal{H}}
=\|\boldsymbol{\varphi^{\prime}}\|_{\mathcal{H}^*}=\mathcal{O}(\delta |\ln \delta|^{1/2}).
$$
This completes the proof.
\end{proof}

\begin{lem}\label{le:1}
For $\omega \ll 1$, the following asymptotic results hold.\\
(i) The single layer potential operator $\mathbf{S}_{D}^{\omega}:\mathcal{H}^* \rightarrow \mathcal{H}$ can be expanded as follows:

\begin{equation}\label{eq:123}
\mathbf{S}_{D}^{\omega}=\mathbf{S}_{D}+\omega \mathbf{S}_{D, 1}+\omega^2 \mathbf{S}_{D, 2}+\mathcal{O}(\omega^3),
\end{equation}
where
\begin{equation}\label{S1}
\mathbf{S}_{D, 1}[\boldsymbol{\varphi}](\mathbf{x})=\alpha_3\int_{\partial D}\boldsymbol{\varphi}(\mathbf{y}) d s(\mathbf{y}),
\end{equation}
with
\begin{equation}\label{al3}
\alpha_3=-\frac{\rm{i}}{12\pi}\left(\frac{2}{{\mu}^{3/2}}+\frac{1}{{(\lambda+2 \mu)}^{3/2}}\right),
\end{equation}
and
\begin{equation}\label{S2}
\mathbf{S}_{D, 2}[\boldsymbol{\varphi}](\mathbf{x})=\int_{\partial D}\mathbf{\Gamma}_2(\mathbf{x}-\mathbf{y})\boldsymbol{\varphi}(\mathbf{y}) d s(\mathbf{y}),
\end{equation}
\begin{equation}\label{G2}
\mathbf{\Gamma}_2(\mathbf{x}-\mathbf{y})=\alpha_4|\boldsymbol{x}-\boldsymbol{y}|\mathbf{I}+\alpha_5\frac{(\boldsymbol{x}-\boldsymbol{y})
(\boldsymbol{x}-\boldsymbol{y})^T}{|\boldsymbol{x}-\boldsymbol{y}|},
\end{equation}
with
\begin{equation}\label{a45}
\alpha_4=\frac{1}{32\pi}\left(\frac{3}{\mu^2}+\frac{1}{(\lambda+2 \mu)^2}\right),\quad
\alpha_5=-\frac{1}{32\pi}\left(\frac{1}{\mu^2}-\frac{1}{(\lambda+2\mu)^2}\right).
\end{equation}
The following estimate holds:
\begin{equation}\label{2123}
\left\|\mathbf{S}_{D}\right\|_{\mathcal{L}\left(\mathcal{H}^*, \mathcal{H}\right)}=1,\quad
\left\|\mathbf{S}_{D, i}\right\|_{\mathcal{L}\left(\mathcal{H}^*, \mathcal{H}\right)}=\mathcal{O}\left(1\right), i=1, 2.
\end{equation}

(ii)The operator $(\mathbf{S}_{D}^{\omega})^{-1}:\mathcal{H}\rightarrow \mathcal{H}^*$ can be expanded as follows:
\begin{equation}\label{1303}
\left(\mathbf{S}_{D}^{\omega}\right)^{-1}=\mathbf{S}_{D}^{-1}+\omega \mathbf{B}_{D, 1}+\omega^2 \mathbf{B}_{D, 2}+\mathcal{O}(\omega^3),
\end{equation}
where
$$
\begin{aligned}
\mathbf{B}_{D, 1} & =-\mathbf{S}_{D}^{-1} \mathbf{S}_{D, 1} \mathbf{S}_{D}^{-1}, \\
\mathbf{B}_{D, 2} & =-\mathbf{S}_{D}^{-1} \mathbf{S}_{D, 2} \mathbf{S}_{D}^{-1}+\mathbf{S}_{D}^{-1} \mathbf{S}_{D, 1} \mathbf{S}_{D}^{-1} \mathbf{S}_{D, 1} \mathbf{S}_{D}^{-1},
\end{aligned}
$$
the following estimate holds:
\begin{equation}\label{1304}
\left\|\mathbf{S}_{D}^{-1}\right\|_{\mathcal{L}\left(\mathcal{H}, \mathcal{H}^*\right)}=1.
\end{equation}

(iii) The Neumann-Poincar\'e operator $\left(\mathbf{K}_{D}^{\omega}\right)^*:\mathcal{H}^* \rightarrow \mathcal{H}^*$ can be expanded as follows:
\begin{equation}\label{1407}
\left(\mathbf{K}_{D}^{\omega}\right)^*=\mathbf{K}_{D}^*+\omega^2\mathbf{K}_{D, 2}+\mathcal{O}(\omega^3),
\end{equation}
where
\begin{equation}\label{K2}
\mathbf{K}_{D, 2}[\boldsymbol{\varphi}](\mathbf{x})=\int_{\partial D}\frac{\partial \mathbf\Gamma_2(\mathbf{x}-\mathbf{y})}{\partial \boldsymbol{\nu}_{\mathbf{x}}}\boldsymbol{\varphi}(\mathbf{y}) d s(\mathbf{y}),
\end{equation}
the following estimate holds:
\begin{equation}\label{1408}
\left\|\mathbf{K}_{D}^*\right\|_{\mathcal{L}\left(\mathcal{H}^*, \mathcal{H}^*\right)}=\frac{1}{2},
\quad \left\|\mathbf{K}_{D, 2}\right\|_{\mathcal{L}\left(\mathcal{H}^*, \mathcal{H}^*\right)}=\mathcal{O}(1).
\end{equation}
\end{lem}
\begin{proof}
(1) We first prove \eqref{2123}. Owing to $\left\|\mathbf{S}_D\right\|_{\mathcal{L}\left(\mathcal{H}^*, \mathcal{H}\right)}=\sup\limits_{\boldsymbol{\varphi}\neq 0} \frac{\left\|\mathbf{S}_D[\boldsymbol{\varphi}]\right\|_{\mathcal{H}}}{\|\boldsymbol{\varphi}\|_{\mathcal{H}^*}}$, and $\left\|\mathbf{S}_D[\boldsymbol{\varphi}]\right\|^2_{\mathcal{H}}=(\mathbf{S}_D[\boldsymbol{\varphi}], \mathbf{S}_D[\boldsymbol{\varphi}])_{\mathcal{H}}=-\left\langle\boldsymbol{\varphi}, \mathbf{S}_D[\boldsymbol{\varphi}]\right\rangle=(\boldsymbol{\varphi}, \boldsymbol{\varphi})_{\mathcal{H}^*}=\left\|\boldsymbol{\varphi}\right\|^2_{\mathcal{H}^*}$. Then $\left\|\mathbf{S}_D\right\|_{\mathcal{L}\left(\mathcal{H}^*, \mathcal{H}\right)}=\sup\limits_{\boldsymbol{\varphi}\neq 0}\frac{\left\|\boldsymbol{\varphi}\right\|_{\mathcal{H}^*}}{\|\boldsymbol{\varphi}\|_{\mathcal{H}^*}}\\=1.$ From \eqref{eq:15510927}, we have
$$
\begin{aligned}
\left\|\mathbf{S}_{D, i}\right\|_{\mathcal{L}\left(\mathcal{H}^*, \mathcal{H}\right)}=&\sup\limits_{\boldsymbol{\varphi}\neq 0} \frac{\left\|\mathbf{S}_{D, i}[\boldsymbol{\varphi}]\right\|_{\mathcal{H}}}{\|\boldsymbol{\varphi}\|_{\mathcal{H}^*}}=\mathcal{O}(1).
\end{aligned}
$$

From \eqref{gamma} and the definition of $\mathbf{S}_{D}^{\omega}$, we can immediately obtain \eqref{eq:123}.

(2) Since $\left\|\mathbf{S}_D^{-1}\right\|_{\mathcal{L}\left(\mathcal{H}, \mathcal{H}^*\right)}=\sup\limits_{\boldsymbol{\varphi}\neq 0} \frac{\left\|\mathbf{S}_D^{-1}[\boldsymbol{\varphi}]\right\|_{\mathcal{H}^*}}{\|\boldsymbol{\varphi}\|_{\mathcal{H}}}$, and $$\left\|\mathbf{S}_D^{-1}[\boldsymbol{\varphi}]\right\|^2_{\mathcal{H}^*}=(\mathbf{S}_D^{-1}[\boldsymbol{\varphi}], \mathbf{S}_D^{-1}[\boldsymbol{\varphi}])_{\mathcal{H}^*}=-\left\langle\mathbf{S}_D^{-1}[\boldsymbol{\varphi}], \boldsymbol{\varphi}\right\rangle=(\boldsymbol{\varphi}, \boldsymbol{\varphi})_{\mathcal{H}}=\left\|\boldsymbol{\varphi}\right\|^2_{\mathcal{H}}.$$
Then $\left\|\mathbf{S}_D^{-1}\right\|_{\mathcal{L}\left(\mathcal{H}, \mathcal{H}^*\right)}=\sup\limits_{\boldsymbol{\varphi}\neq 0}\frac{\left\|\boldsymbol{\varphi}\right\|_{\mathcal{H}}}{\|\boldsymbol{\varphi}\|_{\mathcal{H}}}=1.$ By \eqref{eq:123}, we have $$\mathbf{S}_{D}^{\omega}=\mathbf{S}_{D}\Big(I+\omega \mathbf{S}_{D}^{-1}\mathbf{S}_{D, 1}+\omega^2 \mathbf{S}_{D}^{-1}\mathbf{S}_{D, 2}+\mathcal{O}(\omega^3)\Big),$$
therefore
$$
\begin{aligned}
(\mathbf{S}_{D}^{\omega})^{-1}=&\Big(I+\omega \mathbf{S}_{D}^{-1}\mathbf{S}_{D, 1}+\omega^2 \mathbf{S}_{D}^{-1}\mathbf{S}_{D, 2}+\mathcal{O}(\omega^3)\Big)^{-1}\mathbf{S}_{D}^{-1}\\=&\Big(I-\omega \mathbf{S}_{D}^{-1}\mathbf{S}_{D, 1}-\omega^2 \mathbf{S}_{D}^{-1}\mathbf{S}_{D, 2}+\omega^2\mathbf{S}_{D}^{-1}\mathbf{S}_{D, 1}\mathbf{S}_{D}^{-1}\mathbf{S}_{D, 1}+\mathcal{O}(\omega^3)\Big)\mathbf{S}_{D}^{-1}\\=&\mathbf{S}_{D}^{-1}-\omega \mathbf{S}_{D}^{-1}\mathbf{S}_{D, 1}\mathbf{S}_{D}^{-1}-\omega^2 \mathbf{S}_{D}^{-1}\mathbf{S}_{D, 2}\mathbf{S}_{D}^{-1}+\omega^2\mathbf{S}_{D}^{-1}\mathbf{S}_{D, 1}\mathbf{S}_{D}^{-1}\mathbf{S}_{D, 1}\mathbf{S}_{D}^{-1}\\&+\mathcal{O}(\omega^3).
\end{aligned}
$$

(3) Noting that
$$
\begin{aligned}
(\mathbf{K}_D^*[\boldsymbol{\varphi}], \boldsymbol{\psi})_{\mathcal{H}^*}=&-\left\langle\mathbf{K}_D^*[\boldsymbol{\varphi}], \mathbf{S}_D[\boldsymbol{\psi}]\right\rangle=-\left\langle\boldsymbol{\varphi}, \mathbf{K}_D\mathbf{S}_D[\boldsymbol{\psi}]\right\rangle=-\left\langle\boldsymbol{\varphi}, \mathbf{S}_D\mathbf{K}_D^*[\boldsymbol{\psi}]\right\rangle\\=&(\boldsymbol{\varphi}, \mathbf{K}_D^*[\boldsymbol{\psi}])_{\mathcal{H}^*},
\end{aligned}
$$
that is to say, $\mathbf{K}_D^*$ is a self-adjoint operator in the inner product space $\mathcal{H}^*$. So $\left\|\mathbf{K}_D^*\right\|_{\mathcal{L}\left(\mathcal{H}^*, \mathcal{H}^*\right)}=\sup\limits_j\left|\beta_j\right|=\frac{1}{2},$ where $\beta_j$ is the eigenvalue of $\mathbf{K}_D^*$.
From \eqref{eq:15510927}, we have
$$
\begin{aligned}
\left\|\mathbf{K}_{D, 2}\right\|_{\mathcal{L}\left(\mathcal{H}^*, \mathcal{H}^*\right)}=&\sup\limits_{\boldsymbol{\varphi}\neq 0} \frac{\left\|\mathbf{K}_{D, 2}[\boldsymbol{\varphi}]\right\|_{\mathcal{H}^*}}{\|\boldsymbol{\varphi}\|_{\mathcal{H}^*}}=\mathcal{O}(1).
\end{aligned}
$$

By \eqref{gamma} and the definition of $\left(\mathbf{K}_{D}^{\omega}\right)^*$, we can immediately obtain \eqref{1408}.
\end{proof}

\begin{lem}\label{leS}
Let $\boldsymbol{\varphi}\in\mathcal{H}^*$, then the following asymptotic result hold:
\begin{equation}
\mathbf{S}_{D}[\boldsymbol{\varphi}]=\mathbf{S}_{00}[\boldsymbol{\varphi}]
+\delta\mathbf{S}_{01}[\boldsymbol{\varphi}]+\delta^2\mathbf{S}_{02}[\boldsymbol{\varphi}]+\mathcal{O}(\delta^3), \quad \mathbf{x}\in \partial D,
\end{equation}
where
\begin{equation}
\begin{aligned}
\mathbf{S}_{00}[\boldsymbol{\varphi}](\mathbf{x})=&\chi(S^a\cup S^b)\int_{S^f} \mathbf\Gamma(\mathbf{x}-\mathbf{y})\boldsymbol{\varphi}(\mathbf{y}) d s(\mathbf{y})
+\sum_{\substack{i,j=1}}^{2}\chi(\Gamma_i)\int_{\Gamma_j}\mathbf\Gamma(\mathbf{x}-\mathbf{y})\boldsymbol{\varphi}(\mathbf{y})d s(\mathbf{y})
\\&+\chi(S^a\cup \iota_\delta(P))\int_{S^a}\mathbf\Gamma(\mathbf{x}-\mathbf{y})\boldsymbol{\varphi}(\mathbf{y})d s(\mathbf{y})\\&+\chi(S^b\cup \iota_\delta(Q))\int_{S^b} \mathbf\Gamma(\mathbf{x}-\mathbf{y})\boldsymbol{\varphi}(\mathbf{y})d s(\mathbf{y}),
\end{aligned}
\end{equation}

\begin{equation}
\begin{aligned}
\mathbf{S}_{01}[\boldsymbol{\varphi}](\mathbf{x})=&\chi_P\delta\bigg(\int_{S_{1}^a}\boldsymbol{\Gamma}(\mathbf{x}-P)\tilde{\boldsymbol{\varphi}}(\tilde{\mathbf{y}})d s(\tilde{\mathbf{y}})+\mathcal{O}\Big(\frac{\delta}{|\mathbf{x}-P|^2}\Big)\bigg)\\
&+\chi_Q\delta\bigg(\int_{S_{1}^a}\boldsymbol{\Gamma}(\mathbf{x}-Q)\tilde{\boldsymbol{\varphi}}(\tilde{\mathbf{y}})d s(\tilde{\mathbf{y}})+\mathcal{O}\Big(\frac{\delta}{|\mathbf{x}-Q|^2}\Big)\bigg),
\end{aligned}
\end{equation}
and
\begin{equation}
\begin{aligned}
\mathbf{S}_{02}[\boldsymbol{\varphi}](\mathbf{x})=&\chi(S^a)\int_{S_{1}^b}\boldsymbol{\Gamma}(\mathbf{x}-Q)
\tilde{\boldsymbol{\varphi}}(\tilde{\mathbf{y}})ds(\tilde{\mathbf{y}})+\chi(S^b)\int_{S_{1}^a}\boldsymbol{\Gamma}(\mathbf{x}-P)
\tilde{\boldsymbol{\varphi}}(\tilde{\mathbf{y}})ds(\tilde{\mathbf{y}}).
\end{aligned}
\end{equation}
Here $\chi_P:=\chi\left(S^f \backslash \iota_\delta(P)\right)$ and $\chi_Q:=\chi\left(S^f \backslash \iota_\delta(Q)\right)$.
\end{lem}
\begin{proof}
For $\mathbf{x}\in \partial D$, from the definition of $\mathbf{S}_{D}[\boldsymbol{\varphi}]$, it follows that
$$
\begin{aligned}
&\mathbf{S}_{D}[\boldsymbol{\varphi}](\mathbf{x})=\int_{S^f} \boldsymbol{\Gamma}(\mathbf{x}-\mathbf{y}) \boldsymbol{\varphi}(\mathbf{y}) d s(\mathbf{y})+\int_{S^a\cup S^b} \boldsymbol{\Gamma}(\mathbf{x}-\mathbf{y}) \boldsymbol{\varphi}(\mathbf{y}) d s(\mathbf{y})\\=&\chi(S^f)\int_{S^f} \boldsymbol{\Gamma}(\mathbf{x}-\mathbf{y}) \boldsymbol{\varphi}(\mathbf{y}) d s(\mathbf{y})+\chi(S^a\cup S^b)\int_{S^f} \boldsymbol{\Gamma}(\mathbf{x}-\mathbf{y}) \boldsymbol{\varphi}(\mathbf{y}) d s(\mathbf{y})\\&+\chi(S^f)\int_{S^a\cup S^b} \boldsymbol{\Gamma}(\mathbf{x}-\mathbf{y}) \boldsymbol{\varphi}(\mathbf{y}) d s(\mathbf{y})+\chi(S^a\cup S^b)\int_{S^a\cup S^b} \boldsymbol{\Gamma}(\mathbf{x}-\mathbf{y}) \boldsymbol{\varphi}(\mathbf{y}) d s(\mathbf{y}).
\end{aligned}
$$
For $\mathbf{x}\in S^a$ and $\mathbf{y}\in S^b$,
\begin{equation}
\mathbf{x}-\mathbf{y}=\mathbf{x}-Q-(\mathbf{y}-Q)=\mathbf{x}-Q-\delta(\tilde{\mathbf{y}}-Q),
\end{equation}
\begin{equation}
|\mathbf{x}-\mathbf{y}|=|\mathbf{x}-Q|-\frac{\delta}{|\mathbf{x}-Q|}(\tilde{\mathbf{y}}-Q, \mathbf{x}-Q)+O(\delta^2),
\end{equation}
\begin{equation}
\boldsymbol{\Gamma}(\mathbf{x}-\mathbf{y})=\frac{\alpha_1}{|\boldsymbol{x}-Q|}\mathbf{I}+\alpha_2 \frac{(\boldsymbol{x}-Q)
(\boldsymbol{x}-Q)^T}{|\boldsymbol{x}-Q|^3}+O(\delta),
\end{equation}
then
\begin{equation}
\chi(S^a)\int_{S^b}\boldsymbol{\Gamma}(\mathbf{x}-\mathbf{y})\boldsymbol{\varphi}(\mathbf{y})ds(\mathbf{y})=\chi(S^a)\delta^2\int_{S_1^b}
\boldsymbol{\Gamma}(\mathbf{x}-Q)
\tilde{\boldsymbol{\varphi}}(\tilde{\mathbf{y}})ds(\tilde{\mathbf{y}})+O(\delta^3),
\end{equation}
similarly,
\begin{equation}
\chi(S^b)\int_{S^a}\boldsymbol{\Gamma}(\mathbf{x}-\mathbf{y})\boldsymbol{\varphi}(\mathbf{y})ds(\mathbf{y})=\chi(S^b)\delta^2\int_{S_1^a}
\boldsymbol{\Gamma}(\mathbf{x}-P)
\tilde{\boldsymbol{\varphi}}(\tilde{\mathbf{y}})ds(\tilde{\mathbf{y}})+O(\delta^3).
\end{equation}
For $\mathbf{x}\in S^f$ and $\mathbf{y}\in S^a$, we have
\begin{equation}
\begin{aligned}
&\chi(S^f)\int_{S^a}\boldsymbol{\Gamma}(\mathbf{x}-\mathbf{y})\boldsymbol{\varphi}(\mathbf{y})d s(\mathbf{y})\\
=&\chi\left(S^f \backslash \iota_\delta(P)\right)\delta^2\bigg(\int_{S_{1}^a}\boldsymbol{\Gamma}(\mathbf{x}-P)\tilde{\boldsymbol{\varphi}}(\tilde{\mathbf{y}})d s(\tilde{\mathbf{y}})+\mathcal{O}\Big(\frac{\delta}{|\mathbf{x}-P|^2}\Big)\bigg)\\&+\chi(\iota_\delta(P))\int_{S^a}\boldsymbol{\Gamma}(\mathbf{x}-\mathbf{y})\boldsymbol{\varphi}(\mathbf{y})d s(\mathbf{y}),
\end{aligned}
\end{equation}
similarly, for $\mathbf{x}\in S^f$ and $\mathbf{y}\in S^b$, a similar result holds. This completes the proof.
\end{proof}

\begin{lem}
{\rm (1) For $\mathbf{S}_{D, 1}[\boldsymbol{\varphi}]$ defined in \eqref{S1}, the following asymptotic result hold:
\begin{equation}\label{SD1}
\mathbf{S}_{D, 1}[\boldsymbol{\varphi}]=\mathbf{S}_{10}[\boldsymbol{\varphi}]
+\delta^2\mathbf{S}_{11}[\boldsymbol{\varphi}]+\mathcal{O}(\delta^3), \quad \mathbf{x}\in \partial D,
\end{equation}
where
\begin{equation}
\begin{aligned}
\mathbf{S}_{10}[\boldsymbol{\varphi}](\mathbf{x})=&\chi(S^a\cup S^b)\int_{S^f}\alpha_3\boldsymbol{\varphi}(\mathbf{y}) d s(\mathbf{y})
+\sum_{\substack{i,j=1}}^{2}\chi(\Gamma_i)\int_{\Gamma_j}\alpha_3\boldsymbol{\varphi}(\mathbf{y})d s(\mathbf{y}),
\end{aligned}
\end{equation}
and
\begin{equation}
\begin{aligned}
\mathbf{S}_{11}[\boldsymbol{\varphi}](\mathbf{x})=&\chi(\partial D)\int_{S_1^a\cup S_1^b}\alpha_3\tilde{\boldsymbol{\varphi}}(\tilde{\mathbf{y}}) d s(\tilde{\mathbf{y}}).
\end{aligned}
\end{equation}
(2) For $\mathbf{S}_{D, 2}[\boldsymbol{\varphi}]$ defined in \eqref{S2}, the following asymptotic result hold:
\begin{equation}\label{SD2}
\mathbf{S}_{D, 2}[\boldsymbol{\varphi}]=\mathbf{S}_{20}[\boldsymbol{\varphi}]
+\delta\mathbf{S}_{21}[\boldsymbol{\varphi}]+\delta^2\mathbf{S}_{22}[\boldsymbol{\varphi}]+\mathcal{O}(\delta^3), \quad \mathbf{x}\in \partial D,
\end{equation}
where
\begin{equation}
\begin{aligned}
\mathbf{S}_{20}[\boldsymbol{\varphi}](\mathbf{x})=&\chi(S^a\cup S^b)\int_{S^f} \mathbf\Gamma_2(\mathbf{x}-\mathbf{y})\boldsymbol{\varphi}(\mathbf{y}) d s(\mathbf{y})+\sum_{\substack{i,j=1}}^{2}\chi(\Gamma_i)\int_{\Gamma_j}\mathbf\Gamma_2(\mathbf{x}-\mathbf{y})\boldsymbol{\varphi}(\mathbf{y})d s(\mathbf{y})
\\&+\chi(S^a\cup \iota_\delta(P))\int_{S^a}\mathbf\Gamma_2(\mathbf{x}-\mathbf{y})\boldsymbol{\varphi}(\mathbf{y})d s(\mathbf{y})\\&+\chi(S^b\cup \iota_\delta(Q))\int_{S^b} \mathbf\Gamma_2(\mathbf{x}-\mathbf{y})\boldsymbol{\varphi}(\mathbf{y})d s(\mathbf{y}),
\end{aligned}
\end{equation}
\begin{equation}
\begin{aligned}
\mathbf{S}_{21}[\boldsymbol{\varphi}](\mathbf{x})=&\chi_P\delta\bigg(\int_{S_{1}^a}\mathbf\Gamma_2(\mathbf{x}-P)\tilde{\boldsymbol{\varphi}}(\tilde{\mathbf{y}}) d s(\tilde{\mathbf{y}})+O\Big(\frac{\delta^2}{|\mathbf{x}-P|}\Big)\bigg)\\
&+\chi_Q\delta\bigg(\int_{S_{1}^b}\mathbf\Gamma_2(\mathbf{x}-Q)\tilde{\boldsymbol{\varphi}}(\tilde{\mathbf{y}}) d s(\tilde{\mathbf{y}})+O\Big(\frac{\delta^2}{|\mathbf{x}-Q|}\Big)\bigg),
\end{aligned}
\end{equation}
and
\begin{equation}
\begin{aligned}
\mathbf{S}_{22}[\boldsymbol{\varphi}](\mathbf{x})=&\chi(S^a)\int_{S_{1}^b}\mathbf\Gamma_2(\mathbf{x}-Q)
\tilde{\boldsymbol{\varphi}}(\tilde{\mathbf{y}})d s(\tilde{\mathbf{y}})+\chi(S^b)\int_{S_{1}^a}\mathbf\Gamma_2(\mathbf{x}-P)
\tilde{\boldsymbol{\varphi}}(\tilde{\mathbf{y}})d s(\tilde{\mathbf{y}}),
\end{aligned}
\end{equation}
with $\mathbf\Gamma_2$ is given in \eqref{G2}. Here $\chi_P:=\chi\left(S^f \backslash \iota_\delta(P)\right)$ and $\chi_Q:=\chi\left(S^f \backslash \iota_\delta(Q)\right)$.\\
(3) For $\mathbf{K}_{D, 2}^*[\boldsymbol{\varphi}]$ defined in \eqref{K2}, the following asymptotic result hold:
\begin{equation}\label{KD2}
\mathbf{K}_{D, 2}^*[\boldsymbol{\varphi}](\mathbf{x})=\mathbf{K}_{10}[\boldsymbol{\varphi}](\mathbf{x})+\delta \mathbf{K}_{11}[\boldsymbol{\varphi}](\mathbf{x})+\delta^2 \mathbf{K}_{12}[\boldsymbol{\varphi}](\mathbf{x})+\mathcal{O}\left(\delta^3\right),
\end{equation}
where
\begin{equation}
\begin{aligned}
\mathbf{K}_{10}[\boldsymbol{\varphi}](\mathbf{x})=&\chi(S^a\cup S^b)\int_{S^f} \frac{\partial\mathbf\Gamma_2(\mathbf{x}-\mathbf{y})}{\partial\boldsymbol{\nu}_{\mathbf{x}}}\boldsymbol{\varphi}(\mathbf{y}) d s(\mathbf{y})
\\&+\Big(\lambda\alpha_4+\big(3\lambda+2\mu\big)\alpha_5\Big)\sum_{\substack{i,j=1}}^{2}\chi(\Gamma_i)\int_{\Gamma_j}
\frac{\boldsymbol{\nu}_{\mathbf{x}}(\mathbf{x}-\mathbf{y})^T}{|\mathbf{x}-\mathbf{y}|}\boldsymbol{\varphi}(\mathbf{y})d s(\mathbf{y})
\\&+\mu(\alpha_4+\alpha_5)\sum_{\substack{i,j=1}}^{2}\chi(\Gamma_i)\int_{\Gamma_j}
\left(\frac{(\mathbf{x}-\mathbf{y})\boldsymbol{\nu}_{\mathbf{x}}^T}{|\mathbf{x}-\mathbf{y}|}
+\frac{(\mathbf{x}-\mathbf{y},\boldsymbol{\nu}_{\mathbf{x}})}{|\mathbf{x}-\mathbf{y}|}\right)\boldsymbol{\varphi}(\mathbf{y})d s(\mathbf{y})\\
&-2\mu\alpha_5\sum_{\substack{i,j=1}}^{2}\chi(\Gamma_i)\int_{\Gamma_j}
\frac{(\mathbf{x}-\mathbf{y})(\mathbf{x}-\mathbf{y})^T\boldsymbol{\nu}_{\mathbf{x}}(\mathbf{x}-\mathbf{y})^T}{|\mathbf{x}-\mathbf{y}|^3}
\boldsymbol{\varphi}(\mathbf{y})d s(\mathbf{y})
\\&+\chi(S^a\cup\iota_\delta(P))\int_{S^a}\frac{\partial\mathbf\Gamma_2(\mathbf{x}-\mathbf{y})}{\partial\boldsymbol{\nu}_{\mathbf{x}}}\boldsymbol{\varphi}(\mathbf{y})d s(\mathbf{y})\\&+\chi(S^b\cup\iota_\delta(Q))\int_{S^b} \frac{\partial\mathbf\Gamma_2(\mathbf{x}-\mathbf{y})}{\partial\boldsymbol{\nu}_{\mathbf{x}}}\boldsymbol{\varphi}(\mathbf{y})d s(\mathbf{y}).
\end{aligned}
\end{equation}
The matrix type operator $\mathbf{K}_{11}$ and $\mathbf{K}_{12}$ admit the following forms:
\begin{equation}
\begin{aligned}
\mathbf{K}_{11}[\boldsymbol{\varphi}](\mathbf{x})=&\Big(\lambda\alpha_4+\big(3\lambda+2\mu\big)\alpha_5\Big)\delta\chi_P\Big(\int_{S_{1}^a}
\frac{\boldsymbol{\nu}_{\mathbf{x}}(\mathbf{x}-P)^T}{|\mathbf{x}-P|}\tilde{\boldsymbol{\varphi}}(\tilde{\mathbf{y}})+O(\frac{\delta}{|\mathbf{x}-P|})\Big)
\\&+\Big(\lambda\alpha_4+\big(3\lambda+2\mu\big)\alpha_5\Big)\delta\chi_Q\Big(\int_{S_{1}^a}
\frac{\boldsymbol{\nu}_{\mathbf{x}}(\mathbf{x}-Q)^T}{|\mathbf{x}-Q|}\tilde{\boldsymbol{\varphi}}(\tilde{\mathbf{y}})+O(\frac{\delta}{|\mathbf{x}-Q|})\Big)
\\&+\mu(\alpha_4+\alpha_5)\delta\chi_P\int_{S_{1}^a}
\Big(\frac{(\mathbf{x}-P)\boldsymbol{\nu}_{\mathbf{x}}^T}{|\mathbf{x}-P|}
+\frac{\delta(1-(\tilde{\mathbf{y}}-P,\boldsymbol{\nu}_{\mathbf{x}}))}{|\mathbf{x}-P|}\Big)\tilde{\boldsymbol{\varphi}}(\tilde{\mathbf{y}})\\
&+\mu(\alpha_4+\alpha_5)\delta\chi_Q\int_{S_{1}^b}
\Big(\frac{(\mathbf{x}-Q)\boldsymbol{\nu}_{\mathbf{x}}^T}{|\mathbf{x}-Q|}
+\frac{\delta(1-(\tilde{\mathbf{y}}-Q,\boldsymbol{\nu}_{\mathbf{x}}))}{|\mathbf{x}-Q|}\Big)\tilde{\boldsymbol{\varphi}}(\tilde{\mathbf{y}})
\\&-2\mu\alpha_5\delta\chi_P\int_{S_{1}^a}
\frac{(\mathbf{x}-P)(\mathbf{x}-P)^T\boldsymbol{\nu}_{\mathbf{x}}(\mathbf{x}-P)^T}{|\mathbf{x}-P|^3}
\tilde{\boldsymbol{\varphi}}(\tilde{\mathbf{y}})
\\&-2\mu\alpha_5\delta\chi_Q\int_{S_{1}^b}
\frac{(\mathbf{x}-Q)(\mathbf{x}-Q)^T\boldsymbol{\nu}_{\mathbf{x}}(\mathbf{x}-Q)^T}{|\mathbf{x}-Q|^3}
\tilde{\boldsymbol{\varphi}}(\tilde{\mathbf{y}})\\
\end{aligned}
\end{equation}
and
\begin{equation}
\begin{aligned}
\mathbf{K}_{12}[\boldsymbol{\varphi}](\mathbf{x})=
&\Big(\lambda\alpha_4+\big(3\lambda+2\mu\big)\alpha_5\Big)\chi(S^a)\int_{S_{1}^b}
\frac{\boldsymbol{\nu}_{\mathbf{x}}(\mathbf{x}-Q)^T}{|\mathbf{x}-Q|}\tilde{\boldsymbol{\varphi}}(\tilde{\mathbf{y}})
\\&+\Big(\lambda\alpha_4+\big(3\lambda+2\mu\big)\alpha_5\Big)\chi(S^b)\int_{S_{1}^a}
\frac{\boldsymbol{\nu}_{\mathbf{x}}(\mathbf{x}-P)^T}{|\mathbf{x}-P|}\tilde{\boldsymbol{\varphi}}(\tilde{\mathbf{y}})
\\&+\mu(\alpha_4+\alpha_5)\chi(S^a)\int_{S_{1}^b}
\Big(\frac{(\mathbf{x}-Q)\boldsymbol{\nu}_{\mathbf{x}}^T}{|\mathbf{x}-Q|}
+\frac{(\mathbf{x}-Q,\boldsymbol{\nu}_{\mathbf{x}})}{|\mathbf{x}-Q|}\Big)\tilde{\boldsymbol{\varphi}}(\tilde{\mathbf{y}})\\
&+\mu(\alpha_4+\alpha_5)\chi(S^b)\int_{S_{1}^a}
\Big(\frac{(\mathbf{x}-P)\boldsymbol{\nu}_{\mathbf{x}}^T}{|\mathbf{x}-P|}
+\frac{(\mathbf{x}-P,\boldsymbol{\nu}_{\mathbf{x}})}{|\mathbf{x}-P|}\Big)\tilde{\boldsymbol{\varphi}}(\tilde{\mathbf{y}})
\\&-2\mu\alpha_5\chi(S^a)\int_{S_{1}^b}
\frac{(\mathbf{x}-Q)(\mathbf{x}-Q)^T\boldsymbol{\nu}_{\mathbf{x}}(\mathbf{x}-Q)^T}{|\mathbf{x}-Q|^3}
\tilde{\boldsymbol{\varphi}}(\tilde{\mathbf{y}})
\\&-2\mu\alpha_5\chi(S^b)\int_{S_{1}^a}
\frac{(\mathbf{x}-P)(\mathbf{x}-P)^T\boldsymbol{\nu}_{\mathbf{x}}(\mathbf{x}-P)^T}{|\mathbf{x}-P|^3}
\tilde{\boldsymbol{\varphi}}(\tilde{\mathbf{y}}),
\end{aligned}
\end{equation}
where $\chi_P:=\chi\left(S^f \backslash \iota_\delta(P)\right)$ and $\chi_Q:=\chi\left(S^f \backslash \iota_\delta(Q)\right)$. $\alpha_i, i=4, 5$ appear in \eqref{a45}. $\chi$ denote the characteristic function. In the above integrals, we omit the integral variable $s(\tilde{\mathbf{y}})$.
}
\end{lem}
\begin{proof}
Similar to the derivation of Lemma \ref{leS}, and note that
$$
\begin{aligned}
\frac{\partial \boldsymbol{\Gamma}_2(\mathbf{x}-\mathbf{y})}{\partial \boldsymbol{\nu}_{\mathbf{x}}} =& \Big(\lambda\alpha_4+\big(3\lambda+2\mu\big)\alpha_5\Big)\frac{\boldsymbol{\nu}_{\mathbf{x}}(\mathbf{x}-\mathbf{y})^T}{|\mathbf{x}-\mathbf{y}|}+\mu(\alpha_4+\alpha_5)\bigg(\frac{(\mathbf{x}-\mathbf{y})\boldsymbol{\nu}_{\mathbf{x}}^T}{|\mathbf{x}-\mathbf{y}|}
\\&+\frac{(\mathbf{x}-\mathbf{y},\boldsymbol{\nu}_{\mathbf{x}})}{|\mathbf{x}-\mathbf{y}|}\bigg)-2\mu\alpha_5\frac{(\mathbf{x}-\mathbf{y})(\mathbf{x}-\mathbf{y})^T\boldsymbol{\nu}_{\mathbf{x}}(\mathbf{x}-\mathbf{y})^T}{|\mathbf{x}-\mathbf{y}|^3},
\end{aligned}
$$
by direct calculations, we can obtain \eqref{SD1}, \eqref{SD2} and \eqref{KD2}.
\end{proof}

\begin{lem}\label{le:AD}
The operator $\mathcal{A}_{D}(\omega): \mathcal{H}^*\left(\partial D\right) \rightarrow \mathcal{H}^*\left(\partial D\right)$ has the following expansion:
\begin{equation}\label{1300}
\mathcal{A}_{D}(\omega)=\mathcal{A}_{D, 0}+\omega^2 \mathcal{A}_{D, 2}+\mathcal{O}\left(\omega^3\right),
\end{equation}
where $\mathcal{A}_{D, 0}$ is defined in \eqref{eq:AD0} and
$$
\mathcal{A}_{D, 2}=(1-c)\Big(\frac{I}{2}-\mathbf{K}_{D}^*\Big)\mathbf{S}_{D}^{-1} \mathbf{S}_{D, 2},
$$
with $\mathbf{S}_{D, 2}$ appear in Lemma \ref{le:1}. The following estimate holds:
$$\left\|\mathcal{A}_{D, 2}\right\|_{\mathcal{L}\left(\mathcal{H}^*\left(\partial D\right), \mathcal{H}^*\left(\partial D\right)\right)}=\mathcal{O}\left( 1\right).$$
\end{lem}
\begin{proof}
Using \eqref{F}-\eqref{eq:AD0}, Lemma \ref{le:1} and the fact that $\left(-\frac{I}{2}+\mathbf{K}_{D}^*\right)\mathbf{S}_{D}^{-1}[C]=\mathbf{0}$  (where $C$ is a constant), we can obtain the result by direct calculation. This completes the proof.
\end{proof}
\begin{lem}\label{le:F}
Let $\mathbf{F}$ be defined by \eqref{F}, one has
\begin{equation}\label{eq:FF}
\mathbf{F}=(1-c)\left(\frac{\partial \mathbf{H}_0}{\partial \boldsymbol{\nu}}+\omega\frac{\partial \mathbf{H}_1}{\partial \boldsymbol{\nu}}\right)+\mathcal{O}(\omega^2),
\end{equation}
where $\mathbf{H}_0$ and $\mathbf{H}_1$ are the zeroth-order and first-order terms of $\mathbf{H}$ with respect to $\omega$, respectively.
\end{lem}
\begin{proof}
Recall that $\mathcal{L}_{\lambda, \mu} \mathbf{H}+\omega^2\mathbf{H}=\mathbf{0}$ in $\mathbb{R}^3$, we can set
$\mathbf{H}=\mathbf{H}_0+\omega\mathbf{H}_1+\mathcal{O}(\omega^2)$, then $\mathcal{L}_{\lambda, \mu}\mathbf{H}_i=\mathbf{0}$, $i=0, 1$. It follows from \eqref{F} and Lemma \ref{le:1} that
$$
\begin{aligned}
\mathbf{F}=&\frac{\partial \mathbf{H}}{\partial \boldsymbol{\nu}}-c\left(-\frac{I}{2}+\left(\mathbf{K}_{D}^{\omega_1}\right)^*\right)\left(\mathbf{S}_{D}^{\omega_1}\right)^{-1}\left[\mathbf{H}\right]\\
=&\frac{\partial \mathbf{H}_0}{\partial \boldsymbol{\nu}}+\omega\frac{\partial \mathbf{H}_1}{\partial \boldsymbol{\nu}}-c\left(-\frac{I}{2}+\mathbf{K}_{D}^*\right)\left(\mathbf{S}_{D}\right)^{-1}\left[\mathbf{H}_0\right]\\
&-\omega c\left(-\frac{I}{2}+\mathbf{K}_{D}^*\right)\left(\mathbf{S}_{D}\right)^{-1}\left[\mathbf{H}_1\right]
-\omega_1 c\left(-\frac{I}{2}+\mathbf{K}_{D}^*\right)\mathbf{B}_{D, 1}\left[\mathbf{H}_0\right]+\mathcal{O}(\omega^2)\\
=&\frac{\partial \mathbf{H}_0}{\partial \boldsymbol{\nu}}+\omega\frac{\partial \mathbf{H}_1}{\partial \boldsymbol{\nu}}-c\frac{\partial \mathbf{H}_0}{\partial \boldsymbol{\nu}}-c\omega\frac{\partial \mathbf{H}_1}{\partial \boldsymbol{\nu}}+\mathcal{O}(\omega^2)=(1-c)\left(\frac{\partial \mathbf{H}_0}{\partial \boldsymbol{\nu}}+\omega\frac{\partial \mathbf{H}_1}{\partial \boldsymbol{\nu}}\right)+\mathcal{O}(\omega^2),
\end{aligned}
$$
we have used the following fact in the above formula
\begin{equation}\label{eq:1234}
\left(-\frac{I}{2}+\mathbf{K}_{D}^*\right)\left(\mathbf{S}_{D}\right)^{-1}\left[\mathbf{H}_i\right]=\frac{\partial \mathbf{H}_i}{\partial \boldsymbol{\nu}},\quad i=0, 1,
\end{equation}
\begin{equation}\label{eq:2345}
\left(-\frac{I}{2}+\mathbf{K}_{D}^*\right)\mathbf{B}_{D, 1}\left[\mathbf{H}_0\right]=\mathbf{0}.
\end{equation}
In fact, we set $\left(\mathbf{S}_{D}\right)^{-1}\left[\mathbf{H}_i\right]=\boldsymbol{\Phi}_i$ on $\partial D$, then $\mathbf{S}_{D}[\boldsymbol{\Phi}_i]=\mathbf{H}_i$ on $\partial D$. Owing to $\mathcal{L}_{\lambda, \mu}\mathbf{S}_{D}[\boldsymbol{\Phi}_i]=\mathbf{0}$ in $D$ and $\mathcal{L}_{\lambda, \mu}\mathbf{H}_i=\mathbf{0}$ in $D$, $i=0, 1$, so we have $\mathbf{S}_{D}[\boldsymbol{\Phi}_i]=\mathbf{H}_i$ in $D$. By jump formula \eqref{eq:04}, \eqref{eq:1234} holds. \eqref{eq:2345} follows from the fact that $\left(-\frac{I}{2}+\mathbf{K}_{D}^*\right)\mathbf{S}_{D}^{-1}[C]=\mathbf{0}$.
\end{proof}

Next, we further analyze \eqref{eq:main1}. The following Lemma holds:
\begin{lem}\label{leD}
Suppose $\boldsymbol{\psi}$ is given in \eqref{eq:03}, then we have
\begin{equation}\label{1015}
\begin{aligned}
&\left(\begin{array}{c} \boldsymbol{\psi}(x_1, -\delta\cos\theta_1, -\delta\sin\theta_1)\\[8pt]
 \boldsymbol{\psi}(x_1, \delta\cos\theta_1, \delta\sin\theta_1)\end{array} \right)\\=&\left(\begin{array}{cc} \lambda_1-\mathcal{D}_{11}& \mathcal{D}_{12}\\[8pt] \mathcal{D}_{21}& \lambda_1-\mathcal{D}_{22}\end{array} \right)^{-1}\left(\begin{array}{c}-\mathbf{f}_{\mathbf{H}_0}\\[8pt]
         \mathbf{f}_{\mathbf{H}_0}\end{array}\right)+\delta\left(\begin{array}{cc} \lambda_1-\mathcal{D}_{11}& \mathcal{D}_{12}\\[8pt] \mathcal{D}_{21}& \lambda_1-\mathcal{D}_{22}\end{array} \right)^{-1}\left(\begin{array}{c}\mathbf{g}_{\mathbf{H}_0}\\[8pt]
         \mathbf{g}_{\mathbf{H}_0}\end{array}\right)\\+&\omega\left(\begin{array}{cc} \lambda_1-\mathcal{D}_{11}& \mathcal{D}_{12}\\[8pt] \mathcal{D}_{21}& \lambda_1-\mathcal{D}_{22}\end{array} \right)^{-1}\left(\begin{array}{c}-\mathbf{f}_{\mathbf{H}_1}\\[8pt]
         \mathbf{f}_{\mathbf{H}_1}\end{array}\right)+\omega\delta\left(\begin{array}{cc} \lambda_1-\mathcal{D}_{11}& \mathcal{D}_{12}\\[8pt] \mathcal{D}_{21}& \lambda_1-\mathcal{D}_{22}\end{array} \right)^{-1}\left(\begin{array}{c}\mathbf{g}_{\mathbf{H}_1}\\[8pt]
         \mathbf{g}_{\mathbf{H}_1}\end{array}\right)\\+&\chi\left(\iota_{\delta^{\epsilon}}(P) \cup \iota_{\delta^{\epsilon}}(Q)\right) \mathcal{O}\left(\delta^{3(1-\varepsilon)}\right)+o(\delta)+\chi\left(\iota_{\delta^{\epsilon}}(P) \cup \iota_{\delta^{\epsilon}}(Q)\right)\omega\cdot\mathcal{O}\left(\delta^{3(1-\varepsilon)}\right)+\omega\cdot o(\delta)\\+&\mathcal{O}(\omega^2),\quad |x_1|\leq L/2-O(\delta),
\end{aligned}
\end{equation}
for any $0<\theta_1<\pi,\quad 0\leq\epsilon<1$, where
$$\mathbf{f}_{\mathbf{H}_k}=\left(\lambda(\nabla\cdot \mathbf{H}_k(\mathbf{z}_{\mathbf{x}}))
         +\mu(\nabla \mathbf{H}_k(\mathbf{z}_{\mathbf{x}})+\nabla \mathbf{H}_k(\mathbf{z}_{\mathbf{x}})^T)\right)\boldsymbol{\nu}_{\mathbf{x}},$$
and
$$\mathbf{g}_{\mathbf{H}_k}=\lambda A\nabla(\nabla\cdot \partial_{x_2}\mathbf{H}_k(\mathbf{z}_{\mathbf{x}}))+\mu((A\nabla)^T(\nabla\mathbf{H}_k(\mathbf{z}_{\mathbf{x}})+\nabla\mathbf{H}_k(\mathbf{z}_{\mathbf{x}})^T))^T, k=0, 1.$$
The matrix $A$ appears above is defined in Theorem \ref{3main}. The operators $\mathcal{D}_{ij}$, $i=1,2$, $j=1,2$ admit the following forms
         \begin{footnotesize}
\begin{equation}\label{1115}
\begin{aligned}
\mathcal{D}_{11}&[\boldsymbol{\psi}](x_1, \theta_1)=-\frac{\mu\delta}{4\pi(\lambda+2\mu)}\int_{-\frac{L}{2}}^{\frac{L}{2}}\int_{0}^{\pi}\frac{\mathbf{A}_1}{\big((x_1-y_1)^2+2\delta^2(1-\cos(\theta_1-\theta_2))\big)^{\frac{3}{2}}} \boldsymbol{\psi}(y_1, \theta_2)d\theta_2dy_1 \\&+\frac{\mu\delta}{4\pi(\lambda+2\mu)}\int_{-\frac{L}{2}}^{\frac{L}{2}}\int_{0}^{\pi}\frac{\delta\big(1-\cos(\theta_2-\theta_1)\big)}{\big((x_1-y_1)^2+2\delta^2(1-\cos(\theta_1-\theta_2))\big)^{\frac{3}{2}}}
\boldsymbol{\psi}(y_1, \theta_2)d\theta_2dy_1
\\&+\frac{3(\lambda+\mu)\delta}{4\pi(\lambda+2\mu)}
\int_{-\frac{L}{2}}^{\frac{L}{2}}\int_{0}^{\pi}\frac{\delta\big(1-\cos(\theta_2-\theta_1)\big)}{\big((x_1-y_1)^2+2\delta^2(1-\cos(\theta_1-\theta_2))\big)^{\frac{5}{2}}}\mathcal{C}_1(\mathbf{x}-\mathbf{y}) \boldsymbol{\psi}(y_1, \theta_2)d\theta_2dy_1,
\end{aligned}
\end{equation}
\end{footnotesize}
\begin{footnotesize}
\begin{equation}\label{1116}
\begin{aligned}
\mathcal{D}_{12}&[\boldsymbol{\psi}](x_1, \theta_1)=\frac{\mu\delta}{4\pi(\lambda+2\mu)}\int_{-\frac{L}{2}}^{\frac{L}{2}}\int_{0}^{\pi}\frac{\mathbf{A}_2}{\big((x_1-y_1)^2+2\delta^2(1+\cos(\theta_1-\theta_2))\big)^{\frac{3}{2}}}\boldsymbol{\psi}(y_1, \theta_2)d\theta_2dy_1 \\&-\frac{\mu\delta}{4\pi(\lambda+2\mu)}\int_{-\frac{L}{2}}^{\frac{L}{2}}\int_{0}^{\pi}\frac{\delta\big(1+\cos(\theta_2-\theta_1)\big)}{\big((x_1-y_1)^2+2\delta^2(1+\cos(\theta_1-\theta_2))\big)^{\frac{3}{2}}}
\boldsymbol{\psi}(y_1, \theta_2)d\theta_2dy_1
\\&-\frac{3(\lambda+\mu)\delta}{4\pi(\lambda+2\mu)}
\int_{-\frac{L}{2}}^{\frac{L}{2}}\int_{0}^{\pi}\frac{\delta\big(1+\cos(\theta_2-\theta_1)\big)}{\big((x_1-y_1)^2+2\delta^2(1+\cos(\theta_1-\theta_2))\big)^{\frac{5}{2}}}\mathcal{C}_2(\mathbf{x}-\mathbf{y}) \boldsymbol{\psi}(y_1, \theta_2)d\theta_2dy_1,
\end{aligned}
\end{equation}
\end{footnotesize}
\begin{footnotesize}
\begin{equation}\label{1117}
\begin{aligned}
\mathcal{D}_{21}&[\boldsymbol{\psi}](x_1, \theta_1)=
\frac{\mu\delta}{4\pi(\lambda+2\mu)}\int_{-\frac{L}{2}}^{\frac{L}{2}}\int_{0}^{\pi}\frac{\mathbf{A}_3}{\big((x_1-y_1)^2+2\delta^2(1+\cos(\theta_1-\theta_2))\big)^{\frac{3}{2}}} \boldsymbol{\psi}(y_1, \theta_2)d\theta_2dy_1
\\&-\frac{\mu\delta}{4\pi(\lambda+2\mu)}\int_{-\frac{L}{2}}^{\frac{L}{2}}\int_{0}^{\pi}\frac{\delta\big(1+\cos(\theta_2-\theta_1)\big)}{\big((x_1-y_1)^2+2\delta^2(1+\cos(\theta_1-\theta_2))\big)^{\frac{3}{2}}}
\boldsymbol{\psi}(y_1, \theta_2)d\theta_2dy_1
\\&-\frac{3(\lambda+\mu)\delta}{4\pi(\lambda+2\mu)}
\int_{-\frac{L}{2}}^{\frac{L}{2}}\int_{0}^{\pi}\frac{\delta\big(1+\cos(\theta_2-\theta_1)\big)}{\big((x_1-y_1)^2+2\delta^2(1+\cos(\theta_1-\theta_2))\big)^{\frac{5}{2}}}\mathcal{C}_3(\mathbf{x}-\mathbf{y}) \boldsymbol{\psi}(y_1, \theta_2)d\theta_2dy_1,
\end{aligned}
\end{equation}
\end{footnotesize}
and
\begin{footnotesize}
\begin{equation}\label{1118}
\begin{aligned}
\mathcal{D}_{22}&[\boldsymbol{\psi}](x_1, \theta_1)=
-\frac{\mu\delta}{4\pi(\lambda+2\mu)}\int_{-\frac{L}{2}}^{\frac{L}{2}}\int_{0}^{\pi}\frac{\mathbf{A}_0}{\big((x_1-y_1)^2+2\delta^2(1-\cos(\theta_1-\theta_2))\big)^{\frac{3}{2}}} \boldsymbol{\psi}(y_1, \theta_2)d\theta_2dy_1
\\&+\frac{\mu\delta}{4\pi(\lambda+2\mu)}\int_{-\frac{L}{2}}^{\frac{L}{2}}\int_{0}^{\pi}\frac{\delta\big(1-\cos(\theta_2-\theta_1)\big)}{\big((x_1-y_1)^2+2\delta^2(1-\cos(\theta_1-\theta_2))\big)^{\frac{3}{2}}}
\boldsymbol{\psi}(y_1, \theta_2)d\theta_2dy_1
\\&+\frac{3(\lambda+\mu)\delta}{4\pi(\lambda+2\mu)}
\int_{-\frac{L}{2}}^{\frac{L}{2}}\int_{0}^{\pi}\frac{\delta\big(1-\cos(\theta_2-\theta_1)\big)}{\big((x_1-y_1)^2+2\delta^2(1-\cos(\theta_1-\theta_2))\big)^{\frac{5}{2}}}\mathcal{C}_0(\mathbf{x}-\mathbf{y}) \boldsymbol{\psi}(y_1, \theta_2)d\theta_2dy_1,
\end{aligned}
\end{equation}
\end{footnotesize}
where
\begin{equation}\label{A1}
\mathbf{A}_1= \left(
                                                                     \begin{array}{ccc}
                                                                       0 & \cos\theta_1(x_1-y_1)&  \sin\theta_1(x_1-y_1) \\
                                                                       -\cos\theta_1(x_1-y_1)& 0 & -\delta\sin(\theta_2-\theta_1) \\
                                                                       -\sin\theta_1(x_1-y_1)& \delta\sin(\theta_2-\theta_1)& 0
                                                                     \end{array}
                                                                   \right),
\end{equation}
\begin{equation}\label{A2}
\mathbf{A}_2= \left(
                                                                     \begin{array}{ccc}
                                                                       0 & \cos\theta_1(x_1-y_1)&  \sin\theta_1(x_1-y_1) \\
                                                                       -\cos\theta_1(x_1-y_1)& 0 & \delta\sin(\theta_2-\theta_1) \\
                                                                       -\sin\theta_1(x_1-y_1)& -\delta\sin(\theta_2-\theta_1)& 0
                                                                     \end{array}
                                                                   \right),
\end{equation}
\begin{equation}\label{A3}
\mathbf{A}_3= \left(
                                                                     \begin{array}{ccc}
                                                                       0 & -\cos\theta_1(x_1-y_1)&  -\sin\theta_1(x_1-y_1) \\
                                                                       \cos\theta_1(x_1-y_1)& 0 & \delta\sin(\theta_2-\theta_1) \\
                                                                       \sin\theta_1(x_1-y_1)& -\delta\sin(\theta_2-\theta_1)& 0
                                                                     \end{array}
                                                                   \right),
\end{equation}
and $\mathbf{A}_0$ is defined by
\begin{equation}\label{A0}
\mathbf{A}_0= \left(
                                                                     \begin{array}{ccc}
                                                                       0 & -\cos\theta_1(x_1-y_1)&  -\sin\theta_1(x_1-y_1) \\
                                                                       \cos\theta_1(x_1-y_1)& 0 & -\delta\sin(\theta_2-\theta_1) \\
                                                                       \sin\theta_1(x_1-y_1)& \delta\sin(\theta_2-\theta_1)& 0
                                                                     \end{array}
                                                                   \right).
\end{equation}
The matrix operators $\mathcal{C}_1(\mathbf{x}-\mathbf{y})$, $\mathcal{C}_2(\mathbf{x}-\mathbf{y})$ and $\mathcal{C}_3(\mathbf{x}-\mathbf{y})$ are defined by
\begin{footnotesize}
\begin{equation}\label{C1}
\left(\begin{array}{ccc} (x_1-y_1)^2 & -\delta(\cos\theta_1-\cos\theta_2)(x_1-y_1)& -\delta(\sin\theta_1-\sin\theta_2)(x_1-y_1)\\[4pt] -\delta(\cos\theta_1-\cos\theta_2)(x_1-y_1)&\delta^2(\cos\theta_1-\cos\theta_2)^2&\delta^2(\cos\theta_1-\cos\theta_2)(\sin\theta_1-\sin\theta_2)\\[4pt] -\delta(\sin\theta_1-\sin\theta_2)(x_1-y_1)&\delta^2(\cos\theta_1-\cos\theta_2)(\sin\theta_1-\sin\theta_2)&\delta^2(\sin\theta_1-\sin\theta_2)^2\end{array} \right),
\end{equation}
\end{footnotesize}
\begin{footnotesize}
\begin{equation}\label{C2}
\left(\begin{array}{ccc} (x_1-y_1)^2 & -\delta(\cos\theta_1+\cos\theta_2)(x_1-y_1)&-\delta(\sin\theta_1+\sin\theta_2)(x_1-y_1)\\[4pt] -\delta(\cos\theta_1+\cos\theta_2)(x_1-y_1)&\delta^2(\cos\theta_1+\cos\theta_2)^2&\delta^2(\cos\theta_1+\cos\theta_2)(\sin\theta_1+\sin\theta_2)\\[4pt] -\delta(\sin\theta_1+\sin\theta_2)(x_1-y_1)&\delta^2(\cos\theta_1+\cos\theta_2)(\sin\theta_1+\sin\theta_2)&\delta^2(\sin\theta_1+\sin\theta_2)^2\end{array} \right),
\end{equation}
\end{footnotesize}
\begin{footnotesize}
\begin{equation}\label{C3}
\left(\begin{array}{ccc} (x_1-y_1)^2 & \delta(\cos\theta_1+\cos\theta_2)(x_1-y_1)&\delta(\sin\theta_1+\sin\theta_2)(x_1-y_1)\\[4pt] \delta(\cos\theta_1+\cos\theta_2)(x_1-y_1)&\delta^2(\cos\theta_1+\cos\theta_2)^2&\delta^2(\cos\theta_1+\cos\theta_2)(\sin\theta_1+\sin\theta_2)\\[4pt] \delta(\sin\theta_1+\sin\theta_2)(x_1-y_1)&\delta^2(\cos\theta_1+\cos\theta_2)(\sin\theta_1+\sin\theta_2)&\delta^2(\sin\theta_1+\sin\theta_2)^2\end{array} \right),
\end{equation}
\end{footnotesize}
respectively, and $\mathcal{C}_0(\mathbf{x}-\mathbf{y})=(\mathbf{x}-\mathbf{y})(\mathbf{x}-\mathbf{y})^T$.
\end{lem}
\begin{proof}
From \eqref{eq:main1}, Lemmas \ref{le:AD} and \ref{le:F}, we have
\begin{equation}\label{eq:1011}
\Big(\mathcal{A}_{D, 0}+\omega^2 \mathcal{A}_{D, 2}+\mathcal{O}\left(\omega^3(\ln \delta)^{-1}\right)\Big)\big[\boldsymbol{\psi}\big]=(1-c)\left(\frac{\partial \mathbf{H}_0}{\partial \boldsymbol{\nu}}+\omega\frac{\partial \mathbf{H}_1}{\partial \boldsymbol{\nu}}\right)+\mathcal{O}(\omega^2).
\end{equation}
Let $\boldsymbol{\psi}=\boldsymbol{\psi}_0+\omega\boldsymbol{\psi}_1+\mathcal{O}(\omega^2)$, then \eqref{eq:1011} can be rewritten as
\begin{equation}\label{eq:1012}
\mathcal{A}_{D, 0}[\boldsymbol{\psi}_0]+\omega\mathcal{A}_{D, 0}[\boldsymbol{\psi}_1]=(1-c)\left(\frac{\partial \mathbf{H}_0}{\partial \boldsymbol{\nu}}+\omega\frac{\partial \mathbf{H}_1}{\partial \boldsymbol{\nu}}\right)+\mathcal{O}(\omega^2),
\end{equation}
comparing the orders of $\omega$ on both sides of equation \eqref{eq:1012}, we can obtain
$$
\mathcal{A}_{D, 0}[\boldsymbol{\psi}_0]=(1-c)\frac{\partial \mathbf{H}_0}{\partial \boldsymbol{\nu}},\quad
\mathcal{A}_{D, 0}[\boldsymbol{\psi}_1]=(1-c)\frac{\partial \mathbf{H}_1}{\partial \boldsymbol{\nu}}.
$$
From the definition of operator $\mathcal{A}_{D, 0}$ in \eqref{eq:AD0}, we can further obtain
\begin{equation}\label{eq:1614}
\left(\lambda_1 I-\mathbf{K}_D^*\right)[\boldsymbol{\psi}_i]=\frac{\partial \mathbf{H}_i}{\partial \boldsymbol{\nu}},\quad i=0, 1,
\end{equation}
where $\lambda_1=\frac{c+1}{2(c-1)}$.
Then, by an argument analogous to the proof of Lemma 3.2 in \cite{TDF}, we obtain the result. This completes the proof.
\end{proof}
\begin{lem}\label{leSa}
Suppose $\boldsymbol{\psi}$ is given in \eqref{eq:03}, then there hold
\begin{equation}\label{1032}
\begin{aligned}
\int_{S^a \cup \iota_\delta(P)}\boldsymbol{\psi}=&-\pi\delta^2(\lambda_1-\frac{1}{2})^{-1}\left[\lambda(\nabla\cdot \mathbf{H}_0(P))\mathbf{e}_1
         +2\mu\nabla^s \mathbf{H}_0(P)\mathbf{e}_1\right]\\&-\omega\pi\delta^2(\lambda_1-\frac{1}{2})^{-1}\left[\lambda(\nabla\cdot \mathbf{H}_1(P))\mathbf{e}_1
         +2\mu\nabla^s \mathbf{H}_1(P)\mathbf{e}_1\right]
         +o(\delta^2)\\&+\omega\cdot o(\delta^2)+\delta^2\cdot\mathcal{O}(\omega^2),
\end{aligned}
\end{equation}
\begin{equation}\label{1033}
\begin{aligned}
\int_{S^b \cup \iota_\delta(Q)}\boldsymbol{\psi}=&\pi\delta^2(\lambda_1-\frac{1}{2})^{-1}\left[\lambda(\nabla\cdot \mathbf{H}_0(Q))\mathbf{e}_1+2\mu\nabla^s \mathbf{H}_0(Q)\mathbf{e}_1\right]\\&+\omega\pi\delta^2(\lambda_1-\frac{1}{2})^{-1}\left[\lambda(\nabla\cdot \mathbf{H}_1(Q))\mathbf{e}_1+2\mu\nabla^s \mathbf{H}_1(Q)\mathbf{e}_1\right]
         +o(\delta^2)\\&+\omega\cdot o(\delta^2)+\delta^2\cdot\mathcal{O}(\omega^2).
\end{aligned}
\end{equation}
\end{lem}
\begin{proof}
For any $\mathbf{U}_i \in \mathcal{H}^*, i=0, 1$, we consider the following boundary integral equation
$$
\left(\lambda_1 I-\mathbf{K}_D^*\right)[\boldsymbol{\psi}_i]=\mathbf{U}_i.
$$
Taking boundary integral on both sides of the above equation, it holds
\begin{equation}
\int_{\partial D}\left(\lambda_1 I-\mathbf{K}_D^*\right)[\boldsymbol{\psi}_i](\mathbf{x})ds(\mathbf{x})=(\lambda_1-\frac{1}{2})\int_{\partial D}\boldsymbol{\psi}_i(\mathbf{y})ds(\mathbf{y})=\int_{\partial D}\mathbf{U}_i.
\end{equation}
By assuming $\mathbf{U}_i=\chi\left(S^a\cup\iota_\delta(P)\right)[\frac{\partial\mathbf{H}_i(P)}{\partial\boldsymbol{\nu}}]$ and using the fact that $\mathcal{L}_{\lambda, \mu}\mathbf{H}_i=\mathbf{0}$ in $\mathbb{R}^3$, one has
\begin{equation}
\begin{aligned}
(\lambda_1-\frac{1}{2})&\int_{S^a\cup\iota_\delta(P)}\boldsymbol{\psi}_i=\int_{S^a\cup\iota_\delta(P)}\frac{\partial\mathbf{H}_i(P)}{\partial\boldsymbol{\nu}}
\\&=-\pi\delta^2
\left[\lambda(\nabla\cdot \mathbf{H}_i(P))\mathbf{e}_1
         +\mu(\nabla \mathbf{H}_i(P)+\nabla \mathbf{H}_i(P)^T)\mathbf{e}_1\right]+o(\delta^2),
\end{aligned}
\end{equation}
noting that $\boldsymbol{\psi}=\boldsymbol{\psi}_0+\omega\boldsymbol{\psi}_1+\mathcal{O}(\omega^2)$, so we can derive \eqref{1032}. \eqref{1033} can be derived similarly.
\end{proof}
\subsection*{Proof of Theorem \ref{3main}}
\begin{proof}
Denote by $\Gamma_{ij}$ the entries of the matrix $\mathbf\Gamma$, $i, j=1, 2, 3$. By using Taylor's expansion, one has
\begin{equation*}
\begin{aligned}
\mathbf{u}(\mathbf{x})=&\mathbf{H}(\mathbf{x}) +\int_{S^f \backslash\left(\iota_\delta(P) \cup \iota_\delta(Q)\right)} \mathbf\Gamma\left(\mathbf{x}-\mathbf{z}_{\mathbf{y}}\right) \boldsymbol{\psi}(\mathbf{y}) d s(\mathbf{y})
\\&-\delta \int_{S^f \backslash\left(\iota_\delta(P) \cup\iota_\delta(Q)\right)} \boldsymbol{\nu}_{\mathbf{y}}\cdot\nabla_{\mathbf{x}}\mathbf\Gamma\left(\mathbf{x}-\mathbf{z}_{\mathbf{y}}\right)\boldsymbol{\psi}(\mathbf{y}) d s(\mathbf{y})\\&+\int_{S^a \cup \iota_\delta(P)} \mathbf\Gamma\left(\mathbf{x}-\mathbf{z}_{\mathbf{y}}\right) \boldsymbol{\psi}(\mathbf{y}) d s(\mathbf{y})+\int_{S^b \cup\iota_\delta(Q)} \mathbf\Gamma\left(\mathbf{x}-\mathbf{z}_{\mathbf{y}}\right) \boldsymbol{\psi}(\mathbf{y}) d s(\mathbf{y})
\end{aligned}
\end{equation*}
\begin{equation}\label{1623}
\begin{aligned}
\\&+\omega\alpha_3\int_{S^f \backslash\left(\iota_\delta(P) \cup \iota_\delta(Q)\right)}\boldsymbol{\psi}(\mathbf{y}) d s(\mathbf{y})+\omega\alpha_3\int_{S^a \cup \iota_\delta(P)}\boldsymbol{\psi}(\mathbf{y}) d s(\mathbf{y})+\omega\alpha_3\int_{S^b \cup \iota_\delta(Q)}\boldsymbol{\psi}(\mathbf{y}) d s(\mathbf{y})\\&+\delta^2\cdot\mathcal{O}(\omega^2)+o\left(\delta^2\right).
\end{aligned}
\end{equation}
By Lemma \ref{leD} and proofs of Theorem 2.1 in \cite{TDF}, we can directly calculate and obtain
\begin{equation}\label{1619}
\begin{aligned}
&\int_{S^f \backslash\left(\iota_\delta(P) \cup \iota_\delta(Q)\right)}\mathbf\Gamma\left(\mathbf{x}-\mathbf{z}_{\mathbf{y}}\right) \boldsymbol{\psi}(\mathbf{y}) d s(\mathbf{y})\\=&\int_{\Gamma_1 \backslash\left(\iota_\delta(P) \cup \iota_\delta(Q)\right)}\mathbf\Gamma\left(\mathbf{x}-\mathbf{z}_{\mathbf{y}}\right) \boldsymbol{\psi}(\mathbf{y}) d s(\mathbf{y})+\int_{\Gamma_2 \backslash\left(\iota_\delta(P) \cup \iota_\delta(Q)\right)}\mathbf\Gamma\left(\mathbf{x}-\mathbf{z}_{\mathbf{y}}\right) \boldsymbol{\psi}(\mathbf{y}) d s(\mathbf{y})
\\
=&2\delta^2\int_{-L/2}^{L/2}\int_{0}^{\pi}\mathbf\Gamma\left(\mathbf{x}-\mathbf{z}_{\mathbf{y}}\right)\mathbf{G}^{\prime}_{\mathbf{H}_0}(\mathbf{y})d\theta_2 dy_1+o(\delta^2)\\
&+2\omega\delta^2\int_{-L/2}^{L/2}\int_{0}^{\pi}\mathbf\Gamma\left(\mathbf{x}-\mathbf{z}_{\mathbf{y}}\right)\mathbf{G}^{\prime}_{\mathbf{H}_1}(\mathbf{y})d\theta_2 dy_1
+\omega\cdot o(\delta^2),
\end{aligned}
\end{equation}
where $\mathbf{G}^{\prime}_{\mathbf{H}_i}$, $i=0, 1$ can be represented by
\begin{equation}\label{3G}
\begin{aligned}
\mathbf{G}^{\prime}_{\mathbf{H}_i}=&2(\lambda+2\mu)l_1\Big(\lambda\big(\nabla\cdot \partial_{y_2}\mathbf{H}_i(\mathbf{z}_{\mathbf{y}})\big)+2\mu^2 l_2(\partial_{y_2}\partial_{y_2}H_{i2}+\partial_{y_2}\partial_{y_3}H_{i3})\Big) \cos\theta_2\boldsymbol{\nu}_{\mathbf{y}}\\&+2(\lambda+2\mu)l_1\Big(\lambda\big(\nabla\cdot \partial_{y_3}\mathbf{H}_i(\mathbf{z}_{\mathbf{y}})\big)+2\mu^2 l_2(\partial_{y_3}\partial_{y_2}H_{i2}+\partial_{y_3}\partial_{y_3}H_{i3})\Big) \sin\theta_2\boldsymbol{\nu}_{\mathbf{y}}\\[5pt]
&+
\frac{\lambda+\mu}{2\lambda_1-1}l_1 \left(\begin{array}{c}0 \\ \lambda\big(\nabla\cdot \partial_{y_2}\mathbf{H}_i(\mathbf{z}_{\mathbf{y}})\big)+2\mu^2l_2(\partial_{y_2}\partial_{y_2}H_{i2}+\partial_{y_2}\partial_{y_3}H_{i3})\\[8pt] \lambda\big(\nabla\cdot \partial_{y_3}\mathbf{H}_i(\mathbf{z}_{\mathbf{y}})\big)+2\mu^2l_2(\partial_{y_3}\partial_{y_2}H_{i2}+\partial_{y_3}\partial_{y_3}H_{i3})\end{array} \right)
\\[5pt]
&+\frac{\mu}{\lambda_1} (\mathbf{g}_{\mathbf{H}_i, 1}+\mathbf{g}_{\mathbf{H}_i, 2})-\frac{\mu^2}{\lambda_1}l_2 (\mathbf{g}_{\mathbf{H}_i, 1}^0+\mathbf{g}_{\mathbf{H}_i, 2}^0)
\\&+\frac{\mu}{2\lambda_1(2\lambda_1-1)} \left(\begin{array}{c} \partial_{y_2}\partial_{y_2}H_{i1}+\partial_{y_2}\partial_{y_1}H_{i2}+\partial_{y_3}\partial_{y_3}H_{i1}+\partial_{y_1}\partial_{y_3}H_{i3}\\ 2\partial_{y_2}\partial_{y_2}H_{i2}+\partial_{y_3}\partial_{y_2}H_{i3}+\partial_{y_3}\partial_{y_3}H_{i2} \\ \partial_{y_2}\partial_{y_2}H_{i3}+\partial_{y_2}\partial_{y_3}H_{i2}+2\partial_{y_3}\partial_{y_3}H_{i3}\end{array} \right)
\\[5pt]&-\frac{\mu^2}{2\lambda_1-1}l_2\left(\begin{array}{c}0 \\2\partial_{y_2}\partial_{y_3}H_{i3}-(\partial_{y_3}\partial_{y_2}H_{i3}+\partial_{y_3}\partial_{y_3}H_{i2}) \\ -(\partial_{y_2}\partial_{y_2}H_{i3}+\partial_{y_3}\partial_{y_2}H_{i2})+2\partial_{y_3}\partial_{y_2}H_{i2}\end{array} \right)
\\[5pt]&-\frac{\mu^2}{2\lambda_1(2\lambda_1-1)}l_2 \left(\begin{array}{c} 0\\ 2\partial_{y_2}\partial_{y_2}H_{i2}+ \partial_{y_3}\partial_{y_2}H_{i3}+\partial_{y_3}\partial_{y_3}H_{i2}\\ \partial_{y_2}\partial_{y_2}H_{i3}+\partial_{y_3}\partial_{y_2}H_{i2}+2\partial_{y_3}\partial_{y_3}H_{i3}\end{array} \right),
\end{aligned}
\end{equation}
similarly, by direct calculations (see also proofs of Theorem 2.1 in \cite{TDF}) we have
\begin{equation}\label{1620}
\begin{aligned}
&\delta\int_{S^f \backslash\left(\iota_\delta(P) \cup\iota_\delta(Q)\right)}\boldsymbol{\nu}_{\mathbf{y}}\cdot\nabla_{\mathbf{x}}\mathbf\Gamma\left(\mathbf{x}-\mathbf{z}_{\mathbf{y}}\right)\boldsymbol{\psi}(\mathbf{y}) d s(\mathbf{y})\\
=&2\delta^2\int_{-L/2}^{L/2}\int_{0}^{\pi}\boldsymbol{\nu}_{\mathbf{y}}\cdot\nabla_{\mathbf{x}}\mathbf\Gamma\left(\mathbf{x}-\mathbf{z}_{\mathbf{y}}\right)\mathbf{R}^{\prime}_{\mathbf{H}_0}(\mathbf{y})d\theta_2dy_1+o(\delta^2)\\
&+2\omega\delta^2\int_{-L/2}^{L/2}\int_{0}^{\pi}\boldsymbol{\nu}_{\mathbf{y}}\cdot\nabla_{\mathbf{x}}\mathbf\Gamma\left(\mathbf{x}-\mathbf{z}_{\mathbf{y}}\right)\mathbf{R}^{\prime}_{\mathbf{H}_1}(\mathbf{y})d\theta_2dy_1+\omega\cdot o(\delta^2),
\end{aligned}
\end{equation}
where
\begin{equation*}
\begin{aligned}
\mathbf{R}^{\prime}_{\mathbf{H}_i}=&\Big(\big(l_{01}\nabla\cdot \mathbf{H}_i(\mathbf{z}_{\mathbf{y}})+l_{02}\nabla\cdot(B\mathbf{H}_i(\mathbf{z}_{\mathbf{y}}))\big)I_3 \\
&+\frac{\mu}{\lambda_1}\big(I_3-\mu l_2 B\big)\big(\nabla \mathbf{H}_i(\mathbf{z}_{\mathbf{y}})+\nabla \mathbf{H}_i(\mathbf{z}_{\mathbf{y}})^T\big)\Big)\boldsymbol{\nu}_{\mathbf{y}}.
\end{aligned}
\end{equation*}
It follows from Lemma \ref{leSa} that
\begin{equation}\label{1621}
\begin{aligned}
&\int_{S^a \cup \iota_\delta(P)} \mathbf\Gamma\left(\mathbf{x}-\mathbf{z}_{\mathbf{y}}\right) \boldsymbol{\psi}(\mathbf{y}) d s(\mathbf{y})=\int_{S^a \cup \iota_\delta(P)} \mathbf\Gamma\left(\mathbf{x}-P\right) \boldsymbol{\psi}(\mathbf{y}) d s(\mathbf{y})+o(\delta^2)\\=&-\pi\delta^2(\lambda_1-\frac{1}{2})^{-1}\mathbf\Gamma\left(\mathbf{x}-P\right)\left[\lambda(\nabla\cdot \mathbf{H}_0(P))\mathbf{e}_1
         +2\mu\nabla^s \mathbf{H}_0(P)\mathbf{e}_1\right]\\&-\omega\pi\delta^2(\lambda_1-\frac{1}{2})^{-1}\mathbf\Gamma\left(\mathbf{x}-P\right)\left[\lambda(\nabla\cdot \mathbf{H}_1(P))\mathbf{e}_1
         +2\mu\nabla^s \mathbf{H}_1(P)\mathbf{e}_1\right]+o(\delta^2)+\delta^2\cdot\mathcal{O}(\omega^2),
\end{aligned}
\end{equation}
similarly,
\begin{equation}\label{1622}
\begin{aligned}
&\int_{S^b \cup\iota_\delta(Q)} \mathbf\Gamma\left(\mathbf{x}-\mathbf{z}_{\mathbf{y}}\right) \boldsymbol{\psi}(\mathbf{y}) d s({\mathbf{y}})=\int_{S^b \cup\iota_\delta(Q)} \mathbf\Gamma\left(\mathbf{x}-Q\right) \boldsymbol{\psi}(\mathbf{y}) d s(\mathbf{y})+o(\delta^2)\\
=&\pi\delta^2(\lambda_1-\frac{1}{2})^{-1}\mathbf\Gamma\left(\mathbf{x}-Q\right)\left[\lambda(\nabla\cdot \mathbf{H}_0(Q))\mathbf{e}_1+2\mu\nabla^s \mathbf{H}_0(Q)\mathbf{e}_1\right]\\&+\omega\pi\delta^2(\lambda_1-\frac{1}{2})^{-1}\mathbf\Gamma\left(\mathbf{x}-Q\right)\left[\lambda(\nabla\cdot \mathbf{H}_1(Q))\mathbf{e}_1+2\mu\nabla^s \mathbf{H}_1(Q)\mathbf{e}_1\right]+o(\delta^2)+\delta^2\cdot\mathcal{O}(\omega^2).
\end{aligned}
\end{equation}
Substituting \eqref{1619}-\eqref{1622} into \eqref{1623}, we can obtain \eqref{u}. By using some elementary vectorial calculus we obtain the conclusion. This completes the proof.
\end{proof}


\section{Mathematical analysis of polariton resonance}

\begin{proof}[Proof of Theorem \ref{Thm11}]
(1) Since $c=-\frac{\lambda+\mu}{\lambda+3\mu}+i\varrho$, recalling that $\lambda_1=\frac{c+1}{2(c-1)}$, by direct calculations, we can obtain
$$
\lambda_1=-\frac{\mu}{2(\lambda+2\mu)}-\frac{i\varrho(\lambda+3\mu)^2}{4(\lambda+2\mu)^2}+\mathcal{O}(\varrho^2),
$$
from \eqref{l1}, we have
\begin{equation}\label{eq211a}
\begin{aligned}
l_2=&\frac{1}{2\lambda_1(\lambda+2\mu)+\mu}=\frac{1}{2\big(-\frac{\mu}{2(\lambda+2\mu)}-\frac{i\varrho(\lambda+3\mu)^2}{4(\lambda+2\mu)^2}
+\mathcal{O}(\varrho^2)\big)\big(\lambda+2\mu\big)+\mu}\\=&\frac{1}{-\frac{i\varrho(\lambda+3\mu)^2}{2(\lambda+2\mu)}+\mathcal{O}(\varrho^2)}
=a_1\frac{1}{\varrho}+\mathcal{O}(1),
\end{aligned}
\end{equation}
where $a_1=\frac{2i(\lambda+2\mu)}{(\lambda+3\mu)^2}$. From Theorem \ref{3main}, we can decompose $\mathbf{u}^s(\mathbf{x})$ into singular and regular parts
\begin{equation}\label{eq212a}
\mathbf{u}^s(\mathbf{x})=l_2\cdot\mathbf{B}^1(\mathbf{x})+\mathbf{J}^1(\mathbf{x}),
\end{equation}
where
\begin{equation}\label{B1111}
\begin{aligned}
\mathbf{B}^1(\mathbf{x})
=&\delta^2\mathcal{A}_1[\mathbf{G}^{11}_{\mathbf{H}_0}+\omega\mathbf{G}^{11}_{\mathbf{H}_1}](\mathbf{x})
-\delta^2\mathcal{A}_2[\mathbf{R}^{11}_{\mathbf{H}_0}+\omega\mathbf{R}^{11}_{\mathbf{H}_1}](\mathbf{x})
\\&+\alpha_3\omega\delta^2\int_{-L/2}^{L/2}\mathbf{G}^{21}_{\mathbf{H}_0}(\mathbf{y})dy,
\end{aligned}
\end{equation}
and
\begin{equation}
\begin{aligned}
&\mathbf{J}^1(\mathbf{x})\\
=&\delta^2\mathcal{A}_1[\mathbf{G}^{12}_{\mathbf{H}_0}+\omega\mathbf{G}^{12}_{\mathbf{H}_1}](\mathbf{x})
-\delta^2\mathcal{A}_2[\mathbf{R}^{12}_{\mathbf{H}_0}+\omega\mathbf{R}^{12}_{\mathbf{H}_1}](\mathbf{x})\\
&-\pi\delta^2(\lambda_1-\frac{1}{2})^{-1}\mathbf\Gamma\left(\mathbf{x}-P\right)\left[\lambda\Big(\nabla\cdot \big(\mathbf{H}_0(P)+\omega\mathbf{H}_1(P)\big)\Big)\mathbf{e}_1+2\mu\nabla^s\big(\mathbf{H}_0(P)+\omega\mathbf{H}_1(P)\big)\mathbf{e}_1\right]
\\&+\pi\delta^2(\lambda_1-\frac{1}{2})^{-1}\mathbf\Gamma\left(\mathbf{x}-Q\right)\left[\lambda\Big(\nabla\cdot \big(\mathbf{H}_0(Q)+\omega\mathbf{H}_1(Q)\big)\Big)\mathbf{e}_1+2\mu\nabla^s\big(\mathbf{H}_0(Q)+\omega\mathbf{H}_1(Q)\big)\mathbf{e}_1\right]\\
&+\alpha_3\omega\delta^2\int_{-L/2}^{L/2}\mathbf{G}^{12}_{\mathbf{H}_0}(\mathbf{y}) dy-\alpha_3\pi\omega\delta^2(\lambda_1-\frac{1}{2})^{-1}\left[\lambda(\nabla\cdot \mathbf{H}_0(P))\mathbf{e}_1+2\mu\nabla^s \mathbf{H}_0(P)\mathbf{e}_1\right]\\
&+\alpha_3\pi\omega\delta^2(\lambda_1-\frac{1}{2})^{-1}\left[\lambda(\nabla\cdot \mathbf{H}_0(Q))\mathbf{e}_1+2\mu\nabla^s \mathbf{H}_0(Q)\mathbf{e}_1\right]\\
&+o(\delta^2)+\omega\cdot o(\delta^2)+\delta\cdot\mathcal{O}(\omega^2),
   \end{aligned}
\end{equation}
with $\mathbf{G}^{11}_{\mathbf{H}_i}, \mathbf{G}^{12}_{\mathbf{H}_i} $, $\mathbf{R}^{11}_{\mathbf{H}_i}$ and $\mathbf{R}^{12}_{\mathbf{H}_i}$, $i=0, 1$ are given by
\begin{equation}\label{eqG11}
\begin{aligned}
\mathbf{G}^{11}_{\mathbf{H}_i}=& 4\pi\mu^2l_1\Big((\lambda+2\mu)+ \frac{\lambda+\mu}{2\lambda_1-1}\Big)B\nabla\nabla\cdot \big(B \mathbf{H}_i(\mathbf{z}_{\mathbf{y}})\big)\\
& -\frac{\mu^2\pi}{\lambda_1}B\Big(\big(B\nabla\big)^T\big(\nabla\mathbf{H}_i(\mathbf{z}_{\mathbf{y}})+\nabla\mathbf{H}_i(\mathbf{z}_{\mathbf{y}})^T\big)\Big)^T\\
& -\frac{\mu^2\pi}{\lambda_1(2\lambda_1-1)}B\Big(\Delta \mathbf{H}_i(\mathbf{z}_{\mathbf{y}})-\partial_{y_1}^2 \mathbf{H}_i(\mathbf{z}_{\mathbf{y}})+ \nabla\nabla\cdot \big(B \mathbf{H}_i(\mathbf{z}_{\mathbf{y}})\big)\Big)\\
& -\frac{2\mu^2\pi}{2\lambda_1-1}C\nabla\Big((\nabla\times \mathbf{H}_i(\mathbf{z}_{\mathbf{y}}))\cdot \mathbf{e}_1\Big),
\end{aligned}
\end{equation}
\begin{equation}
\begin{aligned}
\mathbf{G}^{12}_{\mathbf{H}_i}=& 2\pi\lambda l_1\Big((\lambda+2\mu)+\frac{\lambda+\mu}{2\lambda_1-1}\Big) B\nabla\nabla\cdot\mathbf{H}_i(\mathbf{z}_{\mathbf{y}})\\
& + \frac{\mu\pi}{\lambda_1}\Big(\big(B\nabla\big)^T\big(\nabla\mathbf{H}_i(\mathbf{z}_{\mathbf{y}})+\nabla\mathbf{H}_i(\mathbf{z}_{\mathbf{y}})^T\big)\Big)^T\\
& +\frac{\mu\pi}{\lambda_1(2\lambda_1-1)}\Big(\Delta \mathbf{H}_i(\mathbf{z}_{\mathbf{y}})-\partial_{y}^2 \mathbf{H}_i(\mathbf{z}_{\mathbf{y}})+ \nabla\nabla\cdot \big(B \mathbf{H}_i(\mathbf{z}_{\mathbf{y}})\big)\Big),
\end{aligned}
\end{equation}
\begin{equation}\label{eqR11}
\begin{aligned}
\mathbf{R}^{11}_{\mathbf{H}_i}=&-\frac{2\pi\mu(\lambda-\mu)(\lambda+2\mu)}{2(\lambda+2\mu)\lambda_1+\lambda}
\nabla\cdot\big(B\mathbf{H}_i(\mathbf{z}_{\mathbf{y}})\big)I_3-\frac{\mu^2\pi}{\lambda_1} B\big(\nabla \mathbf{H}_i(\mathbf{z}_{\mathbf{y}})+\nabla \mathbf{H}_i(\mathbf{z}_{\mathbf{y}})^T\big),
\end{aligned}
\end{equation}
and
\begin{equation}
\begin{aligned}
\mathbf{R}^{12}_{\mathbf{H}_i}=&\pi l_{01}\nabla\cdot \mathbf{H}_i(\mathbf{z}_{\mathbf{y}})I_3 +\frac{\mu\pi}{\lambda_1}(\nabla \mathbf{H}_i(\mathbf{z}_{\mathbf{y}})+\nabla \mathbf{H}_i(\mathbf{z}_{\mathbf{y}})^T).
\end{aligned}
\end{equation}
Since $\mathbf{B}^1(\mathbf{x})$ and $\mathbf{J}^1(\mathbf{x})$ are bounded, by \eqref{eq211a}-\eqref{eq212a}, it holds
$$
\mathbf{u}^s(\mathbf{x})\sim \frac{a_1}{\varrho}\cdot\mathbf{B}^1(\mathbf{x}).
$$
From \eqref{Eu}, one has
$$
\begin{aligned}
E(\mathbf{u}^s)=&\int_{\mathbb{R}^{3} \backslash \overline{D}}\nabla^s \overline{\mathbf{u}^s(\mathbf{x})}: \mathcal{C} \nabla^s \mathbf{u}^s(\mathbf{x}) d \mathbf{x} \sim  \frac{|a_1|^2}{\varrho^2}\int_{\mathbb{R}^{3} \backslash \overline{D}}\nabla^s \overline{\mathbf{B}^1(\mathbf{x})}: \mathcal{C} \nabla^s \mathbf{B}^1(\mathbf{x}) d \mathbf{x}.
\end{aligned}
$$
By \eqref{B1111} and Lemma \ref{le12345} (1), we can obtain the result. \\
The proof of (2) is similar to that of (1). \\
(3)Define $\zeta=c_0+1+i\varrho$. Since $c=c_0+i\varrho$, where $c_0\rightarrow -1$, then $\zeta=c+1\rightarrow 0$. By direct calculations, one has
$$
\lambda_1=\frac{\zeta}{2\zeta-4}=-\frac{\zeta}{4}\big(1+\mathcal{O}(\zeta)\big),
$$
then
\begin{equation}\label{eq411}
\frac{1}{\lambda_1}=-\frac{4}{\zeta}+\mathcal{O}(1).
\end{equation}
From Theorem \ref{3main}, we can decompose $\mathbf{u}^s(\mathbf{x})$ into singular and regular parts
\begin{equation}\label{eq412}
\mathbf{u}^s(\mathbf{x})=\frac{1}{\lambda_1}\cdot\mathbf{B}^{3}(\mathbf{x})+\mathbf{J}^{3}(\mathbf{x}),
\end{equation}
where
\begin{equation}\label{B4}
\begin{aligned}
\mathbf{B}^3(\mathbf{x})
=&\delta^2\mathcal{A}_1[\mathbf{G}^{31}_{\mathbf{H}_0}+\omega\mathbf{G}^{31}_{\mathbf{H}_1}](\mathbf{x})
-\delta^2\mathcal{A}_2[\mathbf{R}^{31}_{\mathbf{H}_0}+\omega\mathbf{R}^{31}_{\mathbf{H}_1}](\mathbf{x})\\&+\alpha_3\omega\delta^2\int_{-L/2}^{L/2}\mathbf{G}^{31}_{\mathbf{H}_0}(y) dy,
\end{aligned}
\end{equation}
and
\begin{equation}
\begin{aligned}
&\mathbf{J}^3(\mathbf{x})
=\delta^2\mathcal{A}_1[\mathbf{G}^{32}_{\mathbf{H}_0}+\omega\mathbf{G}^{32}_{\mathbf{H}_1}](\mathbf{x})
-\delta^2\mathcal{A}_2[\mathbf{R}^{32}_{\mathbf{H}_0}+\omega\mathbf{R}^{32}_{\mathbf{H}_1}](\mathbf{x})\\
&-\pi\delta^2(\lambda_1-\frac{1}{2})^{-1}\mathbf\Gamma\left(\mathbf{x}-P\right)\left[\lambda\Big(\nabla\cdot \big(\mathbf{H}_0(P)+\omega\mathbf{H}_1(P)\big)\Big)\mathbf{e}_1+2\mu\nabla^s\big(\mathbf{H}_0(P)+\omega\mathbf{H}_1(P)\big)\mathbf{e}_1\right]
\\&+\pi\delta^2(\lambda_1-\frac{1}{2})^{-1}\mathbf\Gamma\left(\mathbf{x}-Q\right)\left[\lambda\Big(\nabla\cdot \big(\mathbf{H}_0(Q)+\omega\mathbf{H}_1(Q)\big)\Big)\mathbf{e}_1+2\mu\nabla^s\big(\mathbf{H}_0(Q)+\omega\mathbf{H}_1(Q)\big)\mathbf{e}_1\right]\\
&+\alpha_3\omega\delta^2\int_{-L/2}^{L/2}\mathbf{G}^{32}_{\mathbf{H}_0}(y) dy-\alpha_3\pi\omega\delta^2(\lambda_1-\frac{1}{2})^{-1}\left[\lambda(\nabla\cdot \mathbf{H}_0(P))\mathbf{e}_1+2\mu\nabla^s \mathbf{H}_0(P)\mathbf{e}_1\right]\\
&+\alpha_3\pi\omega\delta^2(\lambda_1-\frac{1}{2})^{-1}\left[\lambda(\nabla\cdot \mathbf{H}_0(Q))\mathbf{e}_1+2\mu\nabla^s \mathbf{H}_0(Q)\mathbf{e}_1\right]\\
&+o(\delta^2)+\omega\cdot o(\delta^2)+\delta\cdot\mathcal{O}(\omega^2),
   \end{aligned}
\end{equation}
with $\mathbf{G}^{31}_{\mathbf{H}_i}, \mathbf{G}^{32}_{\mathbf{H}_i} $, $\mathbf{R}^{31}_{\mathbf{H}_i}$ and $\mathbf{R}^{32}_{\mathbf{H}_i}$, $i=0, 1$ are given by
\begin{equation}\label{GH4}
\begin{aligned}
\mathbf{G}^{31}_{\mathbf{H}_i}=&\mu\pi\Big(I_3-\mu l_2B\Big)\Big(\big(B\nabla\big)^T\big(\nabla\mathbf{H}_i(\mathbf{z}_{\mathbf{y}})+\nabla\mathbf{H}_i(\mathbf{z}_{\mathbf{y}})^T\big)\Big)^T\\
& +\frac{\mu\pi}{(2\lambda_1-1)}\Big(I_3-\mu l_2B\Big)\Big(\Delta \mathbf{H}_i(\mathbf{z}_{\mathbf{y}})-\partial_{y}^2 \mathbf{H}_i(\mathbf{z}_{\mathbf{y}})+ \nabla\nabla\cdot \big(B \mathbf{H}_i(\mathbf{z}_{\mathbf{y}})\big)\Big),
\end{aligned}
\end{equation}
\begin{equation}
\begin{aligned}
\mathbf{G}^{32}_{\mathbf{H}_i}=&2\pi\Big((\lambda+2\mu)+ \frac{\lambda+\mu}{2\lambda_1-1}\Big)l_1 B\Big(\lambda\nabla\nabla\cdot\mathbf{H}_i(\mathbf{z}_{\mathbf{y}})+2\mu^2l_2\nabla\nabla\cdot \big(B \mathbf{H}_i(\mathbf{z}_{\mathbf{y}})\big)\Big)\\
& -\frac{2\pi\mu^2}{2\lambda_1-1}l_2C\nabla\Big((\nabla\times \mathbf{H}_i(\mathbf{z}_{\mathbf{y}}))\cdot \mathbf{e}_1\Big),
\end{aligned}
\end{equation}
\begin{equation}\label{RH4}
\begin{aligned}
\mathbf{R}^{31}_{\mathbf{H}_i}=\mu\pi\big(I_3-\mu l_2 B\big)\big(\nabla \mathbf{H}_i(\mathbf{z}_{\mathbf{y}})+\nabla \mathbf{H}_i(\mathbf{z}_{\mathbf{y}})^T\big),
\end{aligned}
\end{equation}
and
\begin{equation}
\begin{aligned}
\mathbf{R}^{32}_{\mathbf{H}_i}=&\pi\big(l_{01}\nabla\cdot \mathbf{H}_i(\mathbf{z}_{\mathbf{y}})+l_{02}\nabla\cdot(B\mathbf{H}_i(\mathbf{z}_{\mathbf{y}}))\big)I_3.
\end{aligned}
\end{equation}
Since $\mathbf{B}^{3}(\mathbf{x})$ and $\mathbf{J}^{3}(\mathbf{x})$ are bounded, by \eqref{eq411} and \eqref{eq412}, it holds
$$
\mathbf{u}^s(\mathbf{x})\sim -\frac{4}{\zeta}\cdot\mathbf{B}^{3}(\mathbf{x}).
$$
From \eqref{Eu}, one has
$$
\begin{aligned}
E(\mathbf{u}^s)=&\int_{\mathbb{R}^{3} \backslash \overline{D}}\nabla^s \overline{\mathbf{u}^s(\mathbf{x})}: \mathcal{C} \nabla^s \mathbf{u}^s(\mathbf{x}) d \mathbf{x}\\ \sim & \frac{16}{|\zeta|^2}\int_{\mathbb{R}^{3} \backslash \overline{D}}\nabla^s \overline{\mathbf{B}^{3}(\mathbf{x})}: \mathcal{C} \nabla^s \mathbf{B}^{3}(\mathbf{x}) d \mathbf{x}.
\end{aligned}
$$
From \eqref{B4} and Lemma \ref{le12345} (3), we can obtain the result.\\
(4)$c=c_0+i\varrho$, where $c_0\rightarrow +\infty$, by direct calculation, it holds
\begin{equation}\label{eq511}
(\lambda_1-\frac{1}{2})^{-1}=c-1.
\end{equation}
From Theorem \ref{3main}, we can decompose $\mathbf{u}^s(\mathbf{x})$ into singular and regular parts
\begin{equation}\label{eq512}
\mathbf{u}^s(\mathbf{x})=(\lambda_1-\frac{1}{2})^{-1}\cdot\mathbf{B}^{4}(\mathbf{x})+\mathbf{J}^{4}(\mathbf{x}),
\end{equation}
where
\begin{equation}\label{B5}
\begin{aligned}
\mathbf{B}^4(\mathbf{x})
=&\delta^2\mathcal{A}_1[\mathbf{G}^{41}_{\mathbf{H}_0}+\omega\mathbf{G}^{41}_{\mathbf{H}_1}](\mathbf{x})\\
&-\pi\delta^2\mathbf\Gamma\left(\mathbf{x}-P\right)\left[\lambda\Big(\nabla\cdot \big(\mathbf{H}_0(P)+\omega\mathbf{H}_1(P)\big)\Big)\mathbf{e}_1+2\mu\nabla^s\big(\mathbf{H}_0(P)+\omega\mathbf{H}_1(P)\big)\mathbf{e}_1\right]
\\&+\pi\delta^2\mathbf\Gamma\left(\mathbf{x}-Q\right)\left[\lambda\Big(\nabla\cdot \big(\mathbf{H}_0(Q)+\omega\mathbf{H}_1(Q)\big)\Big)\mathbf{e}_1+2\mu\nabla^s\big(\mathbf{H}_0(Q)+\omega\mathbf{H}_1(Q)\big)\mathbf{e}_1\right]\\
&+\alpha_3\omega\delta^2\int_{-L/2}^{L/2}\mathbf{G}^{41}_{\mathbf{H}_0}(y)dy-\alpha_3\pi\omega\delta^2\left[\lambda(\nabla\cdot \mathbf{H}_0(P))\mathbf{e}_1+2\mu\nabla^s \mathbf{H}_0(P)\mathbf{e}_1\right]\\
&+\alpha_3\pi\omega\delta^2\left[\lambda(\nabla\cdot \mathbf{H}_0(Q))\mathbf{e}_1+2\mu\nabla^s \mathbf{H}_0(Q)\mathbf{e}_1\right],
\end{aligned}
\end{equation}
and
\begin{equation}
\begin{aligned}
\mathbf{J}^4(\mathbf{x})
=&\delta^2\mathcal{A}_1[\mathbf{G}^{42}_{\mathbf{H}_0}+\omega\mathbf{G}^{42}_{\mathbf{H}_1}](\mathbf{x})
-\delta^2\mathcal{A}_2[\mathbf{R}_{\mathbf{H}_0}+\omega\mathbf{R}_{\mathbf{H}_1}](\mathbf{x})\\
&+\alpha_3\omega\delta^2\int_{-L/2}^{L/2}\mathbf{G}^{42}_{\mathbf{H}_0}(y)dy\\
&+o(\delta^2)+\omega\cdot o(\delta^2)+\delta\cdot\mathcal{O}(\omega^2),
   \end{aligned}
\end{equation}
with
\begin{equation}\label{GH5}
\begin{aligned}
\mathbf{G}^{41}_{\mathbf{H}_i}=&\pi(\lambda+\mu)l_1 B\Big(\lambda\nabla\nabla\cdot\mathbf{H}_i(\mathbf{z}_{\mathbf{y}})+2\mu^2l_2\nabla\nabla\cdot \big(B \mathbf{H}_i(\mathbf{z}_{\mathbf{y}})\big)\Big)\\
& +\frac{\mu\pi}{2\lambda_1}\Big(I_3-\mu l_2B\Big)\Big(\Delta \mathbf{H}_i(\mathbf{z}_{\mathbf{y}})-\partial_{y}^2 \mathbf{H}_i(\mathbf{z}_{\mathbf{y}})+ \nabla\nabla\cdot \big(B \mathbf{H}_i(\mathbf{z}_{\mathbf{y}})\big)\Big)\\
& -\pi\mu^2l_2C\nabla\Big((\nabla\times \mathbf{H}_i(\mathbf{z}_{\mathbf{y}}))\cdot \mathbf{e}_1\Big), \quad i=0, 1,
\end{aligned}
\end{equation}
and
\begin{equation}
\begin{aligned}
\mathbf{G}^{42}_{\mathbf{H}_i}=& 2\pi(\lambda+2\mu)l_1 B\Big(\lambda\nabla\nabla\cdot\mathbf{H}_i(\mathbf{z}_{\mathbf{y}})+2\mu^2l_2\nabla\nabla\cdot \big(B \mathbf{H}_i(\mathbf{z}_{\mathbf{y}})\big)\Big)\\
& + \frac{\mu\pi}{\lambda_1}\Big(I_3-\mu l_2B\Big)\Big(\big(B\nabla\big)^T\big(\nabla\mathbf{H}_i(\mathbf{z}_{\mathbf{y}})+\nabla\mathbf{H}_i(\mathbf{z}_{\mathbf{y}})^T\big)\Big)^T, \quad i=0, 1.
\end{aligned}
\end{equation}
Since $\mathbf{B}^{4}(\mathbf{x})$ and $\mathbf{J}^{4}(\mathbf{x})$ are bounded, by \eqref{eq511} and \eqref{eq512}, it holds
$$
\mathbf{u}^s(\mathbf{x})\sim (c-1)\cdot\mathbf{B}^{4}(\mathbf{x}).
$$
From \eqref{Eu}, one has
$$
\begin{aligned}
E(\mathbf{u}^s)=&\int_{\mathbb{R}^{3} \backslash \overline{D}}\nabla^s \overline{\mathbf{u}^s(\mathbf{x})}: \mathcal{C} \nabla^s \mathbf{u}^s(\mathbf{x}) d \mathbf{x}\\ \sim & |c-1|^2\int_{\mathbb{R}^{3} \backslash \overline{D}}\nabla^s \overline{\mathbf{B}^{4}(\mathbf{x})}: \mathcal{C} \nabla^s \mathbf{B}^{4}(\mathbf{x}) d \mathbf{x}.
\end{aligned}
$$
From \eqref{B5} and Lemma \ref{le12345} (4), we can obtain the result.
This completes the proof.
\end{proof}
\begin{thm}
{\rm Consider $c=c_0+i\varrho$, where $\varrho>0$ is the loss parameter. Let $\mathbf{u}^s$ be the scattering solution of \eqref{3equa}.\\
(1) Suppose $\mathbf{H}_0(\mathbf{x})=\left(2x_1x_2, x_1^2-x_2^2, 0\right)^T$. If $c_0=-\frac{\lambda+\mu}{\lambda+3\mu}$,
then it holds
$$
E(\mathbf{u}^s)\sim\frac{4(\lambda+2\mu)^2\delta^4}{(\lambda+3\mu)^4\varrho^2}\int_{\mathbb{R}^{3} \backslash \overline{D}}\nabla^s \overline{\mathbf{B}_0^1(\mathbf{x})}: \mathcal{C} \nabla^s \mathbf{B}_0^1(\mathbf{x}) d \mathbf{x}=\mathcal{O}(\varrho^{-2}\delta^4).
$$
Furthermore, if $\varrho=o\left(\delta^2\right)$ (as $\delta \rightarrow 0, \varrho \rightarrow 0$ ), then it follows that
$$
\sqrt{E(\mathbf{u}^s)} \rightarrow \infty.
$$
(2) Suppose $\mathbf{H}_0(\mathbf{x})=\left( x_1^2- x_2^2,-2 x_1 x_2, 0\right)^T$. If $c_0=-\frac{\mu}{\lambda+\mu}$,
then it holds
$$
E(\mathbf{u}^s)\sim\frac{\lambda^2(\lambda+2\mu)^4\delta^4}{(\lambda+\mu)^4\varrho^2}\int_{\mathbb{R}^{3} \backslash \overline{D}}\nabla^s \overline{\mathcal{A}_2[\mathbf{R}^{21}_{\mathbf{H}_0}](\mathbf{x})}: \mathcal{C} \nabla^s \mathcal{A}_2[\mathbf{R}^{21}_{\mathbf{H}_0}](\mathbf{x}) d \mathbf{x}=\mathcal{O}(\varrho^{-2}\delta^4).
$$
Furthermore, if $\varrho=o\left(\delta^2\right)$ (as $\delta \rightarrow 0, \varrho \rightarrow 0$ ), then it follows that
$$
\sqrt{E(\mathbf{u}^s)} \rightarrow \infty.
$$
(3) Suppose $\mathbf{H}_0(\mathbf{x})=\left(2x_1x_2, x_1^2-x_2^2, 0\right)^T$. If $c_0\rightarrow -1$,
then it holds
$$
E(\mathbf{u}^s)\sim\frac{16\delta^4}{(c_0+1)^2+\varrho^2}\int_{\mathbb{R}^{3} \backslash \overline{D}}\nabla^s \overline{\mathbf{B}_0^3(\mathbf{x})}: \mathcal{C} \nabla^s \mathbf{B}_0^3(\mathbf{x}) d \mathbf{x}=\mathcal{O}\left(\frac{\delta^4}{(c_0+1)^2+\varrho^2}\right).
$$
Furthermore, if $(c_0+1)^2+\varrho^2=o\left(\delta^4\right)$ (as $\delta \rightarrow 0, \varrho \rightarrow 0$ ), then it follows that
$$
\sqrt{E(\mathbf{u}^s)} \rightarrow \infty.
$$
(4) Suppose $\mathbf{H}_0(\mathbf{x})=\left(2x_1x_2, x_1^2-x_2^2, 0\right)^T$. If $c_0\rightarrow -\infty$,
then it holds
$$
E(\mathbf{u}^s)\sim|c-1|^2\delta^4\int_{\mathbb{R}^{3} \backslash \overline{D}}\nabla^s \overline{\mathbf{B}_0^4(\mathbf{x})}: \mathcal{C} \nabla^s \mathbf{B}_0^4(\mathbf{x}) d \mathbf{x}=\mathcal{O}(\delta^4c_0^2).
$$
Furthermore, if $c_0^{-2}=o\left(\delta^4\right)$ (as $\delta \rightarrow 0$ ), then it follows that
$$
\sqrt{E(\mathbf{u}^s)} \rightarrow \infty.
$$
}
\end{thm}
\begin{proof}
We only prove (1), the others can be proved similarly.\\
(1)Since $\mathbf{H}_0(\mathbf{x})=\left(2x_1x_2, x_1^2-x_2^2, 0\right)^T$, from \eqref{eqG11}, we can obtain
$$
\mathbf{G}^{11}_{\mathbf{H}_0}(y) = -8\pi\mu^2 \left[ l_1(\lambda+2\mu) + \frac{l_1(\lambda+\mu)-1}{2\lambda_1-1} \right] \mathbf{e}_2\neq \mathbf{0}.
$$
From Theorem \ref{Thm11} (1), we can obtain the result. This completes the proof.
\end{proof}

\section*{Acknowledgment}
The work of Y. Deng was supported by NSFC-RGC Joint Research Grant No. 12161160314 and the Natural Science
Foundation Innovation Research Team Project of Guangxi (Grant No. 2025GXNSFGA069001) . The work of H. Liu was supported by the NSFC/RGC Joint Research Scheme, N\_CityU101/21; ANR/RGC Joint Research Scheme, A\_CityU203/19; and the Hong Kong RGC General Research Funds (projects 11311122, 12301420 and 11300821). The work of W. Tang was supported by the China Postdoctoral Science Foundation (Grant No. GZC20252030). The work of G. Zheng was supported by the NSF of China (12271151).

\appendix
\section{Auxiliary result}
\begin{lem}\label{le12345}
{\rm Regarding the functions $\mathbf{B}_0^1(\mathbf{x}), \mathcal{A}_2[\mathbf{R}^{21}_{\mathbf{H}_0}](\mathbf{x}), \mathbf{B}_0^3(\mathbf{x})$ and $\mathbf{B}_0^4(\mathbf{x})$, which appear in Theorem \ref{Thm11}, and the space $\boldsymbol{\Psi}$, defined in \eqref{Psi}, we have the following results:\\
(1)The function $\mathbf{B}_0^1(\mathbf{x})$ does not belong to $\boldsymbol{\Psi}$ in $\mathbb{R}^3\setminus\overline{D}$  if and only if its source densities $\mathbf{G}^{11}_{\mathbf{H}_0}(y)$ or $\mathbf{R}^{11}_{\mathbf{H}_0}(y)\mathbf{e}_i$ (for $i=2, 3$) are not identically zero on $[-L/2, L/2]$. i.e.,
$$\mathbf{B}_0^1(\mathbf{x}) \notin \boldsymbol{\Psi} \text{ in } \mathbb{R}^3\setminus\overline{D} \Longleftrightarrow \left( \mathbf{G}^{11}_{\mathbf{H}_0}(y) \not\equiv \mathbf{0} \text{ or } \mathbf{R}^{11}_{\mathbf{H}_0}(y)\mathbf{e}_i \not\equiv \mathbf{0} \text{ for } i=2, 3 \right) \text{ on } [-L/2, L/2].$$
(2)The function $\mathcal{A}_2[\mathbf{R}^{21}_{\mathbf{H}_0}](\mathbf{x})$ does not belong to $\boldsymbol{\Psi}$ in $\mathbb{R}^3\setminus\overline{D}$ if and only if its source densities $\mathbf{R}^{21}_{\mathbf{H}_0}(y)\mathbf{e}_i$ (for $i=2, 3$) are not identically zero on $[-L/2, L/2]$. i.e.,
$$\mathcal{A}_2[\mathbf{R}^{21}_{\mathbf{H}_0}](\mathbf{x}) \notin \boldsymbol{\Psi} \text{ in } \mathbb{R}^3\setminus\overline{D} \Longleftrightarrow \mathbf{R}^{21}_{\mathbf{H}_0}(y)\mathbf{e}_i \not\equiv \mathbf{0} \text{ for } i=2, 3 \text{ on } [-L/2, L/2].$$
(3)The function $\mathbf{B}_0^3(\mathbf{x})$ does not belong to $\boldsymbol{\Psi}$ in $\mathbb{R}^3\setminus\overline{D}$  if and only if its source densities $\mathbf{G}^{31}_{\mathbf{H}_0}(y)$ or $\mathbf{R}^{31}_{\mathbf{H}_0}(y)\mathbf{e}_i$ (for $i=2, 3$) are not identically zero on $[-L/2, L/2]$. i.e.,
$$\mathbf{B}_0^3(\mathbf{x}) \notin \boldsymbol{\Psi} \text{ in } \mathbb{R}^3\setminus\overline{D} \Longleftrightarrow \left( \mathbf{G}^{31}_{\mathbf{H}_0}(y) \not\equiv \mathbf{0} \text{ or } \mathbf{R}^{31}_{\mathbf{H}_0}(y)\mathbf{e}_i \not\equiv \mathbf{0} \text{ for } i=2, 3 \right) \text{ on } [-L/2, L/2].$$
(4)The function $\mathbf{B}_0^4(\mathbf{x})$ does not belong to $\boldsymbol{\Psi}$ in $\mathbb{R}^3\setminus\overline{D}$  if and only if its source density $\mathbf{G}^{41}_{\mathbf{H}_0}$ is not identically zero on $[-L/2, L/2]$ or $\mathbf{C}_P\neq \mathbf{0}$ or $\mathbf{C}_Q\neq \mathbf{0}$. i.e.,
$$
\mathbf{B}_0^4(\mathbf{x}) \notin \boldsymbol{\Psi} \text{ in } \mathbb{R}^3\setminus\overline{D} \Longleftrightarrow \mathbf{G}^{41}_{\mathbf{H}_0}(y) \not\equiv \mathbf{0} \text{ on } [-L/2, L/2] \text{ or } \mathbf{C}_P \neq \mathbf{0} \text{ or } \mathbf{C}_Q \neq \mathbf{0}.
$$
}
\end{lem}
\begin{proof}
We only prove (1), the others can be proved similarly.
We prove the contrapositive, i.e.,
$$
\mathbf{B}_0^1(\mathbf{x}) \in \boldsymbol{\Psi} \text{ in } \mathbb{R}^3\setminus\overline{D} \iff \mathbf{G}^{11}_{\mathbf{H}_0}(y) \equiv \mathbf{0} \text{ and } \mathbf{R}^{11}_{\mathbf{H}_0}(y)\mathbf{e}_i \equiv \mathbf{0} \text{ for } i=2,3 \text{ on } [-L/2, L/2].
$$
$\Leftarrow$: If $\mathbf{G}^{11}_{\mathbf{H}_0}(y)\equiv\mathbf{0}$ and the second and third columns of $\mathbf{R}^{11}_{\mathbf{H}_0}(y)$ are identically zero, then by the definition \eqref{AA1}-\eqref{AA2} of the operators, $\mathcal{A}_1[\mathbf{G}^{11}_{\mathbf{H}_0}] \equiv \mathbf{0}$ and $\mathcal{A}_2[\mathbf{R}^{11}_{\mathbf{H}_0}] \equiv \mathbf{0}$. Consequently, $\mathbf{B}_0^1(\mathbf{x}) \equiv \mathbf{0}$. From \eqref{Psi} , we know that $\mathbf{B}_0^1(\mathbf{x})\in \boldsymbol{\Psi}$ in $\mathbb{R}^3\setminus\overline{D}$.\\
$\Rightarrow$: Since $\mathbf{B}_0^1(\mathbf{x})\in \boldsymbol{\Psi}$ in $\mathbb{R}^3\setminus\overline{D}$, it follows that
$$
\mathbf{B}_0^1(\mathbf{x})=\mathbf{a}+\mathbf{M}\mathbf{x},
$$
where $\mathbf{a}$ is a constant translation vector and $\mathbf{M}$ is a constant skew-symmetric matrix. From the definitions of $\mathcal{A}_1$ and $\mathcal{A}_2$ in \eqref{AA1}-\eqref{AA2}, we know
$$
\lim _{|\mathbf{x}| \rightarrow \infty} \mathbf{B}_0^1(\mathbf{x})=\mathbf{0}.
$$
Substituting the form of $\mathbf{B}_0^1(\mathbf{x})$, we have
$$
\lim _{|\mathbf{x}| \rightarrow \infty}(\mathbf{a}+\mathbf{M}\mathbf{x})=\mathbf{0}.
$$
Since $\mathbf{a}$ is constant, for this limit to hold as $|\mathbf{x}| \rightarrow \infty$, we must have $\mathbf{M}=\mathbf{0}$. This then implies
$$
\lim _{|\mathbf{x}| \rightarrow \infty} \mathbf{a}=\mathbf{a}=\mathbf{0}.
$$
Thus, $\mathbf{B}_0^1(\mathbf{x})=\mathcal{A}_1[\mathbf{G}^{11}_{\mathbf{H}_0}](\mathbf{x})-\mathcal{A}_2[\mathbf{R}^{11}_{\mathbf{H}_0}](\mathbf{x})\equiv\mathbf{0}$ in $\mathbb{R}^3\setminus\overline{D}$.
According to the principle of analytic continuation, we have $\mathbf{B}_0^1(\mathbf{x})\equiv\mathbf{0}$ in $\mathbb{R}^3\setminus\Gamma_{0}$. We now prove the implication by contradiction.\\
(i) Assume that $\mathbf{G}^{11}_{\mathbf{H}_0}(y)\not\equiv \mathbf{0}$ and $\mathbf{R}^{11}_{\mathbf{H}_0}(y)\mathbf{e}_i \equiv \mathbf{0}$ for $i=2,3$ on $[-L/2, L/2]$. By this assumption, there must exist some point $x_0\in (-L/2, L/2)$ such that $\mathbf{G}^{11}_{\mathbf{H}_0}(x_0)\neq \mathbf{0}$. Let us analyze the behavior of $\mathcal{A}_1[\mathbf{G}^{11}_{\mathbf{H}_0}](\mathbf{x})$ as $\mathbf{x}$ approaches the point $\mathbf{z}_{x_0}=(x_0, 0, 0)^T$ on the line segment $\Gamma_0$. We define a local coordinate system by letting $\mathbf{x}=(x_0, x_2, x_3)^T$ and defining $\rho = \sqrt{x_2^2 + x_3^2}$. Note that $\mathbf{x} \to \mathbf{z}_{x_0}$ as $\rho \to 0$. We split the integral representation of $\mathcal{A}_1[\mathbf{G}^{11}_{\mathbf{H}_0}](\mathbf{x})$ as follows:
$$
\mathcal{A}_1[\mathbf{G}^{11}_{\mathbf{H}_0}](\mathbf{x})= \int_{-L/2}^{L/2} \boldsymbol{\Gamma}(\mathbf{x}-\mathbf{z}_y) (\mathbf{G}^{11}_{\mathbf{H}_0}(y) - \mathbf{G}^{11}_{\mathbf{H}_0}(x_0)) dy + \left(\int_{-L/2}^{L/2} \boldsymbol{\Gamma}(\mathbf{x}-\mathbf{z}_y) dy \right) \mathbf{G}^{11}_{\mathbf{H}_0}(x_0).
$$
Since $\mathbf{G}^{11}_{\mathbf{H}_0}(y)$ is derived from the smooth background field $\mathbf{H}_0$, it is at least continuous. Therefore, the difference $(\mathbf{G}^{11}_{\mathbf{H}_0}(y) - \mathbf{G}^{11}_{\mathbf{H}_0}(x_0)) \to \mathbf{0}$ as $y \to x_0$. This zero in the integrand is $\mathcal{O}(|y-x_0|)$, which is sufficient to cancel the singularity of the kernel $\boldsymbol{\Gamma}$. This cancellation ensures the integrand of the first term is bounded, and its integral over a finite domain converges to a finite value as $\rho \to 0$.
The asymptotic behavior of $\mathcal{A}_1[\mathbf{G}^{11}_{\mathbf{H}_0}](\mathbf{x})$ as $\rho \to 0$ is therefore dominated by the second term. Let $\mathbf{I}(\rho) = \int_{-L/2}^{L/2} \boldsymbol{\Gamma}(\mathbf{x}-\mathbf{z}_y) dy$. We analyze this integral by substituting $u = y-x_0$. Let $t_1=-L/2-x_0$ and $t_2=L/2-x_0$. Since $x_0 \in (-L/2, L/2)$, we have $t_1 < 0$ and $t_2 > 0$. The integral becomes:
$$
\mathbf{I}(\rho) = \int_{t_1}^{t_2} \left( \frac{\alpha_1 I_3}{\sqrt{u^2 + \rho^2}} + \frac{\alpha_2}{(u^2 + \rho^2)^{3/2}} \begin{pmatrix} u^2 & -u x_2 & -u x_3 \\ -u x_2 & x_2^2 & x_2 x_3 \\ -u x_3 & x_2 x_3 & x_3^2 \end{pmatrix} \right) du.
$$
We analyze the leading-order terms as $\rho \to 0$. The integral of the first part is
$$
\mathbf{I}_{\alpha_1} = \alpha_1 I_3 \int_{t_1}^{t_2} \frac{du}{\sqrt{u^2 + \rho^2}} = \alpha_1 I_3 \left[ \sinh^{-1}(u/\rho) \right]_{t_1}^{t_2}.
$$
Using the asymptotic behavior $\sinh^{-1}(s) \sim \ln(2s)$ as $s \to \infty$ and $\sinh^{-1}(s) \sim \ln(-2s)$ as $s \to -\infty$, we obtain
$$
\int_{t_1}^{t_2} \frac{du}{\sqrt{u^2 + \rho^2}} \sim \ln(2t_2/\rho) - \ln(-2t_1/\rho) = \ln(t_2/(-t_1)) - 2\ln(\rho).
$$
Thus, $\mathbf{I}_{\alpha_1} = \alpha_1 I_3 (-2\ln \rho + \mathcal{O}(1))$. A similar asymptotic analysis of the components of the $\mathbf{I}_{\alpha_2}$ integral yields the following behavior for the full matrix $\mathbf{I}(\rho)$:
$$
\mathbf{I}(\rho) = \mathbf{I}_{\alpha_1} + \mathbf{I}_{\alpha_2} \sim \alpha_1 \begin{pmatrix} -2\ln \rho & 0 & 0 \\ 0 & -2\ln \rho & 0 \\ 0 & 0 & -2\ln \rho \end{pmatrix} + \alpha_2 \begin{pmatrix} -2\ln \rho & \mathcal{O}(\rho) & \mathcal{O}(\rho) \\ \mathcal{O}(\rho) & \mathcal{O}(1) & \mathcal{O}(1) \\ \mathcal{O}(\rho) & \mathcal{O}(1) & \mathcal{O}(1) \end{pmatrix} + \mathcal{O}(1).
$$
This simplifies to:
$$
\mathbf{I}(\rho) \sim \begin{pmatrix} -2(\alpha_1+\alpha_2)\ln \rho & 0 & 0 \\ 0 & -2\alpha_1 \ln \rho & 0 \\ 0 & 0 & -2\alpha_1 \ln \rho \end{pmatrix} + \mathcal{O}(1).
$$
Therefore, the asymptotic behavior of $\mathbf{v}(\mathbf{x})$ as $\rho \to 0$ is
$$
\mathcal{A}_1[\mathbf{G}^{11}_{\mathbf{H}_0}](\mathbf{x}) \sim \mathbf{I}(\rho) \cdot \mathbf{G}^{11}_{\mathbf{H}_0}(x_0) + \mathcal{O}(1).
$$
Letting $\mathbf{G}^{11}_{\mathbf{H}_0}(x_0) = (G_1, G_2, G_3)^T$, we have
$$
\mathcal{A}_1[\mathbf{G}^{11}_{\mathbf{H}_0}](\mathbf{x})\sim \begin{pmatrix} -2(\alpha_1+\alpha_2)G_1 \ln \rho \\ -2\alpha_1 G_2 \ln \rho \\ -2\alpha_1 G_3 \ln \rho \end{pmatrix} + \mathcal{O}(1).
$$
Since $\alpha_1 \neq 0$ and $\alpha_1+\alpha_2 = -\frac{1}{4\pi\mu} \neq 0$, and we have assumed $\mathbf{G}^{11}_{\mathbf{H}_0}(x_0) \neq \mathbf{0}$, this implies that at least one component of $\mathcal{A}_1[\mathbf{G}^{11}_{\mathbf{H}_0}](\mathbf{x})$ exhibits a logarithmic singularity. As $\rho \to 0$ (i.e., $\mathbf{x} \to \mathbf{z}_{x_0}$), we have $\ln \rho \to -\infty$. Therefore, our singularity analysis shows that:
$$
\lim_{\mathbf{x} \to \mathbf{z}_{x_0}} \|\mathbf{B}_0^1(\mathbf{x})\| = \infty.
$$
However, this contradicts our previous conclusion $\mathbf{B}_0^1(\mathbf{x}) \equiv \mathbf{0}$ (in $\mathbb{R}^3\setminus\Gamma_{0}$), which was obtained via analytic continuation. This is because $\mathbf{B}_0^1(\mathbf{x}) \equiv \mathbf{0}$ necessarily requires:
$$
\lim_{\mathbf{x} \to \mathbf{z}_{x_0}} \mathbf{B}_0^1(\mathbf{x}) = \mathbf{0}.
$$
This contradiction proves that our initial assumption, $\mathbf{G}^{11}_{\mathbf{H}_0}(y)\not\equiv \mathbf{0}$ and $\mathbf{R}^{11}_{\mathbf{H}_0}(y)\mathbf{e}_i \equiv \mathbf{0}$, must be false.\\
(ii) Assume that $\mathbf{G}^{11}_{\mathbf{H}_0}(y)\equiv \mathbf{0}$ and $\mathbf{R}^{11}_{\mathbf{H}_0}(y)\mathbf{e}_i \not\equiv \mathbf{0}$ for $i=2, 3$ on $[-L/2, L/2]$. By this assumption, there must exist some point $x_0\in (-L/2, L/2)$ such that $\mathbf{R}^{11}_{\mathbf{H}_0}(x_0)\mathbf{e}_i\neq \mathbf{0}$. Let us analyze the behavior of $\mathcal{A}_2[\mathbf{R}^{11}_{\mathbf{H}_0}](\mathbf{x})$ as $\mathbf{x}$ approaches the point $\mathbf{z}_{x_0}=(x_0, 0, 0)^T$ on the line segment $\Gamma_0$. We define a local coordinate system by letting $\mathbf{x}=(x_0, x_2, x_3)^T$ and defining $\rho = \sqrt{x_2^2 + x_3^2}$, $x_2 = \rho \cos \theta$, $x_3 = \rho \sin \theta$. Similar to the analysis in (i) and by direct calculation, we can obtain the following result:
$$
\mathbf{B}_0^1(\mathbf{x}) = -\mathcal{A}_2[\mathbf{R}^{11}_{\mathbf{H}_0}](\mathbf{x}) \sim -\frac{1}{\rho} \mathbf{A}(\theta) + \mathcal{O}(\ln \rho),
$$
where
\begin{equation}
\mathbf{A}(\theta) = \begin{pmatrix} \displaystyle -\frac{1}{2\pi\mu} (R_{12}\cos\theta + R_{13}\sin\theta)\\[10pt] \displaystyle \frac{\lambda+\mu}{4\pi\mu(\lambda+2\mu)} \left( \begin{aligned} R_{22}\left[-\Lambda \cos\theta + 2\cos\theta\sin^2\theta\right] + R_{32}\left[\sin\theta - 2\cos^2\theta\sin\theta\right] \\ + R_{23}\left[-\Lambda \sin\theta - 2\sin\theta\cos^2\theta\right] + R_{33}\left[\cos\theta - 2\cos\theta\sin^2\theta\right] \end{aligned} \right) \\[20pt] \displaystyle \frac{\lambda+\mu}{4\pi\mu(\lambda+2\mu)} \left( \begin{aligned} R_{22}\left[\sin\theta - 2\cos^2\theta\sin\theta\right] + R_{32}\left[-\Lambda \cos\theta - 2\cos\theta\sin^2\theta\right] \\ + R_{23}\left[\cos\theta - 2\cos\theta\sin^2\theta\right] + R_{33}\left[-\Lambda \sin\theta + 2\sin\theta\cos^2\theta\right] \end{aligned} \right) \end{pmatrix},
\end{equation}
where $R_{kj}$ denotes the $(k,j)$-th component of the matrix $\mathbf{R}^{11}_{\mathbf{H}_0}(x_0)$ and $\Lambda = \frac{\lambda+3\mu}{\lambda+\mu}$.
The components of the vector $\mathbf{A}(\theta)$ are linear combinations of the basis functions $ \cos\theta$, $\sin\theta$, $\cos 3\theta$, $\sin 3\theta$. These functions are linearly independent on $[0, 2\pi)$. Through a detailed algebraic verification, it can be shown that the coefficients of these basis functions vanish simultaneously if and only if all source components $R_{kj}$ (for $k=1, 2, 3$ and $j=2,3$) are zero. Since we assumed $\mathbf{R}^{11}_{\mathbf{H}_0}(x_0)$ has non-zero entries in the second or third columns, the vector $\mathbf{A}(\theta)$ does not vanish identically for all $\theta$. Consequently, there exists at least one direction $\theta_0$ such that $\|\mathbf{A}(\theta_0)\| \neq 0$.
Along this direction, the field behaves as:
$$
\lim_{\rho \to 0} \|\mathbf{B}_0^1(\mathbf{x})\| = \lim_{\rho \to 0} \frac{1}{\rho} \|\mathbf{A}(\theta_0)\| = \infty.
$$
This contradicts our previous conclusion $\mathbf{B}_0^1(\mathbf{x}) \equiv \mathbf{0}$ (in $\mathbb{R}^3\setminus\Gamma_{0}$).
This contradiction proves that our initial assumption, $\mathbf{G}^{11}_{\mathbf{H}_0}(y)\equiv \mathbf{0}$ and $\mathbf{R}^{11}_{\mathbf{H}_0}(y)\mathbf{e}_i \not\equiv \mathbf{0}$, must be false.\\
(iii)Assume that $\mathbf{G}^{11}_{\mathbf{H}_0}(y)\not\equiv \mathbf{0}$ and $\mathbf{R}^{11}_{\mathbf{H}_0}(y)\mathbf{e}_i \not\equiv \mathbf{0}$ for $i=2, 3$ on $[-L/2, L/2]$. Based on the analysis in (i) and (ii), the field is a superposition of terms with different singularity orders:
$$
\mathbf{B}_0^1(\mathbf{x}) \sim \mathcal{O}(\ln \rho) - \frac{1}{\rho} \mathbf{A}(\theta).
$$
Since the algebraic singularity $\frac{1}{\rho}$ dominates the logarithmic singularity $\ln \rho$ as $\rho \to 0$, the asymptotic behavior is dictated by the $\mathcal{A}_2$ term. As proven in (ii), $\mathbf{A}(\theta)$ is not identically zero. Thus, along a direction $\theta_0$ where $\mathbf{A}(\theta_0) \neq \mathbf{0}$, the field blows up:
$$
\lim_{\rho \to 0} \|\mathbf{B}_0^1(\mathbf{x})\| = \infty.
$$
This contradicts $\mathbf{B}_0^1(\mathbf{x}) \equiv \mathbf{0}$ (in $\mathbb{R}^3\setminus\Gamma_{0}$). \\
Therefore, the assumptions in all three cases are false.
We conclude that $\mathbf{G}^{11}_{\mathbf{H}_0}(y)\equiv \mathbf{0}$ and $\mathbf{R}^{11}_{\mathbf{H}_0}(y)\mathbf{e}_i \equiv \mathbf{0}$ (for $i=2, 3$) must hold for all $y \in [-L/2, L/2]$.
This completes the proof.
\end{proof}

\end{document}